\documentclass[aps,pra,twocolumn,groupedaddress,showpacs,superscriptaddress,scrartcl,longbibliography]{revtex4-2}
\usepackage{array,xcolor,enumitem,tabularx,multirow,hyperref,comment,amssymb,amsmath,graphicx,epstopdf,amsthm,dsfont,enumerate,bm,color}
\newtheorem{theorem}{Theorem}
\newtheorem{lemma}{Lemma}
\newtheorem{proposition}{Proposition}

\newtheorem{definition}{Definition}

\newtheorem*{lemma1*}{Lemma 1}
\newtheorem*{lemma2*}{Lemma 2}
\newtheorem*{lemma3*}{Lemma 3}
\newtheorem*{lemma4*}{Lemma 4}
\newtheorem*{lemma5*}{Lemma 5}
\newtheorem*{lemma8*}{Lemma 8}
\newtheorem*{lemma9*}{Lemma 9}
\newtheorem*{lemma10*}{Lemma 10}
\newtheorem*{lemma11*}{Lemma 11}
\newtheorem*{theorem1*}{Theorem 1}

\newtheorem*{theorem2*}{Theorem 2}
\newtheorem*{theorem3*}{Theorem 3}

\newcommand{\damian}[1]{{\color{black}{#1}}}			
\newcommand{\damiannn}[1]{{\color{black}{#1}}}	
\newcommand{\damiann}[1]{{\color{black}{#1}}}

\newcommand{\dpg}[1]{{\color{violet}{#1}}}

\begin{document}
\interfootnotelinepenalty=10000

\title{Practical quantum tokens without quantum memories and experimental tests}

\author{Adrian Kent}
\affiliation{Centre for Quantum Information and Foundations, DAMTP, Centre for Mathematical Sciences, University of Cambridge, Wilberforce Road, Cambridge, CB3 0WA, U.K.}
\affiliation{Perimeter Institute for Theoretical Physics, 31 Caroline Street North, Waterloo, ON N2L 2Y5, Canada}

\author{David Lowndes}
\affiliation{Quantum Engineering Technology Labs, H. H. Wills Physics Laboratory and Department of Electrical and Electronic Engineering, University of Bristol, Bristol, U.K.
}

\author{Dami\'an Pital\'ua-Garc\'ia}
\email{D.Pitalua-Garcia@damtp.cam.ac.uk}
\affiliation{Centre for Quantum Information and Foundations, DAMTP, Centre for Mathematical Sciences, University of Cambridge, Wilberforce Road, Cambridge, CB3 0WA, U.K.}

\author{John Rarity}
\affiliation{Quantum Engineering Technology Labs, H. H. Wills Physics Laboratory and Department of Electrical and Electronic Engineering, University of Bristol, Bristol, U.K.
}

\date{\today}

\begin{abstract}

Unforgeable quantum money tokens were the first invention of quantum information science, but remain technologically challenging as they require quantum memories and/or long distance quantum communication. More recently, virtual  `S-money' tokens were introduced.   These are generated by quantum cryptography, do not require quantum memories or long distance quantum communication, and yet in principle
guarantee many of the security advantages of quantum money. Here, we describe implementations of S-money schemes with off-the-shelf quantum key distribution technology, and analyse security in the presence of noise, losses, and experimental imperfection. Our schemes satisfy near instant validation without cross-checking. We show that, given standard assumptions in mistrustful quantum cryptographic implementations, unforgeability and user privacy could be guaranteed with attainable refinements of our off-the-shelf setup.
We discuss the possibilities for unconditionally secure  (assumption-free) implementations.
\end{abstract}

\maketitle

\section{Introduction}


Quantum tokens, also called quantum money, were invented by Wiesner
\cite{wiesner1983conjugate} in 1970. In Wiesner's original
  quantum token scheme Bob (the bank) secretly and securely
generates a classical serial number $s$ and a quantum state
$\lvert\psi\rangle$ of $N$ qubits, prepared from a set of different
bases, gives $s$ and $\lvert\psi\rangle$ to Alice, and stores $s$ and
the classical description of $\lvert\psi\rangle$ in a database. Alice presents
the token by giving $s$ and $\lvert\psi\rangle$ back to Bob, and Bob
validates or rejects the token after measuring the received quantum
state in the basis in which $\lvert\psi\rangle$ was prepared. In
refinements of this scheme \cite{gavinsky2012quantum,MVW12,PYLLC12,GK15,MP16,AA17,BDG19,K19,HS20}, Alice can present the token to Bob
or to one of a set of verifiers, by communicating the classical
outcomes of quantum measurements applied on $\lvert\psi\rangle$, as
requested by Bob  or the verifier. Alternatively, Alice presents
the token by giving $s$ and $\lvert\psi\rangle$ to the verifier, who applies quantum measurements on $\lvert\psi\rangle$. The verifier communicates with Bob to validate or reject the token.

There exist quantum token schemes satisfying \emph{unforgeability}, i.e. they guarantee that a token cannot be validated more than once, with \emph{unconditional security}, i.e. based only on the laws of physics without restricting the technology of dishonest Alice \cite{gavinsky2012quantum,MVW12,PYLLC12,GK15,MP16,AA17,BDG19,K19,HS20}. Intuitively, this follows from the no-cloning theorem, stating that it is impossible to perfectly copy unknown quantum states \cite{WZ82,D82}. Unforgeable quantum token schemes based on computational assumptions have also been investigated (e.g. \cite{BBBW83,MS10,AC12,FGHLS12}), with some of these schemes not requiring communication with the bank for token validation (e.g \cite{AC12,FGHLS12}).

However, there exist purely classical token schemes that can also guarantee unforgeability with unconditional security. For example, the token may comprise a classical serial number $s$ and a classical secret password $x$ that Bob gives Alice and that Alice presents by giving to one of a set of verifiers; validation of the token comprises \emph{cross-checking}; for example, the verifier communicates with Bob and validates the token if this has not been presented before and if the given serial number and password correspond to each other. 

In addition to unforgeability, some important properties of quantum token schemes are the following. First, quantum tokens can be transferred while keeping Bob's database static. On the other hand, since classical information can be copied perfectly, in order to satisfy unforgebaility, when a purely classical token with serial number $s$ is transferred from Alice to another party Charlie, Bob must change the classical data associated to $s$; for example, Bob must change $x$ to another value $x'$ and give $s$ and $x'$ to Charlie in the example above \cite{gavinsky2012quantum}.

Second, some quantum token schemes satisfy \emph{instant
  validation}. This means that the schemes do not require
communication between the verifiers and Bob for validation after Alice
presents the token \cite{PYLLC12}. This implies in
particular that the token can be presented by Alice at one
of a set of different spacetime points that can be spacelike separated
without validation delays by the verifier due to cross-checking with
Bob and/or with other verifiers.

Third, quantum token schemes satisfy \emph{future privacy for the user}, or simply \emph{user privacy}. That is, neither Bob, nor the verifiers, can know where and when Alice will present the token. 

It is not difficult to construct purely classical token schemes that
satisfy with unconditional security any two of
  unforgeability, instant validation and user privacy. For example,
the classical token scheme above satisfies unforgeability and user
privacy with unconditional security, but not instant
validation. To the best of our knowledge no purely classical token scheme
has been shown to satisfy all three properties
simultaneously with unconditional security. Classical variations of the quantum token schemes we consider here, based on classical relativistic bit commitments, whose security is hypothesized but not proven, were proposed in Ref. \cite{KSmoney}, which considers their potential advantages and disadvantages.   As far as we are aware, aside from these, there are no known classical schemes that plausibly satisfy all three properties simultaneously with unconditional security.

Among plausible future applications of quantum token schemes
are very high value and time critical transactions requiring very high
security, such as financial trading, where many transactions take place
within half a millisecond \cite{WF10}, or network control, where
semi-autonomous teams need authentication as fast as possible. 
A reasonable assumption for such applications is
that tokens may be transferred a relatively small number of times
among a relatively small set of parties -- the tokens may be valid for
a relatively short time, for
example. In this context, Bob having a static database does not seem
to be a great advantage of quantum token schemes over classical
schemes whose databases must be updated after each transaction,
given that processing classical information is much easier and cheaper
than processing quantum information.  Furthermore, for very
high value transactions one might expect that the communication network among
Bob and the verifiers is sufficiently protected that communication
among them is very rarely (if ever) interrupted.
So, in this context, it appears to be a major advantage of quantum token schemes over classical token
schemes that a quantum token can be presented at one of a set of
spacelike separated points with near instant validation without time
delays due to cross-checking, while satisfying unforgeability and user
privacy with unconditional security.

Standard quantum token schemes satisfying unforgeability, user privacy and instant validation with unconditional security require to store quantum states in quantum memories and/or to transfer quantum states over long distances in order to give Alice enough flexibility in space and time to present the token \cite{wiesner1983conjugate,BBBW83,MS10,gavinsky2012quantum,AC12,FGHLS12,MVW12,PYLLC12,GK15,MP16,AA17,BDG19,K19,HS20}. Recently, a quantum memory of a single qubit with a coherence time of over an hour has been experimentally demonstrated \cite{WUZALZDYK17etal,WLQUZWYGZK21etal}. However, storing large quantum states for more than a fraction of a second remains challenging \cite{WLZSZLDYZ19etal,WMHSG20}. Furthermore, the transmission of quantum states over long distances in practice comprises the transmission of photons through optical fibre or through the atmosphere via satellites. In both cases a great fraction of the transmitted photons is lost. For these reasons, standard quantum token schemes are impractical for most purposes at present.

Recently, experimental investigations of quantum token schemes have been performed \cite{BCGLMN17,BBP17,BOVZKD18,GAAZLYWZP18etal,JBCL19}. Refs. \cite{BCGLMN17,JBCL19} investigated the experimental implementation of forging attacks on quantum token schemes. Ref. \cite{BBP17} presented a simulation of a quantum token scheme in IBM's five-qubit quantum computer. Refs. \cite{BOVZKD18,GAAZLYWZP18etal} reported proof-of-principle experimental demonstrations of the preparation and verification stages of quantum token schemes, by transmitting quantum states encoded in photons over a short distance -- for example, Ref. \cite{GAAZLYWZP18etal} reports optical fibre lengths of up to 10 meters. A full experimental demonstration of a quantum token scheme that includes storing quantum states in a quantum memory and/or transmitting quantum states over long distances remains an important open problem.

`S-money' \cite{KSmoney} is a 
class of quantum token schemes,
which is designed for the settings described above comprising networks
with relativistic or other trusted signalling constraints. These
schemes can guarantee many of the the security advantages of standard
quantum token schemes -- in particular, instant validation,
unforgeability and user privacy -- without requiring either
quantum state storage or long distance transmission of quantum states. Furthermore, S-money tokens that can be transferred among
several parties and that give the users a great flexibility in space
and time to present the token are also possible \cite{KPG20}.
In this paper, we begin to investigate how securely S-money schemes can be implemented in practice with current technology.

Our results are twofold. First, we introduce  quantum token schemes
that extend the quantum S-money scheme of
Ref. \cite{quantumtokenspat}  in practical experimental scenarios that consider losses, errors in the state preparations and
measurements, and deviations from random distributions; and, in photonic setups, photon sources that do not emit exactly single
photons, and single photon detectors with non-unit detection
efficiencies and with non-zero dark count probabilities, which are threshold detectors, i.e. which cannot distinguish the number of photons in detected pulses.
In our schemes, Alice can present the token at one of $2^M$
possible spacetime presentation points, which can have arbitrary
timelike or spacelike separation, for any positive integer $M$. Our
schemes satisfy instant validation and comprise Bob transmitting $N$
quantum states to Alice over a distance which can be arbitrarily
short, Alice measuring the received quantum states without storing
them, and further classical processing and classical communication
over distances which can be arbitrarily large. Thus, our schemes are advantageous over standard quantum token schemes because they do not need
quantum state storage or transmission of quantum states over long
distances. We use the flexible versions of S-money defined in
Ref. \cite{KPG20}, giving Alice the freedom to choose her spacetime
presentation point after having performed the quantum measurements. We show that our schemes satisfy unforgeability and user privacy,
given assumptions that have been standard in
implementations of mistrustful quantum cryptography to date (see Table \ref{tableassu}), but are
nonetheless idealizations.

Second, we performed experimental tests of the quantum stage of one of our schemes for the case of two presentation points, which show that with refinements of our setup our schemes can be
implemented securely, giving guarantees of unforgeability and user privacy, based on the standard assumptions in experimental mistrustful quantum cryptography mentioned above.

\section{Results}

\subsection{Preliminaries and Notation}

\label{section2}

We present below two quantum token schemes that do not require quantum state storage, are practical to implement with current technology, and allow for experimental imperfections. We show that for a range of experimental parameters our token schemes are secure.

In the token schemes below, Bob (the bank) and Alice (the acquirer) agree on spacetime regions $Q_i$ where a token can be presented by Alice to Bob, for $i\in\{0,1\}^M$ and for some agreed integer $M\geq 1$. Bob has trusted agents $\mathcal{B}$ and $\mathcal{B}_i$ controlling secure laboratories, and Alice has trusted agents $\mathcal{A}$ and $\mathcal{A}_i$ controlling secure laboratories, for $i\in\{0,1\}^M$. The agent $\mathcal{A}_i$ can send messages to $\mathcal{B}_i$ in the spacetime region $Q_i$, for $i\in\{0,1\}^M$. All communications among agents of the same party are performed via secure and authenticated classical channels, which can be implemented with previously distributed secret keys. Alice's agent $\mathcal{A}$ and Bob's agent $\mathcal{B}$ perform the specified actions in a spacetime region $P$ that lies within the intersection of the causal pasts of all $Q_i$, unless otherwise stated.

The token schemes comprise two main
stages. Stage I includes the quantum communication between
$\mathcal{B}$ and $\mathcal{A}$, which can take place between adjacent
laboratories, 
and must
be implemented within the intersection of the causal pasts of all the
presentation points.
In particular, this stage can take an arbitrarily long time and can be completed arbitrarily in the past of the presentation points, which is very helpful for practical implementations. Stage II comprises only classical
processing and classical communication among agents of Bob and Alice,
and must be implemented very fast in order to satisfy some
relativistic constraints. A token received by  $\mathcal{B}_b$ from $\mathcal{A}_b$ at $Q_b$ can be validated by $\mathcal{B}_b$ near-instantly at $Q_b$, without the need to cross check with other agents. We note that Alice chooses her presentation point in stage II, meaning in particular that it can take place after her quantum measurements have been completed. This is basically the application of  the refinement of flexible S-money tokens discussed in Ref. \cite{KPG20}, which gives Alice great flexibility in spacetime to choose her presentation point. See Tables \ref{ideal1} -- \ref{real2} for details.

In stage I, $\mathcal{B}$ generates quantum states randomly from a predetermined set and gives these to $\mathcal{A}$. $\mathcal{A}$ measures the received states in bases from a predetermined set. $\mathcal{A}$ sends some classical messages to $\mathcal{B}$, mainly to indicate the set of states that she successfully measured. For all $i\in\{0,1\}^M$, $\mathcal{A}$ communicates her classical outcomes to $\mathcal{A}_i$; $\mathcal{B}$ sends classical messages to $\mathcal{B}_i$, indicating mainly the labels of the states reported by $\mathcal{A}$ to be successfully measured.

In stage II, Alice chooses the label $b\in\{0,1\}^M$ of her chosen presentation point in the intersection of the causal pasts of the presentation points. Further classical communication steps among agents of Alice and Bob take place. The token schemes conclude by Alice giving a classical message $\mathbf{x}$ (the token) to Bob at her chosen presentation point $Q_b$ and Bob validating the token at $Q_b$ if $\mathbf{x}$ satisfies a mathematical condition.

The main difference between the first and second token schemes below (either in their idealized or
  realistic version) is that, in the first one, Alice measures each
received qubit randomly in one of two predetermined bases, while in
the second one Alice measures large sets of qubits in the
same basis, which is chosen randomly by Alice from two predetermined
bases. The first token scheme is more suitable to implement
with setups used for quantum key distribution. The second token scheme requires a slightly different setup.

We say a token scheme satisfies instant validation if, for any presentation point $Q_i$, an agent of Bob receiving a token from Alice at $Q_i$ can validate or reject the token nearly instantly at $Q_i$, without the need to wait for any messages from other agents at spacetime points spacelike separated from $Q_i$.

We say a token scheme is:
\begin{itemize}

\item $\epsilon_{\text{rob}}-$robust if the probability that Bob aborts when Alice and Bob follow the token scheme honestly is not greater than $\epsilon_{\text{rob}}$, for any $b\in\{0,1\}^M$;

\item $\epsilon_{\text{cor}}-$correct if the probability that Bob does not accept Alice's token as valid when Alice and Bob follow the token scheme honestly is not greater than $\epsilon_{\text{cor}}$, for any $b\in\{0,1\}^M$;

\item $\epsilon_{\text{priv}}-$private if the probability that Bob guesses Alice's bit-string $b$ before she presents her token to Bob is not greater than $\frac{1}{2^M}+\epsilon_{\text{priv}}$, if Alice follows the token scheme honestly, for $b\in\{0,1\}^M$ chosen randomly from a uniform distribution by Alice;

\item $\epsilon_{\text{unf}}-$unforgeable, if the probability that Bob accepts Alice's tokens as valid at any two or more different presentation points is not greater than $\epsilon_{\text{unf}}$, if Bob follows the token scheme honestly.

\end{itemize}

We say a token scheme using $N$ transmitted quantum states is:

\begin{itemize}

\item  \emph{robust} if it is $\epsilon_{\text{rob}}-$robust with $\epsilon_{\text{rob}}$ decreasing exponentially with $N$.

\item \emph{correct} if it is $\epsilon_{\text{cor}}-$correct with $\epsilon_{\text{cor}}$ decreasing exponentially with $N$.

\item \emph{private} if it is $\epsilon_{\text{priv}}-$private with $\epsilon_{\text{priv}}$ approaching zero by increasing some security parameter.

\item \emph{unforgeable} if it is $\epsilon_{\text{unf}}-$unforgeable with $\epsilon_{\text{unf}}$ decreasing exponentially with $N$.
\end{itemize}

Note that our definition of privacy is different because it depends on different parameters: see Lemma \ref{Bob1} below.  In our schemes each of the $N$ quantum states is a qubit state with probability $1-P_\text{noqub}$, and a quantum state of arbitrary Hilbert space dimension greater than two with probability $P_\text{noqub}$, where $P_\text{noqub}=0$ in ideal schemes and $P_\text{noqub}>0$ in practical schemes. In photonic implementations, each pulse transmitted by Bob is either vacuum or one-photon with probability $1-P_\text{noqub}$, and multi-photon with probability $P_\text{noqub}$.

Below we present token schemes for two presentation points ($M=1$) that satisfy instant validation and that are robust, correct, private and unforgeable. The extension to $2^M$ presentation points for any $M\in\mathbb{N}$ is given in Appendix \ref{manyapplast}. For clarity of the presentation we first present the ideal quantum token schemes $\mathcal{IQT}_1$ and $\mathcal{IQT}_2$ where there are not any losses, errors, or any other experimental imperfections. These are given in Table \ref{ideal1}. More realistic quantum token schemes $\mathcal{QT}_1$ and $\mathcal{QT}_2$ that allow for various experimental imperfections are presented in Tables \ref{real1} and \ref{real2}, respectively. An illustration of implementation in a token scheme for the case of two spacelike separated presentation points is given in Fig. \ref{fig0}.

We use the following notation. We use bold font notation $\mathbf{a}$ for strings of bits. The bitwise complement of a string $\mathbf{a}$ is denoted by $\bar{\mathbf{a}}$. The $k$th bit entry of a string $\mathbf{a}$ is denoted by $a_k$. We define the set $[N]=\{1,2,\ldots,N\}$. The symbol `$\oplus$' denotes bit-wise sum modulo 2 or sum modulo 2 depending on the context.  We write the Bennett-Brassard 1984 (BB84) states \cite{BB84} as $\lvert \phi_{00}\rangle=\lvert 0\rangle$, $\lvert \phi_{10}\rangle=\lvert 1\rangle$, $\lvert \phi_{01}\rangle=\lvert +\rangle$ and $\lvert \phi_{11}\rangle=\lvert -\rangle$, where $\lvert\pm\rangle=\frac{1}{\sqrt{2}}\bigl(\lvert 0\rangle\pm\lvert1\rangle\bigr)$, and where $\mathcal{D}_{0}=\{\lvert 0\rangle,\lvert 1\rangle\}$ and $\mathcal{D}_{1}=\{\lvert +\rangle,\lvert -\rangle\}$ are qubit orthonormal bases, called the computational and Hadamard bases, respectively. The Hamming distance is denoted by $d(\cdot,\cdot)$.

\begin{table*}
\begin{center}
\begin{tabular}{| p{15cm} | }
\hline
\\
\multicolumn{1}{|c|}{\textbf{Ideal quantum token scheme $\mathcal{IQT}_1$}}\\
\\
\hline
\\
\multicolumn{1}{|c|}{Stage I}\\

\vspace{0.4mm}

1. For $k\in[N]$, $\mathcal{B}$ generates the qubit state $\lvert\psi_k\rangle=\lvert\phi_{t_ku_k}\rangle$ randomly from the BB84 set and sends it to  $\mathcal{A}$ with its label $k$. Let the $N-$bit strings $\mathbf{t}=(t_1,\ldots,t_N)$ and $\mathbf{u}=(u_1,\ldots,u_N)$ denote the states and bases of preparation by $\mathcal{B}$.\\

\vspace{0.4mm}

2. For $k\in[N]$, $\mathcal{A}$ measures each received qubit randomly in the computational basis ($y_k=0$) or in the Hadamard basis ($y_k=1$) and obtains a string of $N$ bit outcomes $\mathbf{x}$. Let the $N$-bit string $\mathbf{y}=(y_1,\ldots,y_N)$ denote Alice's measurement bases.\\

\vspace{0.4mm}

3. $\mathcal{A}$ sends $\mathbf{x}$ to $\mathcal{A}_i$, for $i\in\{0,1\}$. \\

\vspace{0.4mm}

4. $\mathcal{A}$ chooses a bit $z$ randomly and gives $\mathcal{B}$ a string $\mathbf{d}$, where $\mathbf{d}=\mathbf{y}$ if $z=0$, or $\mathbf{d}=\bar{\mathbf{y}}$ if $z=1$.\\

\vspace{0.4mm}

5. For $i\in\{0,1\}$, $\mathcal{B}$ sends $\mathbf{d}$ to $\mathcal{B}_i$, who computes $\mathbf{d}_i$ in the causal past of $Q_i$, where $\mathbf{d}_0=\mathbf{d}$ and $\mathbf{d}_1=\bar{\mathbf{d}}$.\\

\vspace{0.4mm}

6. $\mathcal{B}$ sends $\mathbf{t}$ and $\mathbf{u}$ to $\mathcal{B}_i$, for $i\in\{0,1\}$.\\
\\

\hline
\\

\multicolumn{1}{|c|}{Stage II}\\

\vspace{0.4mm}

7. $\mathcal{A}$ chooses the presentation point $Q_b$ for the token, for some $b\in\{0,1\}$. $\mathcal{A}$ computes the bit $c=b\oplus z$ and sends it to $\mathcal{B}$.\\

\vspace{0.4mm}

8. $\mathcal{B}$ sends $c$ to $\mathcal{B}_i$, for $i\in\{0,1\}$.\\

\vspace{0.4mm}

9. For $i\in\{0,1\}$, in the causal past of $Q_i$, $\mathcal{B}_i$ computes the string $\tilde{\mathbf{d}}_i=\mathbf{d}_i$ if $c=0$, or $\tilde{\mathbf{d}}_i=\bar{\mathbf{d}_i}$ if $c=1$.\\

\vspace{0.4mm}

10. $\mathcal{A}$ sends a signal to $\mathcal{A}_b$ indicating to present the token at $Q_b$, and $\mathcal{A}_b$ presents the token $\mathbf{x}$ to $\mathcal{B}_b$ in $Q_b$.\\

\vspace{0.4mm}

11. $\mathcal{B}_b$ validates the token $\mathbf{x}$ received in $Q_b$ if $\mathbf{x}_b=\mathbf{t}_b$, where $\mathbf{a}_v$ is the restriction of a string $\mathbf{a}\in\{\mathbf{x},\mathbf{t}\}$ to entries $a_k$ with $k\in\Delta_v$, where $\Delta_v=\{k\in [N]\vert \tilde{d}_{v,k}=u_k\}$, and where $\tilde{d}_{v,k}$ is the $k$th bit entry of the string $\tilde{\mathbf{d}}_{v}$, for $k\in[N]$ and for $v\in\{0,1\}$. \damian{That is, Bob validates the token if Alice reports the correct measurement outcome for each qubit that she measured in Bob's preparation basis.}\\
\\
\botrule
\\
\multicolumn{1}{|c|}{\textbf{Ideal quantum token scheme $\mathcal{IQT}_2$}}\\
\\
\hline
\\
\multicolumn{1}{|c|}{Stage I}\\

\vspace{0.4mm}

1. As step 1 of $\mathcal{IQT}_1$.\\

\vspace{0.4mm}

2. The step 2 of $\mathcal{IQT}_1$ is replaced by the following. $\mathcal{A}$ chooses a bit $z$ randomly. $\mathcal{A}$ measures each received qubit in the computational basis if $z=0$ or in the Hadamard basis if $z=1$. The string $\mathbf{y}\in\{0,1\}^{N}$ denoting Alice's measurement bases has bit entries $y_k=z$ for $k\in[N]$.\\

\vspace{0.4mm}

3. As step 3 of $\mathcal{IQT}_1$. The steps 4 and 5 of $\mathcal{IQT}_1$ are discarded.\\

\vspace{0.4mm}

4. As step 6 of $\mathcal{IQT}_1$.\\
\\

\hline
\\

\multicolumn{1}{|c|}{Stage II}\\

\vspace{0.4mm}

5. As steps 7 and 8 of $\mathcal{IQT}_1$.\\

\vspace{0.4mm}

6. The step 9 of $\mathcal{IQT}_1$ is replaced by the following. For $i\in\{0,1\}$, in the causal past of $Q_i$, $\mathcal{B}_i$ computes the string $\tilde{\mathbf{d}}_i\in\{0,1\}^{N}$ with bit entries $\tilde{d}_{i,k}= i\oplus c$, for $k\in[N]$.\\

\vspace{0.4mm}

7. As steps 10 and 11 of $\mathcal{IQT}_1$.\\
\\

\botrule
\end{tabular}
\caption{Ideal quantum token schemes $\mathcal{IQT}_1$ and $\mathcal{IQT}_2$ for two presentation points. Steps 1 to 8 in $\mathcal{IQT}_1$, and 1 to 5 in $\mathcal{IQT}_2$, take place within the intersection of the causal pasts of the presentation points.}
\label{ideal1}
\end{center}
\end{table*}

\begin{figure}
\includegraphics{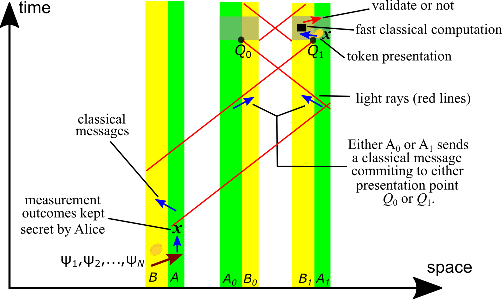}
\caption{\label{fig0} 
\textbf{Illustration of implementation in a quantum token scheme.} A  case of two presentation points in a Minkowski spacetime diagram in $1+1$ dimensions is illustrated. Bob has laboratories $B$, $B_0$ and $B_1$, controlled by agents $\mathcal{B}$, $\mathcal{B}_0$ and $\mathcal{B}_1$ (yellow rectangles), and Alice has laboratories $A$, $A_0$ and $A_1$, controlled by agents $\mathcal{A}$, $\mathcal{A}_0$ and $\mathcal{A}_1$ (green rectangles), adjacent to Bob's laboratories. The quantum communication stage takes place within $B$ and $A$, can take an arbitrarily long time and can be completed arbitrarily in the past of the presentation points ($Q_0$ and $Q_1$). Alice's classical measurement outcomes $\mathbf{x}$ are kept secret by Alice and communicated to her laboratories $A_0$ and $A_1$ via secure and authenticated classical channels. In this illustrated example, Alice sends classical messages to Bob at the laboratory $B$, and either at $B_0$ or $B_1$. The messages sent to $B$ can take place anywhere in the past of $Q_0$ and $Q_1$ after the quantum communication stage, and includes a message indicating the labels of the quantum states successfully measured by Alice. These messages are communicated from $B$ to $B_0$ and $B_1$ via secure and authenticated classical channels. Alice chooses to present her token at $Q_b$ within the intersection of the causal pasts of $Q_0$ and $Q_1$. The message at either $B_0$ or $B_1$ is the bit $c=b\oplus z$, effectively committing Alice to present her token at $Q_b$.
Alice presents the token by giving Bob
$\mathbf{x}$ at $Q_b$. The case $b=1$ is illustrated. The small black box represents a fast classical computation performed at Bob's laboratory receiving the token, to validate or reject Alice's token, as described in step 12 of the scheme $\mathcal{QT}_1$ (see Table \ref{real1}), for instance. As illustrated, this would require this computation to be completed within a time shorter than the time that light takes to travel between the locations of laboratories $B_0$ and $B_1$, which could be of 10 $\mu$s if $B_0$ and $B_1$ are separated by 3 km, for example.  \damian{This is because, as discussed in the introduction, we require presentation and acceptance to be completed within spacelike separated regions in order to achieve an advantage over purely classical token schemes.}}
\end{figure}

The quantum token schemes $\mathcal{IQT}_1$ and $\mathcal{IQT}_2$ given in Table \ref{ideal1} have the following properties.

First, the token schemes are correct. Since we assume there are not any errors in the state preparations and measurements, if Alice and Bob follow the token scheme honestly then Bob validates Alice's token at her chosen presentation point $Q_b$ with unit probability. If Alice and Bob follow $\mathcal{IQT}_1$ honestly, $\tilde{d}_{b,k}=d_{b,k}\oplus c=d_k\oplus b\oplus c=y_k\oplus z\oplus b\oplus c=y_k$, for $k\in[N]$. Thus, $\tilde{\mathbf{d}}_{b}=\mathbf{y}$, which means that $y_k=u_k$ for all $k\in\Delta_b$, hence, Alice measures in the same basis of preparation by Bob for all states $\lvert\psi_k\rangle$ with labels $k\in\Delta_b$. Therefore, Alice obtains the correct outcomes for these states: $\mathbf{x}_b=\mathbf{t}_b$.  Similarly,  if Alice and Bob follow $\mathcal{IQT}_2$ honestly then we have that $\tilde{\mathbf{d}}_b$ has bit entries $\tilde{d}_{b,k}=b\oplus c=z=y_k$, for $k\in[N]$. Thus, as above, $\tilde{\mathbf{d}}_b=\mathbf{y}$, i.e. $\tilde{\mathbf{d}}_b$ corresponds to the string of measurement basis implemented by Alice. Therefore, in both token schemes $\mathcal{IQT}_1$ and $\mathcal{IQT}_2$, Alice obtains $\mathbf{x}_b=\mathbf{t}_b$ and Bob validates Alice's token at $Q_b$ with unit probability.

Second, the token schemes are robust. More precisely, neither Bob nor Alice have the possibility to abort. This is because we assume there are not any losses of the transmitted quantum states and that Alice successfully measures all the received quantum states. Thus, Alice does not need to report to Bob any labels of states that she successfully measured, in contrast to the extended token schemes $\mathcal{QT}_1$ and $\mathcal{QT}_2$ discussed below.

Third, the token schemes are private, i.e. Bob cannot obtain any information about $b$ in the causal past of $Q_b$. This is because the messages Alice sends Bob in the causal past of $Q_b$ carry no information about $b$ and we assume that Alice's laboratories and communication channels are secure.

Fourth, the token schemes are unforgeable. This follows from the following lemma, which is shown in Appendix \ref{idealunfapp}. Alternative proofs are given in Ref. \cite{CK12}\damian{, based on quantum state discrimination tasks. We have chosen the proof given in Appendix \ref{idealunfapp} because an extension of it allows us to to prove Theorem \ref{Alice1} too.}

\begin{lemma}
	\label{lemma01}
	The quantum token schemes $\mathcal{IQT}_1$ and $\mathcal{IQT}_2$ are $\epsilon_{\text{unf}}-$unforgeable with
	\begin{equation}
	\label{000}
	\epsilon_{\text{unf}}=\Bigl(\frac{1}{2}+\frac{1}{2\sqrt{2}}\Bigr)^N.
\end{equation}	
\end{lemma}

Fifth, the token schemes satisfy instant validation. We note from step 11 of $\mathcal{IQT}_1$ that a token received by Bob's agent $\mathcal{B}_b$ from Alice's agent $\mathcal{A}_b$ at a presentation point $Q_b$ can be validated by $\mathcal{B}_b$ near-instantly at $Q_b$. In particular, $\mathcal{B}_b$ does not need to wait for any signals coming from other agents  of Bob.

Finally, the token schemes above can be modified in various ways. For example, in $\mathcal{IQT}_1$, step 3 can be discarded, and step 10 can be replaced by the following: after choosing $b$, $\mathcal{A}$ sends $\mathbf{x}$ to $\mathcal{A}_b$ and $\mathcal{A}_b$ presents the token $\mathbf{x}$ to $\mathcal{B}_b$ in $Q_b$. In another variation, step 5 in $\mathcal{IQT}_1$ can be modified so that $\mathcal{B}$ computes $\mathbf{d}_i$ and sends it to $\mathcal{B}_i$; in both versions of step 5, $\mathcal{B}_i$ must have $\mathbf{d}_i$ in the causal past of $Q_i$, for $i\in\{0,1\}$. In another variation, the step 9 in $\mathcal{IQT}_1$ is performed only by Bob's agent $\mathcal{B}_b$ receiving a token from Alice. The  version we have chosen for step 9 allows $\mathcal{B}_b$ to reduce the computation time after receiving a token, hence, allowing faster token validation. Further variations of the token schemes can be devised in order to satisfy specific requirements; for example, some steps might need to be completed within very short times, which might require to reduce the computations within these steps, which can be achieved by delegating some computations within some other steps, for instance.

\subsection{Practical quantum token schemes $\mathcal{QT}_1$ and $\mathcal{QT}_2$ for two presentation points}
\label{sec1}

The quantum token schemes $\mathcal{QT}_1$ and $\mathcal{QT}_2$ presented in Tables \ref{real1} and \ref{real2} extend the quantum token schemes $\mathcal{IQT}_1$ and $\mathcal{IQT}_2$ to allow for various experimental imperfections (see Table \ref{tableimp}), and under some assumptions (see Table \ref{tableassu}). $\mathcal{QT}_1$ and $\mathcal{QT}_2$ can be implemented in practice with the photonic setups of \damian{Fig. \ref{setup}}. 

\begin{table*}
\begin{center}
\begin{tabular}{| p{15cm} | }
\hline
\\
\multicolumn{1}{|c|}{Preparation Stage}\\

\vspace{0.4mm}

0. Alice and Bob agree on a reference frame, on two presentation points $Q_0$ and $Q_1$ in the agreed frame, and on parameters $N\in\mathbb{N}$,  $\beta_{\text{PB}}\in\bigl(0,\frac{1}{2}\bigr)$, and $\gamma_{\text{det}},\gamma_{\text{err}}\in(0,1)$.\\

\\

\hline
\\
\multicolumn{1}{|c|}{Stage I}\\

\vspace{0.4mm}

1. For $k\in[N]$, $\mathcal{B}$ prepares bits $t_k$ and $u_k$ with respective probability distributions $P_{\text{PS}}^k(t_k)$ and $P_{\text{PB}}^k(u_k)$, satisfying $\frac{1}{2}-\beta_{\text{X}}\leq P_{\text{X}}^k(t) \leq \frac{1}{2}+\beta_{\text{X}}$, where $\beta_{\text{X}}\in\bigl(0,\frac{1}{2}\bigr)$ is a small parameter, for $\text{X}\in\{\text{PS},\text{PB}\}$, $t\in\{0,1\}$ and $k\in[N]$. We define $\mathbf{t}=(t_1,\ldots,t_N)$ and $\mathbf{u}=(u_1,\ldots,u_N)$. For $k\in[N]$, $\mathcal{B}$ prepares a quantum system $A_k$ in a quantum state $\lvert \psi_k\rangle$ and sends it to $\mathcal{A}$ with its label $k$. $\mathcal{B}$ chooses $k\in\Omega_{\text{noqub}}$ with probability $P_\text{noqub}>0$ or $k\in\Omega_{\text{qub}}$ with probability $1-P_\text{noqub}$. For $k\in\Omega_{\text{qub}}$, $\lvert\psi_k\rangle=\lvert \phi_{t_ku_k}^k\rangle$ is a qubit state, where $\langle \phi_{0u}^k\vert \phi_{1u}^k\rangle=0$ for $u\in\{0,1\}$, where the qubit orthonormal basis $\mathcal{D}_{u}^k=\{\lvert\phi_{tu}^k\rangle\}_{t=0}^1$ is the computational (Hadamard) basis up to an uncertainty angle $\theta$ on the Bloch sphere if $u=0$ ($u=1$).  For $k\in\Omega_{\text{noqub}}$, $\lvert\psi_k\rangle=\lvert \Phi_{t_ku_k}^k\rangle$ is a quantum state of arbitrary finite Hilbert space dimension greater than two. In photonic implementations, a vacuum or one-photon pulse has label $k\in\Omega_{\text{qub}}$, with a one-photon pulse encoding a qubit state, while a multi-photon pulse has label $k\in\Omega_{\text{noqub}}$ and encodes a quantum state of finite Hilbert space dimension greater than two.

\\

\vspace{0.4mm}

2. For $k\in[N]$, $\mathcal{A}$ measures $A_k$ in the qubit orthonormal basis $\mathcal{D}_{w_k}$, for $w_k\in\{0,1\}$ and $k\in[N]$. Due to losses, $\mathcal{A}$ only successfully measures quantum states $\vert \psi_k\rangle$ with labels $k$ from a proper subset $\Lambda$ of $[N]$. Let $W$ be the string of bit entries $w_k$ for $k\in\Lambda$ and let $n=\vert\Lambda\vert$. Conditioned on $k\in\Lambda$, the probability 
that $\mathcal{A}$ measures $A_k$ in the basis $\mathcal{D}_{w_k}$ satisfies $P_{\text{MB}}(w_k)=\frac{1}{2}$, for $w_k\in\{0,1\}$ and $k\in[N]$. $\mathcal{A}$ reports to $\mathcal{B}$ the set $\Lambda$. $\mathcal{B}$ does not abort if and only if $n\geq \gamma_\text{det} N$.\\

\vspace{0.4mm}

3. $\mathcal{A}$ chooses a one-to-one function $g: \Lambda\rightarrow [n]$, for example the numerical ordering, and sends it to $\mathcal{B}$. Let $y_j\in\{0,1\}$ indicate the basis $\mathcal{D}_{y_j}$ on which the quantum state $\vert \psi_k\rangle$ is measured by $\mathcal{A}$ and let $x_j\in\{0,1\}$ be the measurement outcome, where $j=g(k)$, for $k\in\Lambda$ and $j\in[n]$. Let $\mathbf{y}\in\{0,1\}^{n}$ and $\mathbf{x}\in\{0,1\}^{n}$ denote the strings of Alice's measurement bases and outcomes, respectively.\\

\vspace{0.4mm}

4. $\mathcal{A}$ sends $\mathbf{x}$ to $\mathcal{A}_i$, for $i\in\{0,1\}$.\\

\vspace{0.4mm}

5. $\mathcal{A}$ chooses a bit $z$ with probability $P_{\text{E}}(z)$ that satisfies $\frac{1}{2}-\beta_{\text{E}}\leq P_{\text{E}}(z) \leq \frac{1}{2}+\beta_{\text{E}}$, for $z\in\{0,1\}$, and for a small parameter $\beta_{\text{E}}\in\bigl(0,\frac{1}{2}\bigr)$. $\mathcal{A}$ computes the string $\mathbf{d}\in\{0,1\}^{n}$ with bit entries $d_j=y_j\oplus z$, for $j\in[n]$. $\mathcal{A}$ sends $\mathbf{d}$ to $\mathcal{B}$.\\

\vspace{0.4mm}

6. For $i\in\{0,1\}$, $\mathcal{B}$ sends $\mathbf{d}$ to  $\mathcal{B}_i$ and $\mathcal{B}_i$ computes the string $\mathbf{d}_i\in\{0,1\}^{n}$ with bit entries $d_{i,j}=d_j\oplus i$, for $j\in[n]$.\\

\vspace{0.4mm}

7. $\mathcal{B}$ uses $\mathbf{t},\mathbf{u},\Lambda$ and $g$ to compute the strings $\mathbf{s},\mathbf{r}\in\{0,1\}^{n}$, as follows. We define $r_j=t_k$, and $s_j=u_k$, where $j=g(k)$, for $j\in[n]$ and $k\in\Lambda$. We define $\mathbf{r}$ and $\mathbf{s}$ as the strings with bit entries $r_j$ and $s_j$, for $j\in[n]$. $\mathcal{B}$ sends $\mathbf{s}$ and $\mathbf{r}$ to $\mathcal{B}_i$, for $i\in\{0,1\}$.\\
\\

\hline
\\

\multicolumn{1}{|c|}{Stage II}\\

\vspace{0.4mm}

8. $\mathcal{A}$ chooses the presentation point $Q_b$ where to present the token, for some $b\in\{0,1\}$. $\mathcal{A}$ computes the bit $c=b\oplus z$ and sends it to $\mathcal{B}$.\\

\vspace{0.4mm}

9. $\mathcal{B}$ sends $c$ to $\mathcal{B}_i$, for $i\in\{0,1\}$.\\

\vspace{0.4mm}

10. For $i\in\{0,1\}$, in the causal past of $Q_i$, $\mathcal{B}_i$ computes the string $\tilde{\mathbf{d}}_i\in\{0,1\}^{n}$ with bit entries $\tilde{d}_{i,j}=d_{i,j}\oplus c$, for $j\in[n]$.\\

\vspace{0.4mm}

11. $\mathcal{A}$ sends a signal to $\mathcal{A}_b$ indicating to present the token at $Q_b$, and $\mathcal{A}_b$ presents the token $\mathbf{x}$ to $\mathcal{B}_b$ in $Q_b$.\\

\vspace{0.4mm}

12. $\mathcal{B}_b$ validates the token $\mathbf{x}$ received in $Q_b$ if  the Hamming distance between the strings $\mathbf{x}_b$ and $\mathbf{r}_b$ satisfies $d(\mathbf{x}_b,\mathbf{r}_b)\leq \lvert \Delta_b\rvert \gamma_\text{err}$, where $\Delta_v=\{j\in [n]\vert \tilde{d}_{v,j}=s_j\}$, and where $\mathbf{a}_v$ is the restriction of a string $\mathbf{a}\in\{\mathbf{x},\mathbf{r}\}$ to entries $a_j$ with $j\in\Delta_v$, for $v\in\{0,1\}$.\\

\\

\botrule
\end{tabular}
\caption{Practical quantum token scheme $\mathcal{QT}_1$ for two presentation points. Steps 1 to 9 take place within the intersection of the causal pasts of the presentation points. \damian{See Table \ref{notationtable} for a summary of the notation} \damiann{and Fig. \ref{newfigure} for an illustration of the scheme.}}
\label{real1}
\end{center}
\end{table*}

\begin{table}
\begin{center}
\begin{tabular}{| p{8cm} | }
\hline
\\
\multicolumn{1}{|c|}{Preparation Stage}\\

\vspace{0.5mm}

0. As step 0 of $\mathcal{QT}_1$.\\

\\

\hline
\\
\multicolumn{1}{|c|}{Stage I}\\

\vspace{0.5mm}

1. As step 1 of $\mathcal{QT}_1$.\\

\vspace{0.5mm}

2. The step 2 of $\mathcal{QT}_1$ is replaced by the following. $\mathcal{A}$ chooses a bit $z$ with probability $P_{\text{E}}(z)$ satisfying $\frac{1}{2}-\beta_{\text{E}}\leq P_{\text{E}}(z)\leq\frac{1}{2}+\beta_{\text{E}}$, for $z\in\{0,1\}$ and for a small parameter $\beta_{\text{E}}\in\bigl(0,\frac{1}{2}\bigr)$. $\mathcal{A}$ measures $A_k$ in the qubit orthonormal basis $\mathcal{D}_{z}$, for all $k\in[N]$. Due to losses, $\mathcal{A}$ only successfully measures quantum states $\vert \psi_k\rangle$ with labels $k$ from a proper subset $\Lambda$ of $[N]$. $\mathcal{A}$ reports to $\mathcal{B}$ the set $\Lambda$. Let $n=\lvert \Lambda\rvert$. $\mathcal{B}$ does not abort if and only if $n\geq \gamma_\text{det} N$.\\

\vspace{0.5mm}

3. As step 3 of $\mathcal{QT}_1$. The string $\mathbf{y}\in\{0,1\}^{n}$ of Alice's measurement bases has bit entries $y_j=z$ for $j\in[n]$.\\

\vspace{0.5mm}

4. As step 4 of $\mathcal{QT}_1$. The steps 5 and 6 of $\mathcal{QT}_1$ are discarded.\\

\vspace{0.5mm}

5. As step 7 of $\mathcal{QT}_1$.\\
\\

\hline
\\

\multicolumn{1}{|c|}{Stage II}\\

\vspace{0.5mm}

6. As steps 8 and 9 of $\mathcal{QT}_1$.\\

\vspace{0.5mm}

7. The step 10 of $\mathcal{QT}_1$ is replaced by the following. For $i\in\{0,1\}$, $\mathcal{B}_i$ computes the string $\tilde{\mathbf{d}}_i\in\{0,1\}^{n}$ with bit entries $\tilde{d}_{i,j}= i\oplus c$, for $j\in[n]$.\\

\vspace{0.5mm}

8. As steps 11 and 12 of $\mathcal{QT}_1$.\\

\\

\botrule
\end{tabular}
\caption{Practical quantum token scheme $\mathcal{QT}_2$ for two presentation points. \damian{See Table \ref{notationtable} for a summary of the notation} \damiann{and Fig. \ref{newfigure} for an illustration of the scheme.}}
\label{real2}
\end{center}
\end{table}

\begin{table}[!hbt]
\begin{center}
\small
\begin{tabular}{| >{\centering}p{2.0cm} | p{7.0cm} |}
\hline
\multicolumn{1}{|c|}{\normalsize \textbf{Symbol}} & \multicolumn{1}{|c|}{\normalsize \textbf{Brief description}} \\
\hline
\multicolumn{1}{|c|}{\normalsize $Q_i$} & Presentation points  \\
\hline
\multicolumn{1}{|c|}{\normalsize $\mathcal{A}$ ($\mathcal{B}$)} & Alice's (Bob's) agent participating in the quantum communication stage  \\
\hline
\multicolumn{1}{|c|}{\normalsize $\mathcal{A}_i$ ($\mathcal{B}_i$)} & Alice's (Bob's) agent by the presentation point $Q_i$  \\
\hline
\multicolumn{1}{|c|}{\normalsize $A_k$} & Quantum systems sent to Alice by Bob  \\
\hline
\multicolumn{1}{|c|}{\normalsize $N$} & Number of quantum states that Bob sends Alice \\
\hline
\multicolumn{1}{|c|}{\normalsize $\Omega_\text{qub}$} & Set of labels for prepared qubits states \\
\hline
\multicolumn{1}{|c|}{\normalsize $\Omega_\text{noqub}$} & Set of labels for prepared quantum states with dimension greater than two \\
\hline
\multicolumn{1}{|c|}{\normalsize $P_\text{noqub}$} & Probability that a prepared quantum state has dimension greater than two \\
\hline
\multicolumn{1}{|c|}{\normalsize $\bold{t}$} & String of bits encoding the quantum states prepared by Bob \\
\hline
\multicolumn{1}{|c|}{\normalsize $\bold{u}$} & String of bits encoding the bases of preparation by Bob \\
\hline
\multicolumn{1}{|c|}{\normalsize $\mathcal{D}_u^k$} & Qubit orthonormal bases of preparation by Bob \\
\hline
\multicolumn{1}{|c|}{\normalsize $\mathcal{D}_{w_k}$} & Qubit orthonormal bases of measurement by Alice \\
\hline
\multicolumn{1}{|c|}{\normalsize $P_{\text{MB}}(w_k)$} & Probability distribution for Alice's measurement bases \\
\hline
\multicolumn{1}{|c|}{\normalsize $\beta_\text{PB}$} & Bias for preparation basis \\
\hline
\multicolumn{1}{|c|}{\normalsize $\beta_\text{PS}$} & Bias for preparation state \\
\hline
\multicolumn{1}{|c|}{\normalsize $\Lambda$} & Set of labels for quantum states successfully measured by Alice \\
\hline
\multicolumn{1}{|c|}{\normalsize $W$} & String of bits encoding the measurement bases for the quantum states successfully measured by Alice \\
\hline
\multicolumn{1}{|c|}{\normalsize $\gamma_\text{det}$} & Minimum rate for states reported by Alice as successfully measured for Bob not aborting \\
\hline
\multicolumn{1}{|c|}{\normalsize $\gamma_\text{err}$} & Maximum error rate allowed by Bob for validating Alice's token  \\
\hline
\multicolumn{1}{|c|}{\normalsize $g$} & One-to-one function $g:\lvert \Lambda\rvert\rightarrow [n]$  \\
\hline
\multicolumn{1}{|c|}{\normalsize $\bold{y}$ ($\bold{x}$)} & String of bits encoding Alice's measurement outcomes (bases)  \\
\hline
\multicolumn{1}{|c|}{\normalsize $z$} & Bit chosen by Alice  \\
\hline
\multicolumn{1}{|c|}{\normalsize $P_\text{E}(z)$} & Probability distribution for bit $z$ chosen by Alice \\
\hline
\multicolumn{1}{|c|}{\normalsize $\beta_\text{E}$} & Bias for the probability distribution $P_\text{E}(z)$ \\
\hline
\multicolumn{1}{|c|}{\normalsize $\bold{d}$} & String with bit entries $d_j=y_j\oplus z$ that Alice sends Bob \\
\hline
\multicolumn{1}{|c|}{\normalsize $\bold{d}_i$} & String with bit entries $d_{i,j}=d_j\oplus i$ computed by Bob's agent $\mathcal{B}_i$ \\
\hline
\multicolumn{1}{|c|}{\normalsize $\bold{r}$ ($\bold{s}$)} & String of bits encoding Bob's prepared states (preparation bases) for the states that Alice reports as successfully measured \\
\hline
\multicolumn{1}{|c|}{\normalsize $b$} & Bit encoding Alice's chosen presentation point \\
\hline
\multicolumn{1}{|c|}{\normalsize $c$} & Bit $c=b\oplus z$, which Alice sends Bob \\
\hline
\multicolumn{1}{|c|}{\normalsize $\tilde{\bold{d}}_i$} & String with bit entries  $\tilde{d}_{i,j}=d_{i,j}\oplus c$ computed by Bob's agent $\mathcal{B}_i$\\
\hline
\multicolumn{1}{|c|}{\normalsize $\Delta_\nu$} & Set of labels defined by $\Delta_\nu=\bigl\{j\in[n]\vert \tilde{d}_{\nu,j}=s_j\bigr\}$, for $\nu\in\{0,1\}$\\
\hline
\multicolumn{1}{|c|}{\normalsize $\bold{a}_\nu$} & The substring of $\bold{a}\in\{\bold{x},\bold{r}\}$ restricted to bit entries $a_k$ with $k\in\Delta_\nu$, for $\nu\in\{0,1\}$\\
\hline
\botrule
\end{tabular}
\caption{\normalsize Summary of notation used for $\mathcal{QT}_1$ and $\mathcal{QT}_2$.\label{notationtable}}
\end{center}
\end{table}

\begin{figure*}[!htb]
\includegraphics{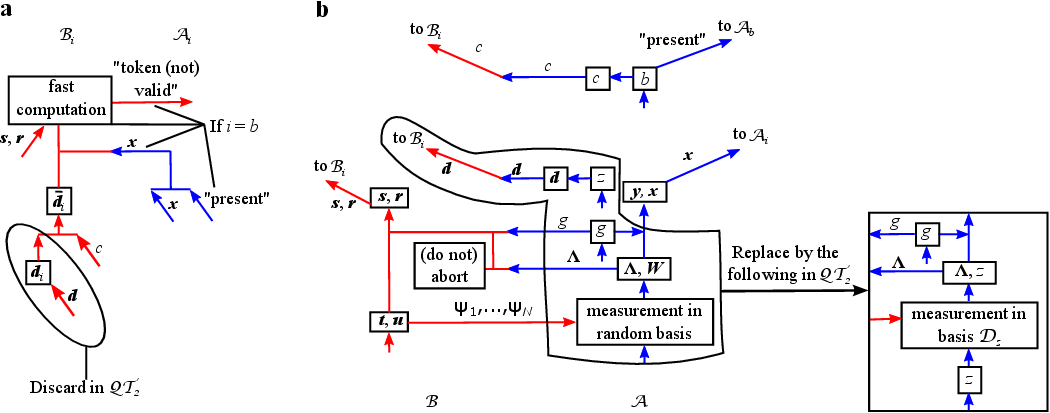}
\caption{\label{newfigure} \textbf{Diagram of the quantum token schemes $\mathcal{QT}_1$ and $\mathcal{QT}_2$.} Alice's (Bob's) steps are indicated with the blue (red) arrows. The differences between $\mathcal{QT}_1$ and $\mathcal{QT}_2$ are shown. \textbf{b} The steps perfomed by Alice's and Bob's agents $\mathcal{A}$ and $\mathcal{B}$ in $\mathcal{QT}_1$ are illustrated. \textbf{a} The steps of Alice's and Bob's agents $\mathcal{A}_i$ and $\mathcal{B}_i$ in $\mathcal{QT}_1$ are shown, for $i\in\{0,1\}$. The case $i=b$ represents Alice's token presentation and Bob's validation/rejection. }
\end{figure*}

\begin{figure}[!htb]
\includegraphics[scale=0.54]{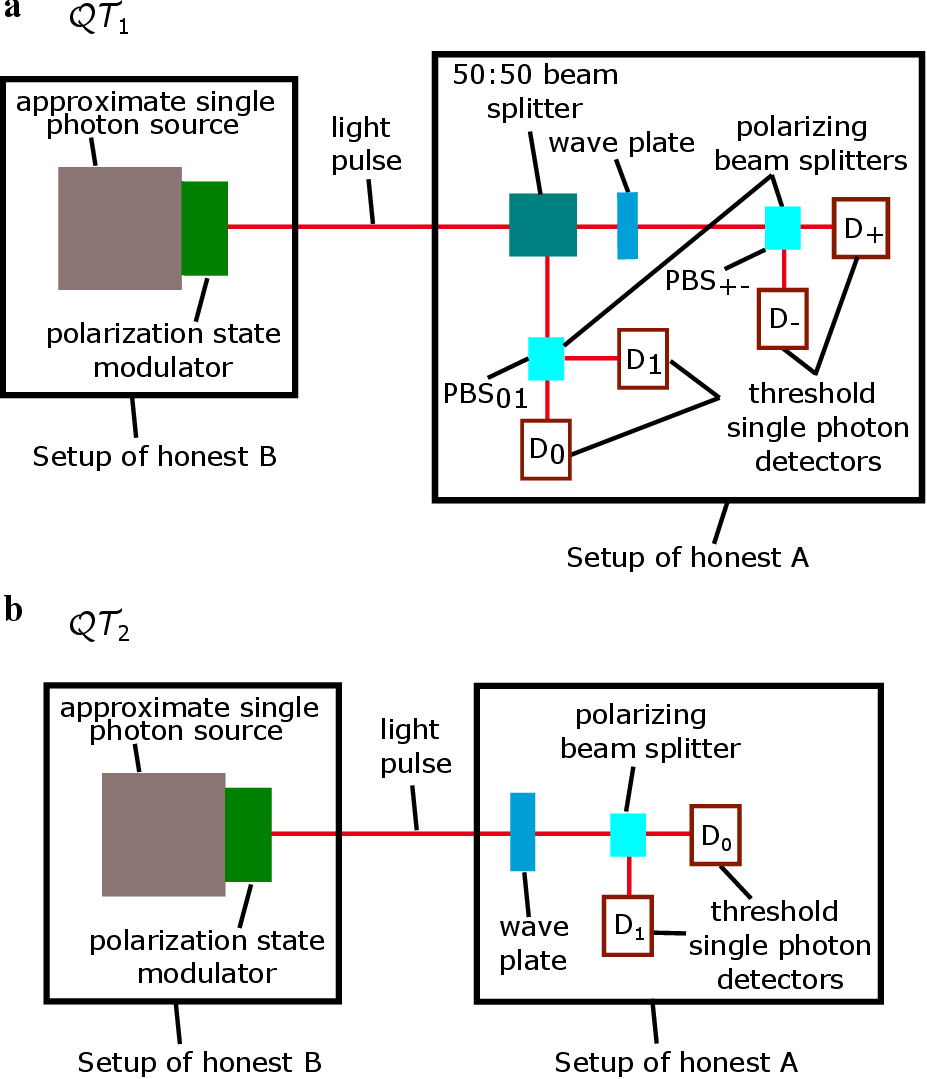}
\caption{\label{setup} \textbf{\damian{Photonic setups to implement the quantum stage of $\mathcal{QT}_1$ and $\mathcal{QT}_2$}.} 
\damian{In both $\mathcal{QT}_1$ and $\mathcal{QT}_2$, the} setup of honest $\mathcal{B}$ comprises an
  approximate single photon source and a polarization state modulator, encoding the quantum state $\lvert\psi_k\rangle$ in the polarization degrees of freedom of a photon pulse labelled by $k$, for $k\in[N]$. \damian{\textbf{a} In $\mathcal{QT}_1$, the} setup of honest $\mathcal{A}$ comprises a 50:50 beam splitter, a wave plate, two polarizing beam splitters (PBS$_{01}$ and PBS$_{+-}$) and four threshold single photon detectors D$_0$, D$_1$, D$_+$ and D$_-$. In order to counter multi-photon attacks by $\mathcal{B}$, $\mathcal{A}$ implements the following reporting strategy that we call here \emph{reporting strategy 1}: $\mathcal{A}$ assigns successful measurement outcomes in the basis $\mathcal{D}_0$ ($\mathcal{D}_1$) with unit probability for the pulses in which at least one of the detectors D$_0$ and D$_1$ (D$_+$ and D$_-$) click and  D$_+$ and D$_-$ (D$_0$ and D$_1$) do not click. As follows from Ref. \cite{BCDKPG21}, this reporting strategy offers perfect protection against arbitrary multi-photon attacks, given assumption
   F (see Table \ref{tableassu}, and Lemma \ref{lemmax} in Methods). \damian{\textbf{b} In $\mathcal{QT}_2$, the} setup of honest $\mathcal{A}$ comprises a wave plate set in one of two positions, according to the value of her bit $z$, a polarizing beam splitter and two threshold
  single photon detectors D$_0$ and D$_1$. In order to counter multi-photon attacks by $\mathcal{B}$, $\mathcal{A}$ implements the following reporting strategy that we call here \emph{reporting strategy 2}: $\mathcal{A}$ reports to $\mathcal{B}$ as successful measurements those in which at least one of her two detectors click. As follows from Ref. \cite{BCDKPG21}, this reporting strategy offers perfect protection against arbitrary multi-photon attacks, given assumption F 
  (see Table \ref{tableassu}, and Lemma \ref{lemmax} in Methods).}
\end{figure}

\begin{table*}
\begin{center}
\small
\begin{tabular}{| >{\centering}p{0.5cm} | p{5.5cm} | p{11.3cm} |}
\hline
\multicolumn{1}{|c|}{\normalsize \textbf{No}} & \multicolumn{1}{|c|}{\normalsize \textbf{Brief description}} & \multicolumn{1}{|c|}{\normalsize \textbf{Explanation and comments}}\\
\hline
\multicolumn{1}{|c|}{\normalsize \textbf{1}} & For $k\in[N]$, there is a small probability $P_\text{noqub}>0$ for $\mathcal{B}$ to prepare a quantum state $\lvert\psi_k\rangle$ of arbitrary finite Hilbert space dimension greater than two.  & In photonic implementations, we define $P_\text{noqub}$ and $\Omega_{\text{noqub}}\subseteq[N]$ as the probability that a pulse is multi-photon and as the set of labels for multi-photon pulses (see Methods). We define $\Omega_{\text{qub}}=[N]\setminus\Omega_{\text{noqub}}$ as the set of labels for vacuum or one-photon pulses, where the subindex refers to `qubit'. When showing unforgeability, we treat vacuum pulses as one-photon pulses encoding the qubit state Bob attempted to send. Since this gives Alice extra options that
cannot make it more difficult for her to cheat, the
deduced unforgeability bound 
holds in
general. 
 A Poissonian photon source (e.g. weak coherent) with average photon number $\mu<<1$ gives $P_\text{noqub}=1-(1+\mu)e^{-\mu}=\frac{\mu^2}{2}+O(\mu^3)$,  while a heralded single-photon source can give extremely small values for $P_\text{noqub}$, of the order of $10^{-10}$ for usual experimental parameters.
\\
\hline

\multicolumn{1}{|c|}{\normalsize \textbf{2}}  & For $k\in\Omega_{\text{qub}}$, $\mathcal{B}$ prepares $\lvert \psi_k\rangle=\lvert \phi_{t_ku_k}^k\rangle$ in a qubit orthonormal basis $\mathcal{D}_{u_k}^k$ that is the computational (Hadamard) basis within an uncertainty angle $\theta\in\bigl(0,\frac{\pi}{4}\bigr)$ on the Bloch sphere if $u_k=0$ ($u_k=1$).   & 
Thus, the angle on the Bloch sphere between the states $\lvert \phi_{t0}^k\rangle$ and $\lvert \phi_{t1}^k\rangle$ is guaranteed to be within the range $[\frac{\pi}{2}-2\theta,\frac{\pi}{2}+2\theta]$, for $k\in\Omega_{\text{qub}}$. We define $O(\theta)=\frac{1}{\sqrt{2}}\bigl(\cos\theta +\sin\theta \bigr)$, where the notation refers to `overlap' on the Bloch sphere. It follows that $\lvert \langle \phi_{t0}^k\vert \phi_{t'1}^k\rangle \rvert\leq O(\theta)$, for some $O(\theta)\in\bigl(\frac{1}{\sqrt{2}},1\bigr)$, for $t,t'\in\{0,1\}$ and $k\in\Omega_{\text{qub}}$.\\
\hline
\multicolumn{1}{|c|}{\normalsize \textbf{3}}  & For $k\in[N]$, $\mathcal{B}$ generates the bits $t_k$ and $u_k$ with probability distributions $P_{\text{PS}}^k(t_k)$ and $P_{\text{PB}}^k(u_k)$ that have small deviations from the random distributions given by  biases $\beta_{\text{PS}}, \beta_{\text{PB}} >0 $.  &  That is, we have $\frac{1}{2}-\beta_{\text{X}}\leq P_{\text{X}}^k(t)\leq \frac{1}{2}+\beta_\text{X}$, with $0<\beta_\text{X}<\frac{1}{2}$, for $t\in\{0,1\}$, $k\in[N]$ and $\text{X}\in\{\text{PS},\text{PB}\}$. The subindices `PS' and `PB' refer to `preparation state' and `preparation basis', respectively.
It is useful for our security analysis to define: $\lambda(\theta,\beta_{\text{PB}})=\frac{1}{2}\bigl(1-\sqrt{1-[1-(O(\theta))^2](1-4\beta_\text{PB}^2)}\bigr)$. It follows from $0<\beta_\text{PB}<\frac{1}{2}$ and $\frac{1}{\sqrt{2}}<O(\theta)<1$ that $0<\lambda(\theta,\beta_{\text{PB}}) < \frac{1}{2}\bigl(1-O(\theta)\bigr)<\frac{1}{2}\bigl(1-\frac{1}{\sqrt{2}}\bigr)$.
\\
\hline
\multicolumn{1}{|c|}{\normalsize \textbf{4}}  & A fraction of the quantum states transmitted from $\mathcal{B}$ to $\mathcal{A}$ is lost. In photonic setups, $\mathcal{A}$ has single photon detectors with non unit detection efficiencies. & Because of losses and non unit detection efficiencies (in
photonic setups), $\mathcal{A}$ must report to $\mathcal{B}$ the set $\Lambda\subset [N]$ of labels of the successfully measured states. $\mathcal{B}$ does not abort if and only if $\lvert \Lambda \rvert \geq \gamma_\text{det} N$,
where the subindex `det' stands for `detection'. \\
\hline
\multicolumn{1}{|c|}{\normalsize \textbf{5}}  & For $k\in[N]$, $\mathcal{A}$ measures the received state $\lvert\psi_k\rangle$ in one of two distinct orthogonal qubit basis, $\mathcal{D}_0$ and $\mathcal{D}_1$, where this pair of bases is arbitrary.  &  $\mathcal{A}$ applying a measurement on a qubit basis $\mathcal{D}_{0}$ ($\mathcal{D}_{1}$) on a received quantum state that is not a qubit, i.e. for $k\in\Omega_{\text{noqub}}$, means that $\mathcal{A}$ sets her devices  as she would do to apply a measurement in the qubit basis $\mathcal{D}_{0}$ ($\mathcal{D}_{1}$) -- we note that $\mathcal{A}$ does not know the sets $\Omega_{\text{qub}}$ and $\Omega_{\text{noqub}}$. For photonic setups, this may include arranging a set of wave plates, polarizing beam splitters and single photon detectors in a particular setting.
If $\mathcal{D}_{0}$ and $\mathcal{D}_{1}$ are very different to the computational and Hadamard bases, the number of measurement errors in Alice's outcomes is high. But, this is considered in our security analysis via Alice's error rate. Moreover, the set of two measurement bases applied by $\mathcal{A}$ could vary slightly for different quantum states $\lvert\psi_k\rangle$, i.e., for different $k\in[N]$. However, we can include these deviations from the measurement bases $\mathcal{D}_{0}$ and $\mathcal{D}_{1}$ of $\mathcal{A}$ in the bases $\mathcal{D}_{u_k}^k$ of preparation by $\mathcal{B}$, and assume that $\mathcal{A}$ applies either $\mathcal{D}_{0}$ or $\mathcal{D}_{1}$ to $\lvert \psi_k\rangle$, for $k\in[N]$. In other words, the uncertainty angle $\theta$ on the Bloch sphere accounts for both preparation and measurement misalignments. Thus, our analysis is without loss of generality.\\
\hline
\multicolumn{1}{|c|}{\normalsize \textbf{6}}  & There are errors in Alice's quantum measurements. & Thus, Alice obtains some errors in the measurements that she performs in the same basis of preparation by Bob. For this reason, in the validation stage, Bob allows a fraction of bit errors in Alice's reported measurement outcomes, up to a predetermined small threshold $\gamma_{\text{err}}>0$, where `err' stands for `errors'.\\
\hline
\multicolumn{1}{|c|}{\normalsize \textbf{7}}  & $\mathcal{A}$ generates the bit $z$ with probability distribution $P_{\text{E}}(z)$ that has small deviation from the random distribution given by a bias $\beta_{\text{E}}>0$.  & That is, we have that $\frac{1}{2}-\beta_{\text{E}}\leq P_{\text{E}}(z)\leq
\frac{1}{2}+\beta_{\text{E}}$,
for $z\in\{0,1\}$. The subindex `E'  refers to `encoding'. 
$\mathcal{A}$ can guarantee the parameter $\beta_{\text{E}}$ to decrease exponentially with a number $N_\text{CRB}$ of close-to-random bits with biases not greater than $\beta_{\text{CRB}}\in\bigl(0,\frac{1}{2}\bigr)$, as follows from the Piling-up Lemma \cite{Pilinguplemma}.\\

\hline
\multicolumn{1}{|c|}{\normalsize \textbf{8}}  & In photonic setups, the single photon detectors used by $\mathcal{A}$ are threshold, i.e. they cannot distinguish the number of photons activating a detection. Moreover, the detectors have non-zero dark count probabilities.  & Thus, for some photon pulses received from $\mathcal{B}$, more than one of the detectors of $\mathcal{A}$ click. In order to counter multi-photon attacks \cite{BCDKPG21}, in which $\mathcal{B}$ sends and tracks multi-photon pulses to obtain information about the measurement bases of $\mathcal{A}$, and guarantee privacy, $\mathcal{A}$ must carefully choose how to report multiple clicks to $\mathcal{B}$, i.e. how to define successful measurements. For this reason, in the second step of our token schemes $\mathcal{QT}_1$ and $\mathcal{QT}_2$ with the photonic setups of \damian{Fig. \ref{setup}}, $\mathcal{A}$ implements the reporting strategies 1 and 2, respectively. As follows from straightforward extensions of Lemmas 1 and 12 of Ref. \cite{BCDKPG21}, assumption F (see Table \ref{tableassu}) guarantees that these reporting strategies offer perfect protection against arbitrary multi-photon attacks (see Lemma \ref{lemmax} in Methods).\\
\botrule
\end{tabular}
\caption{\normalsize Allowed experimental imperfections for $\mathcal{QT}_1$ and $\mathcal{QT}_2$.\label{tableimp}}
\end{center}
\end{table*}

\begin{table*}
\begin{center}
\small
\begin{tabular}{| >{\centering}p{1cm} | p{6.5cm} | p{9.8cm} |}
\hline
\multicolumn{1}{|c|}{\normalsize \textbf{Label}} & \multicolumn{1}{|c|}{\normalsize \textbf{Brief description}} & \multicolumn{1}{|c|}{\normalsize \textbf{Explanation and comments}}\\
\hline
\multicolumn{1}{|c|}{\normalsize \textbf{A}}  & For $k\in\Omega_{\text{qub}}$, $\mathcal{B}$ prepares  $\lvert\psi_k\rangle=\lvert \phi_{t_ku_k}^k\rangle$, where $\langle \phi_{0u}^k\vert \phi_{1u}^k\rangle=0$, defining the qubit orthonormal basis $\mathcal{D}_{u}^k=\{\lvert\phi_{tu}^k\rangle\}_{t=0}^1$, for $u\in\{0,1\}$. & That is, we assume that $\mathcal{B}$ prepares each qubit state from exactly two qubit bases. 
However, in the most general case (not considered here), $\mathcal{B}$ prepares each qubit state from a set of four qubit states that does not necessarily define two qubit basis.
\\
\hline
\multicolumn{1}{|c|}{\normalsize \textbf{B}}  & $\mathcal{B}$ generates the bit strings $\mathbf{t}=(t_1,\ldots,t_N)$ and $\mathbf{u}=(u_1,\ldots,u_N)$ with probability distributions that are exactly products of single bit probability distributions. & In the general case (not considered here), the strings $\mathbf{t}$ and $\mathbf{u}$ could be generated with a probability distribution in which $\mathbf{t}$ and $\mathbf{u}$, and different bit entries of $\mathbf{t}$ and $\mathbf{u}$, could be correlated.\\
\hline
\multicolumn{1}{|c|}{\normalsize \textbf{C}}  & The set $\Lambda$ of labels transmitted to $\mathcal{B}$ in step 2 of $\mathcal{QT}_1$ and $\mathcal{QT}_2$ gives $\mathcal{B}$ no information about the string $W$ and the bit $z$. & In the photonic setups of \damian{Fig. \ref{setup}} to implement $\mathcal{QT}_1$ and $\mathcal{QT}_2$, with the reporting strategies 1 and 2, respectively, assumption C (and also assumption D for $\mathcal{QT}_1$) follows from assumptions E and F (see Lemma \ref{lemmax} in Methods). \\
\hline
\multicolumn{1}{|c|}{\normalsize \textbf{D}}  & In $\mathcal{QT}_1$, conditioned on reporting the quantum state $\lvert\psi_k\rangle$ as successfully measured, i.e. conditioned on $k\in\Lambda$, $\mathcal{A}$ measures $\lvert\psi_k\rangle$ in an orthogonal qubit basis $\mathcal{D}_{w_k}$ with a probability distribution $P_{\text{MB}}(w_k)=\frac{1}{2}$, for $w_k\in\{0,1\}$ and $k\in[N]$, where the subindex denotes `measurement basis'. & This is a necessary, but in general not sufficient, condition for $\mathcal{QT}_1$ to satisfy assumption C. If this assumption did not hold, there would be at least one label $k'\in\Lambda$ for which $P_{\text{MB}}(w_{k'}=i)>P_{\text{MB}}(w_{k'}=i\oplus 1)$, for some $i\in\{0,1\}$. Thus, $\mathcal{B}$ could in principle guess the entry $w_{k'}$ of $W$ with probability greater than $\frac{1}{2}$. This would mean that the set $\Lambda$ reported by $\mathcal{A}$ would have given $\mathcal{B}$ some information about $W$, in contradiction with assumption C.\\
\hline
\multicolumn{1}{|c|}{\normalsize \textbf{E}}  & $\mathcal{B}$ cannot use degrees of freedom not previously agreed for the transmission of the quantum states to affect, or obtain information about, the statistics of the quantum measurement devices of $\mathcal{A}$. & This assumption guarantees that $\mathcal{A}$ is perfectly protected from any side-channel attack by $\mathcal{B}$ in any type of physical setup (not necessarily photonic) \cite{BCDKPG21}. \\
\hline
\multicolumn{1}{|c|}{\normalsize \textbf{F}}  & In the photonic setup of Fig. \ref{setup}, the detectors D$_0$, D$_1$, D$_+$ and D$_-$ of $\mathcal{A}$ have equal detection efficiencies $\eta\in(0,1)$, and respective dark count probabilities $d_0,  d_1, d_+, d_- \in(0,1)$ satisfying $(1-d_0)(1-d_1)=(1-d_+)(1-d_-)$, for $k\in[N]$.
In the photonic setup of \damian{Fig. \ref{setup}}, the detectors D$_0$ and D$_1$ of $\mathcal{A}$ satisfy that: 1) their detection efficiencies have the same value $\eta\in(0,1)$, for $k\in[N]$; and 2) their dark count probabilities have values $d_0\in(0,1)$ and $d_1\in(0,1)$, for $k\in[N]$. Dark counts and each photo-detection are independent random events, for $k\in[N]$. & In our notation, the term `detection efficiency' includes the quantum
efficiency of the detectors of $\mathcal{A}$ and the transmission
efficiency from the setup of $\mathcal{B}$ to the detectors. We note that the condition $(1-d_0)(1-d_1)=(1-d_+)(1-d_-)$ can be satisfied if $d_0=d_+$ and $d_1=d_-$, or if $d_0=d_-$ and $d_1=d_+$, for instance. Exactly equal detection efficiencies cannot be guaranteed in practice. But, attenuators can be used to make the detector efficiencies approximately equal. Furthermore, $\mathcal{A}$ can effectively make the dark count probabilities of her detectors approximately equal by simulating dark counts in the detectors with lower dark count probabilities so that they approximate the dark count probability of the detector with highest dark count probability. To our knowledge, that dark counts and each photo-detection are independent random events is a valid assumption.\\
\hline
\multicolumn{1}{|c|}{\normalsize \textbf{G}}  & In photonic setups, from the perspective of $\mathcal{A}$, the pulses of $\mathcal{B}$ are mixtures
of Fock states: in particular $\mathcal{A}$ has no information about
relative phases of the components with definite photon number. & If this assumption is not satisfied, the quantum state received by $\mathcal{A}$ could not be described by our analysis, opening the possibility to attacks more powerful than the ones considered in our security proof (e.g. more powerful state discrimination attacks \cite{DHL00}). This assumption is consistent with our security analysis and is satisfied in practice if $\mathcal{B}$ uses a weak coherent source and he uniformly randomizes the phase of each pulse transmitted to $\mathcal{A}$; or if $\mathcal{B}$ uses an arbitrary photonic source with arbitrary signal states and he applies a physical operation to the transmitted pulses with the property that it applies a random
phase $\varphi$ per photon -- i.e. an $l$-photon pulse acquires an amplitude
$e^{il\varphi}$ \cite{ILM07}. Alternatively, this condition can be satisfied to a good approximation if $\mathcal{B}$ uses a photonic source with low spatio-temporal coherence, for example, a source comprising LEDs \cite{MSDS10}, as in our experimental tests reported below.\\
\botrule
\end{tabular}
\caption{\normalsize Assumptions for $\mathcal{QT}_1$ and $\mathcal{QT}_2$.\label{tableassu}}
\end{center}
\end{table*}

The token schemes $\mathcal{QT}_1$ and $\mathcal{QT}_2$ can be modified in various ways, as discussed for the token schemes $\mathcal{IQT}_1$ and $\mathcal{IQT}_2$.

Regarding correctness, we note in the token scheme $\mathcal{QT}_1$ that if Alice follows the token scheme honestly and chooses to present the token in $Q_b$, then we have that  $\tilde{\mathbf{d}}_b$ has bit entries $\tilde{d}_{b,j}=d_{b,j}\oplus c=d_j\oplus b\oplus c=d_j\oplus z=y_j$, for $j\in[n]$. Thus, $\tilde{\mathbf{d}}_b=\mathbf{y}$, i.e. $\tilde{\mathbf{d}}_b$ corresponds to the string of measurement bases implemented by Alice on the quantum states reported to be successfully measured.  Similarly, in the token scheme $\mathcal{QT}_2$ if Alice follows the token scheme honestly and chooses to present the token in $Q_b$, then we have that $\tilde{\mathbf{d}}_b$ has bit entries $\tilde{d}_{b,j}=b\oplus c=z=y_j$, for $j\in[n]$. Thus, as above, $\tilde{\mathbf{d}}_b=\mathbf{y}$, i.e. $\tilde{\mathbf{d}}_b$ corresponds to the string of measurement bases implemented by Alice on the quantum states reported to be successfully measured. Therefore, in both token schemes $\mathcal{QT}_1$ and $\mathcal{QT}_2$, if Alice can guarantee her error probability to be bounded by $E<\gamma_\text{err}$ then with very high probability she will make less than $\lvert \Delta_b\rvert \gamma_\text{err}$ bit errors in the $\lvert \Delta_b\rvert$ quantum states that she measured in the basis of preparation by Bob.

Let $P_{\text{det}}$ be the probability that a quantum state $\lvert \psi_k\rangle$ transmitted by Bob is reported by Alice as being successfully measured, with label $k\in\Lambda$, for $k\in[N]$. Let $E$ be the probability that Alice obtains a wrong measurement outcome when she measures a quantum state $\lvert\psi_k\rangle$ in the basis of preparation by Bob; if the error rates $E_{tu}$ are different for different prepared states, labelled by $t$, and for different measurement bases, labelled by $u$, we simply take $E=\max_{t,u}\{E_{tu}\}$. 

The robustness, correctness, privacy and unforgeability of $\mathcal{QT}_1$ and $\mathcal{QT}_2$ are stated by the following lemmas, proven in Appendix \ref{proofoflemmas}, and theorem, proven in Appendix \ref{newapp}. These lemmas and theorem consider parameters $\gamma_\text{det},\gamma_\text{err}\in(0,1)$, allow for the experimental imperfections of Table \ref{tableimp} and make the assumptions of Table \ref{tableassu}. \damiann{A diagram presenting the conditions under which robustness, correctness and unforgeability are satisfied simultaneously is given in Fig. \ref{figdiagram}.}

\begin{lemma}
\label{robust1}
If
\begin{equation}
				\label{rob1}
		0<\gamma_\text{det} <P_{\text{det}},
	\end{equation}
then $\mathcal{QT}_1$ and $\mathcal{QT}_2$ are $\epsilon_{\text{rob}}-$robust with
\begin{equation}
			\label{rob2}
			\epsilon_{\text{rob}}=e^{-\frac{P_{\text{det}}N}{2}\bigl(1-\frac{\gamma_\text{det} }{P_{\text{det}}}\bigr)^2.
			}
			\end{equation}
\end{lemma}

\begin{lemma}
\label{correct1}
If
\begin{eqnarray}
				\label{cor1}
				0&<&\frac{\gamma_\text{err}}{2}<E<\gamma_\text{err},\nonumber\\
		0&<&\nu_\text{cor}<\frac{P_{\text{det}}(1-2\beta_\text{PB})}{2},
	\end{eqnarray}
then $\mathcal{QT}_1$ and $\mathcal{QT}_2$ are $\epsilon_{\text{cor}}-$correct with
\begin{equation}
			\label{cor2}
			\epsilon_{\text{cor}}=e^{-\frac{P_{\text{det}}(1-2\beta_\text{PB})N}{4}\bigl(1-\frac{2\nu_\text{cor}}{P_{\text{det}}(1-2\beta_\text{PB})}\bigr)^2}+e^{-\frac{E\nu_\text{cor} N}{3}\bigl(\frac{\gamma_\text{err}}{E}-1\bigr)^2}.
			\end{equation}
\end{lemma}

\begin{lemma}
\label{Bob1}
$\mathcal{QT}_1$ and  $\mathcal{QT}_2$ are $\epsilon_{\text{priv}}-$private with
\begin{equation}
\label{Bo1}
\epsilon_{\text{priv}}=\beta_{\text{E}}.
\end{equation}
\end{lemma}

\begin{theorem}
\label{Alice1}
Consider the constraints
\begin{eqnarray}
\label{Al1}
0&<&\gamma_\text{err}<\lambda(\theta,\beta_{\text{PB}}),\nonumber\\
0&<&P_\text{noqub}<\nu_\text{unf}<\min\biggl\{\!2P_\text{noqub},\gamma_\text{det}\biggl(\!1\!-\!\frac{\gamma_\text{err}}{\lambda(\theta,\beta_{\text{PB}})}\!\biggr)\!\biggr\},\nonumber\\
		0&<&\beta_\text{PS}<\frac{1}{2}\biggl[e^{\frac{\lambda(\theta,\beta_{\text{PB}})}{2}\bigl(1-\frac{\delta}{\lambda(\theta,\beta_{\text{PB}})}\bigr)^2}-1\biggr].
	\end{eqnarray}
	We define the function
\begin{eqnarray}
\label{Al3}	
	&&f(\!\gamma_\text{err},\!\beta_\text{PS},\!\beta_\text{PB},\!\theta,\!\nu_\text{unf},\!\gamma_\text{det})\nonumber\\
	&&\quad=\!(\gamma_\text{det}\!-\!\nu_\text{unf})\!\biggl[\!\frac{\lambda(\theta,\beta_{\text{PB}}) }{2}\!\biggl(\!1\!-\!\frac{\delta}{\lambda(\theta,\beta_{\text{PB}})}\!\biggr)^2\!\!-\!\ln(1\!+\!2\beta_\text{PS})\!\biggr]\nonumber\\
	&&\quad\qquad-\bigl(1\!-\!(\gamma_\text{det}\!-\!\nu_\text{unf})\bigr)\ln\bigl[1+h(\beta_\text{PS},\beta_\text{PB},\theta)\bigr],
	\end{eqnarray}	
	 where
	 \begin{eqnarray}
\label{Al4}
h(\beta_\text{PS},\beta_\text{PB},\theta)&=&2\beta_\text{PS}\sqrt{\frac{1}{2}\!+\!2\beta_\text{PB}^2\!+\!\Bigl(\frac{1}{2}\!-\!2\beta_\text{PB}^2\Bigr)\sin(2\theta)},\nonumber\\
\delta&=&\frac{\gamma_\text{det}\gamma_\text{err}}{\gamma_\text{det}-\nu_\text{unf}}.
\end{eqnarray}
There exist parameters satisfying the constraints (\ref{Al1}), for which $f(\!\gamma_\text{err},\!\beta_\text{PS},\!\beta_\text{PB},\!\theta,\!\nu_\text{unf},\!\gamma_\text{det})>0$. For these parameters, $\mathcal{QT}_1$ and $\mathcal{QT}_2$ are $\epsilon_{\text{unf}}-$unforgeable with
\begin{equation}
\label{Al5}
\epsilon_{\text{unf}}=e^{-\frac{P_\text{noqub}N}{3}\bigl(\frac{\nu_\text{unf}}{P_\text{noqub}}-1\bigr)^2} +e^{-Nf(\gamma_\text{err},\beta_\text{PS},\beta_\text{PB},\theta,\nu_\text{unf},\gamma_\text{det})}.	
\end{equation}
\end{theorem}

\begin{figure}[!htb]
\includegraphics[scale=0.8]{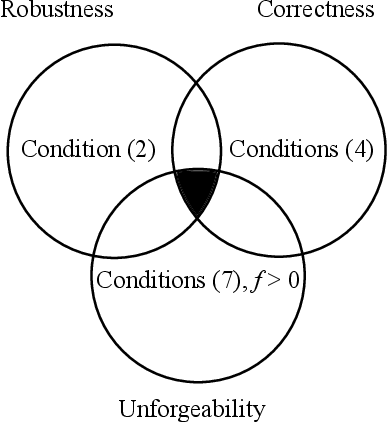}
\caption{\label{figdiagram} 
\textbf{Illustration of security conditions for $\mathcal{QT}_1$ and $\mathcal{QT}_2$.} A diagram presenting the conditions under which robustness, correctness and unforgeability of the quantum token schemes $\mathcal{QT}_1$ and $\mathcal{QT}_2$ are satisfied simultaneously is illustrated (see Lemmas \ref{robust1} and \ref{correct1}, and Theorem \ref{Alice1}). The function $f$ is defined by (\ref{Al3}). If all the conditions are satisfied (filled area) then there exist a sufficiently large integer $N$ such that $\mathcal{QT}_1$ and $\mathcal{QT}_2$ are $\epsilon_\text{rob}$-robust, $\epsilon_\text{cor}$-correct and $\epsilon_\text{unf}$-unforgeable, for desired values of $\epsilon_\text{rob},\epsilon_\text{cor},\epsilon_\text{unf}>0$. }
\end{figure}

We note in step 0 of $\mathcal{QT}_1$ and $\mathcal{QT}_2$ that Alice and Bob agree on parameters $N$, $\beta_{\text{PB}}$, $\gamma_\text{det}$ and $\gamma_\text{err}$. As follows from Lemmas \ref{robust1} -- \ref{Bob1}, in order for Alice to obtain a required degree of correctness, robustness and privacy, she must guarantee her experimental parameters $P_\text{det}$, $E$ and $\beta_{\text{E}}$ to be good enough. This is independent of any experimental parameters of Bob, except for the previously agreed parameter $\beta_{\text{PB}}$, which plays a role in correctness but not in robustness or privacy. Additionally, Alice must choose a suitable mathematical variable $\nu_{\text{cor}}$ to compute a guaranteed degree of correctness, as given by the bound of Lemma \ref{correct1}.

On the other hand, as follows from Theorem \ref{Alice1}, in order for Bob to obtain a required degree of unforgeability, he must guarantee his experimental parameters $P_{\text{noqub}}$, $\theta$, $\beta_{\text{PB}}$ and $\beta_{\text{PS}}$ to be good enough. This is independent of any experimental parameters of Alice. Additionally, Bob must choose a suitable mathematical variable $\nu_{\text{unf}}$ to compute a guaranteed degree of unforgeability, as given by the bound of Theorem \ref{Alice1}.

Furthermore, as follows from Lemma \ref{Bob1}, in order for Alice to obtain a required degree of privacy, she must guarantee her experimental parameter $\beta_{\text{E}}$ to be small enough.

The parameters $N$, $\beta_{\text{PB}}$, $\gamma_\text{det}$ and $\gamma_\text{err}$ agreed by Alice and Bob must be good enough to achieve their required degrees of robustness, correctness and unforgeability. But they must also be achievable given their experimental setting.

\subsection{Extension of $\mathcal{QT}_1$ and $\mathcal{QT}_2$ to $2^M$ presentation points}

Extensions of the quantum token schemes $\mathcal{QT}_1$ and $\mathcal{QT}_2$ to $2^M$ presentation points, for any integer $M\geq 1$, and the proof of the following theorem are given in Appendix \ref{manyapplast}. 

\begin{theorem}
\label{manypoints}
For any integer $M\geq 1$, there exist quantum token schemes $\mathcal{QT}_1^M$ and $\mathcal{QT}_2^M$ extending $\mathcal{QT}_1$ and $\mathcal{QT}_2$ to $2^M$ presentation points, in which Bob sends Alice $NM$ quantum states, satisfying instant validation and the following properties. Consider parameters $\beta_\text{PB}$, $\beta_\text{PS}$,  $\beta_{\text{E}}$, $P_{\text{det}}$, $P_\text{noqub}$, $E$ and $\theta$ satisfying the constraints (\ref{rob1}), (\ref{cor1}), (\ref{Al1}) of Lemmas \ref{robust1} and \ref{correct1} and Theorem \ref{Alice1}, for which the function $f(\!\gamma_\text{err},\!\beta_\text{PS},\!\beta_\text{PB},\!\theta,\!\nu_\text{unf},\!\gamma_\text{det})$ defined by (\ref{Al3}) is positive. For these parameters, $\mathcal{QT}_1^M$ and $\mathcal{QT}_2^M$ are $\epsilon_{\text{rob}}^{M}-$robust, $\epsilon_{\text{cor}}^M-$correct, $\epsilon_{\text{priv}}^M-$private and $\epsilon_{\text{unf}}^M-$unforgeable with 
\begin{eqnarray}
\epsilon_{\text{rob}}^M&=&M\epsilon_{\text{rob}},\nonumber\\
\epsilon_{\text{cor}}^M&=&M\epsilon_{\text{cor}},\nonumber\\
\epsilon_{\text{priv}}^M&=&\frac{1}{2^M}\bigl[(1+2\epsilon_\text{priv})^M-1\bigr],\nonumber\\
\epsilon_{\text{unf}}^M&=&C\epsilon_{\text{unf}},
\end{eqnarray}
where $C$ is the number of pairs of spacelike separated presentation points, and where $\epsilon_{\text{rob}}$, $\epsilon_{\text{cor}}$, $\epsilon_{\text{priv}}$ and $\epsilon_{\text{unf}}$ are given by (\ref{rob2}), (\ref{cor2}), (\ref{Bo1}) and (\ref{Al5}).
\end{theorem}




\subsection{Quantum experimental tests}
\label{section4}

We performed experimental tests for the quantum stage of the $\mathcal{QT}_1$ scheme for the case of two presentation points ($M=1$), using the photonic setup of Fig. \ref{setup} and reporting strategy 1 (see Methods for details). Using a photon source with Poissonian distribution of average photon number $\mu=0.09$, and at repetition rate of 10 MHz, we generated a token of $N=4\times 10^7$ photon pulses, with detection efficiency of $\eta=0.21$, detection probability of $P_\text{det}=0.019$ and error rate of $E = 0.058$. We obtained deviations from the random distributions for the basis and state generation of $\beta_\text{PB}=2.4\times 10^{-3}$ and $\beta_\text{PS}=3.6\times 10^{-3}$, respectively. In order to guarantee unforgeability using Theorem \ref{Alice1} we need to improve some experimental parameters (see Fig. \ref{fig2}).

\begin{figure}
\includegraphics[scale=0.80]{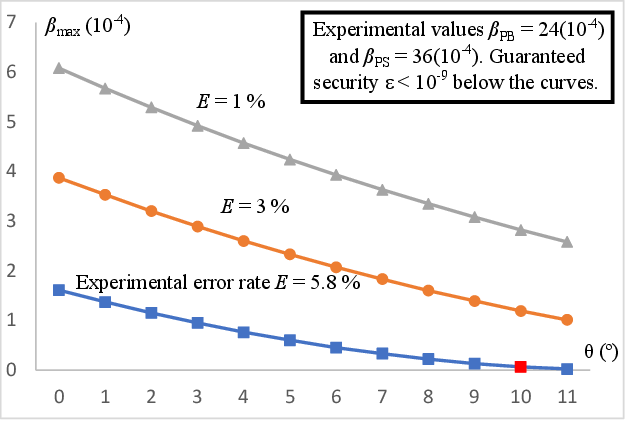}
\caption{\label{fig2}\textbf{Numerical example.} The plots denote the maximum value $\beta_{\text{max}}$ for $\beta_\text{PB}$ and $\beta_\text{PS}$ that our bounds can allow to guarantee correctness, robustness and unforgeability  simultaneously in a numerical example with the allowed experimental imperfections 
of Table \ref{tableimp} and under the assumptions 
of Table \ref{tableassu} for our quantum token schemes $\mathcal{QT}_1$ and $\mathcal{QT}_2$. The region below the plotted curves denote the secure region in which we have set $\epsilon_{\text{rob}}=\epsilon_{\text{cor}}=\epsilon_{\text{unf}}=10^{-9}$ in Lemmas \ref{robust1} and \ref{correct1} and in Theorem \ref{Alice1}. The plotted values keep all parameters fixed to the experimental values reported above, except for the deviations from the random distributions for basis and state generation, $\beta_\text{PB}$ and $\beta_\text{PS}$, the uncertainty $\theta$ on the Bloch sphere in the state generation, and the error rate  $E$. The blue curve denotes the values obtained for the experimentally obtained value $E=0.058$. The red square denotes the assumed upper bound for our experimental values of $\theta\leq 10^{\circ}$, and corresponds to a value of $\beta_{\text{max}}=6\times10^{-6}$, which is about 400 and 600 times smaller than the obtained experimental values of $\beta_\text{PB}=2.4\times10^{-3}$ and $\beta_\text{PS}=3.6\times 10^{-3}$, respectively. The orange and gray curves plot the values of $\beta_{\text{max}}$ assuming $E=0.03$ and $E=0.01$, respectively. In an ideal case in which $\theta=0^{\circ}$ and $E=0.01$, the value for $\beta_{\text{max}}$ would be approximately $6\times 10^{-4}$, which is about four and six times smaller than our obtained experimental values for $\beta_\text{PB}$ and $\beta_\text{PS}$, respectively. In a more realistic case, with $\theta=5^{\circ}$  and $E=0.03$, our numerical example gives approximately $\beta_{\text{max}}=2.3\times 10^{-4}$; meaning that with these experimental values, by reducing our obtained experimental values for $\beta_\text{PB}$ and $\beta_\text{PS}$ by respective factors of approximately 
10 and 16, we could guarantee correctness, robustness and unforgeability  simultaneously in our schemes.}
\end{figure}

Guaranteeing privacy in our schemes $\mathcal{QT}_1$ and $\mathcal{QT}_2$ can be satisfied with good enough random number generators, as follows from Lemma \ref{Bob1}. Due to the piling-up lemma, by using a large number of close-to-random bits, we can guarantee $\epsilon_{\text{priv}}$ to be arbitrarily small in practice.

\section{Discussion}
\label{section5}

We have presented two quantum token schemes that do not require either quantum state storage or long distance quantum communication, and are practical with current technology.
 Our schemes allow for losses, errors in the state preparations and measurements, and deviations from random distributions; and, in photonic setups, photon sources that do not emit exactly single photons, and threshold single photon detectors with non-unit detection efficiencies and with non-zero dark count probabilities 
(see Table \ref{tableimp}).

Our analyses follow much of the literature on practical mistrustful quantum cryptography (e.g. \cite{NJMKW12,LKBHTKGWZ13etal,LCCLWCLLSLZZCPZCP14etal,PJLCLTKD14etal,ENGLWW14}) in making the assumptions
of Table \ref{tableassu}. Under these assumptions, we have shown that there exist attainable experimental parameters for which our schemes can satisfy instant validation, correctness, robustness, unforgeability and user privacy. Importantly, Theorem \ref{manypoints} shows that this holds, in principle, for $2^M$ presentation points with arbitrary $M$. As in the schemes of Ref. \cite{KPG20}, our schemes allow the user to choose her presentation point $Q_b$ after her quantum measurements are completed, as long as she chooses $Q_b$ within the intersection of the causal past of all the presentation points. This means that the quantum communication stage of our schemes can take an arbitrarily long time and  can be implemented arbitrarily in the past of the presentation points, which is very convenient for practical implementations.

\damiannn{We note that the security of our quantum token schemes does not rely on any spacetime constraints. In principle, all presentation points could be timelike separated, for example. However, as discussed in the introduction, in order for our quantum token schemes to have an advantage over purely classical schemes, some spacetime presentation points need to be spacelike separated.

In practice, this means that some classical processing and classical communication steps in our schemes must be implemented sufficiently fast. This is in general feasible with current technology (for example, using field programmable gate arrays), if the presentation points are sufficiently far apart,  as demonstrated by previous implementations of relativistic cryptographic protocols \cite{LKBHTKGWZ13etal,LCCLWCLLSLZZCPZCP14etal,LKBHTWZ15etal,VMHBBZ16,ABCDHSZ21etal}. Furthermore, Alice's and Bob's laboratories must be synchronized securely to a common reference frame with sufficiently high time precision. This can be implemented using GPS devices and atomic clocks \cite{LKBHTKGWZ13etal,LCCLWCLLSLZZCPZCP14etal,LKBHTWZ15etal,VMHBBZ16,ABCDHSZ21etal}, for example. A detailed analysis of these experimental challenges is left for future work.

Using quantum key distribution for secure communications in our quantum token schemes can be useful, but it is not crucial. As discussed, Alice's and Bob's agents must communicate via secure and authenticated classical channels, which can be implemented with previously distributed secret keys. In an ideal situation where Alice's and Bob's agents have access to enough quantum channels, for example in a quantum network \cite{ECPPSY05,S17,Liaoetal18etal,DWTSTLPYDCTEGWP19etal} or in the envisaged quantum internet \cite{K08,WER18}, these keys can be expanded securely with quantum key distribution \cite{BB84,FLDCTPSYS17etal,Liaoetal17etal}. However, it is also possible to distribute the secret keys via secure physical transportation, as implemented in previous demonstrations of relativistic quantum cryptography \cite{LKBHTKGWZ13etal,LCCLWCLLSLZZCPZCP14etal,LKBHTWZ15etal,VMHBBZ16,ABCDHSZ21etal}.}

We note that in our proof of unforgeability, our only
 potential restriction on the technology and capabilities of dishonest
Alice is indirectly made through assumption G in photonic setups (see Table \ref{tableassu}), in the case where Bob's photon source does not perfectly conceal phase information.
In fact, we believe that assumptions A, B and G can be significantly
weakened. Investigating unforgeability for realistic weaker
forms of these assumptions is left as an open problem.

We implemented experimental tests of the quantum part of our scheme ($\mathcal{QT}_1$) using a free space optical setup \cite{DGHMR06,DLthesis} for quantum key distribution (QKD) that was slightly adapted for our scheme, and which can operate at daylight conditions. Importantly, Bob's transmission device is small, hand-held and low cost. These type of QKD setups are designed for future daily-life applications, for example with mobile devices (see e.g. \cite{MVRCCOW16etal,MFLVRSRW17etal,CCFCBGCNWOB17etal}).

Experiments with our relatively low precision devices do not guarantee unforgeability, but show it can be guaranteed with refinements. Crucial experimental parameters that we need to improve to achieve this are the deviations $\beta_\text{PB}$ and $\beta_\text{PS}$ from random basis and state generation, respectively. In our tests we obtained $\beta_\text{PB} =2.4\times 10^{-3}$ and $\beta_\text{PS} =3.6\times 10^{-3}$. An implementation in which the uncertainty in basis choices was bounded by $\theta=5^{\circ}$ and the error rate by $E=0.03$ would guarantee unforgeability if $\beta_\text{PB} \approx \beta_\text{PS} \approx 2.3\times 10^{-4}$ (about a factor of 10 and 16 lower than our values). This highlights that it is crucial to consider the parameters $\beta_\text{PS}$ and $\beta_\text{PB}$ in practical security proofs. 
For example, if we simply assumed $\beta_\text{PS}=\beta_\text{PB}=0$ as our experimental values then our results would imply that we had attained unforgeability, even for $\theta=10^{\circ}$ (see Fig. \ref{fig2}).
Taking $\beta_\text{PS}=\beta_\text{PB}= \theta = 0$, as implicitly assumed in some previous analyses of practical mistrustful quantum cryptography (e.g. \cite{LKBHTKGWZ13etal,LCCLWCLLSLZZCPZCP14etal,ENGLWW14}), is unsafe.

User privacy can also be guaranteed by using good enough random number generators. However, further security issues arise from the assumptions that
Bob cannot use degrees of freedom not previously agreed for the transmission of the quantum states to affect, or obtain information about, the statistics of Alice's quantum measurement devices; and, in photonic setups, that Alice's single photon detectors have equal efficiencies and equal dark count probabilities
(assumptions E and F in Table \ref{tableassu}). These issues are not specific to our implementations or to quantum token schemes: they arise quite generally in practical mistrustful quantum cryptographic schemes
in which one party measures states sent by the other. The attacks they allow, and defences against these (such as requiring single photon sources and 
using attenuators to equalize detector efficiencies) are analysed in detail elsewhere \cite{BCDKPG21}. As noted in Ref. \cite{BCDKPG21}, further options, such as iterating the scheme and using the XOR of the bits generated, also merit investigation. Importantly, our analyses here take into account multi-photon attacks \cite{BCDKPG21} in photonic setups, and the reporting strategies we have considered offer perfect protection against arbitrary multi-photon attacks, given our assumptions (see \damian{Fig. \ref{setup}}, and Lemma \ref{lemmax} in Methods).

In conclusion, our theoretical and experimental results give a proof of principle that quantum token schemes are implementable with current technology, and
that, conditioned on standard technological assumptions, security can be maintained in the presence of the various experimental imperfections 
we have considered (see Table \ref{tableimp}). As with other practical implementations of mistrustful quantum cryptography (and indeed quantum key distribution), completely unconditional security would require
defences against every possible collection of physical systems Bob might transmit to Alice, including programmed nano-robots that could enter and
reconfigure her laboratory \cite{LC99}. Attaining this is beyond current technology, but such far-fetched possibilities also illustrate that security based on suitable technological
assumptions (which may depend on context) may suffice for practical purposes.     More work on attacks and defences in practical mistrustful quantum cryptography is
undoubtedly needed to reach consensus on trustworthy technologies.   That said, as our schemes are built on simple mistrustful cryptographic primitives, 
we expect they can be refined to incorporate any agreed practical defences \cite{BCDKPG21}.

\section{Methods}

\subsection{Protection against multi-photon attacks in photonic implementations}

The following lemma is a straightforward extension of Lemmas 1 and 12 of Ref. \cite{BCDKPG21} to the case of $N>1$ transmitted photon pulses. Note that Alice (Bob) in our notation refers to Bob (Alice) in the notation of Ref. \cite{BCDKPG21}. The proof is given in Appendix \ref{appnewlemma}.

\begin{lemma}
\label{lemmax}
Suppose that Bob sends Alice $N$ photon pulses, labelled by $k\in[N]$. Let the $k$th pulse have $L_k$ photons. Let $\rho$ be an arbitrary quantum state prepared by Bob in the polarization degrees of freedom of the photons sent to Alice, which can be arbitrarily entangled among all photons in all pulses and can also be arbitrarily entangled with an ancilla held by Bob. Let $\mathcal{D}_0$ and $\mathcal{D}_1$ be two arbitrary qubit orthogonal bases. Suppose that either Alice uses the setup of Fig. \ref{setup} with reporting strategy 1 to implement the quantum token scheme $\mathcal{QT}_1$ (see Table \ref{real1}), or Alice uses the setup of \damian{Fig. \ref{setup}} with reporting strategy 2 to implement the quantum token scheme $\mathcal{QT}_2$ (see Table \ref{real2}). Suppose also that assumptions E and F (see Table \ref{tableassu}) hold.
For $k\in[N]$, let $m_k=1$ if Alice assigns a successful measurement to the $k$th pulse and $m_k=0$ otherwise; let $w_k=0$ ($w_k=1$) if Alice assigns a measurement basis to the $k$th pulse in the basis $\mathcal{D}_{0}$ ($\mathcal{D}_{1}$). If Alice uses the setup of Fig. \ref{setup} and reporting strategy 1 to implement the scheme $\mathcal{QT}_1$, without loss of generality, suppose also that Alice sets $w_k=0$ with unit probability, if $m_k=0$, for $k\in[N]$. Let $m=(m_1,\ldots,m_N)$, $w=(w_1,\ldots,w_N)$ and $L=(L_1,\ldots,L_N)$.

If Alice uses the setup of Fig. \ref{setup} with reporting strategy 1 to implement the scheme $\mathcal{QT}_1$, then the probability that Alice reports the string  $m$ to Bob and assigns the string of measurement bases $w$, given $\rho$  and $L$, is
\begin{equation}
\label{hhh}
P^{(1)}_{\text{rep}}(m,w\lvert \rho, L)=\prod_{k=1}^N G^{(1)}_{m_k,w_k}(d_0,d_1,\eta,L_k),
\end{equation}
where
\begin{eqnarray}
\label{hhh0}
G^{(1)}_{1,b}(d_0,d_1,\eta,a)&=&(1-d_0)(1-d_1)\Bigl(1-\frac{\eta}{2}\Bigr)^{a}\nonumber\\
&&\qquad-(1-d_0)^2(1-d_1)^2(1-\eta)^{a},\nonumber\\
G^{(1)}_{0,0}(d_0,d_1,\eta,a)&=&1-2G^{(1)}_{1,0}(d_0,d_1,\eta,a),\nonumber\\
G^{(1)}_{0,1}(d_0,d_1,\eta,a)&=&0,
\end{eqnarray}
for $b\in\{0,1\}$, $m,w\in\{0,1\}^N$ and $a,L_1,\ldots,L_N\in\{0,1,2,\ldots\}$. Furthermore, the probability $P_\text{MB}(w_k)$ that Alice assigns a measurement in the basis $\mathcal{D}_{w_k}$, conditioned on the value $m_k=1$, for the $k$th pulse, satisfies  
\begin{equation}
\label{h000.1}
P_\text{MB}(w_k)=\frac{1}{2},
\end{equation}
for $w_k\in\{0,1\}$ and $k\in[N]$.

If Alice uses the setup of \damian{Fig. \ref{setup}} with reporting strategy 2 to implement the scheme $\mathcal{QT}_2$, then the probability that Alice reports the string  $m$ to Bob, given $\rho$, $w$ and $L$, is
\begin{equation}
\label{ggg}
P^{(2)}_{\text{rep}}(m\lvert w, \rho, L)=\prod_{k=1}^N G^{(2)}_{m_k}(d_0,d_1,\eta,L_k),
\end{equation}
where
\begin{eqnarray}
\label{ggg0}
G^{(2)}_0(d_0,d_1,\eta,a)&=&(1-d_0)(1-d_1)(1-\eta)^{a},\nonumber\\
G^{(2)}_1(d_0,d_1,\eta,a)&=&1-(1-d_0)(1-d_1)(1-\eta)^{a},
\end{eqnarray}
for
$m,w\in\{0,1\}^N$ and $a,L_1,\ldots,L_N\in\{0,1,2,\ldots\}$.

In any of the two cases, the message $m$ gives Bob no information about the bit entries $w_k$ for which $m_k=1$. Equivalently, the set $\Lambda\subset[N]$ of labels transmitted to Bob in step 2 of $\mathcal{QT}_1$ and $\mathcal{QT}_2$ gives Bob no information about the string $W$ and the bit $z$.

\end{lemma}

\subsection{Clarification about unforgeability in photonic implementations}

A subtle technical issue when implementing our quantum token schemes with photonic setups
is that in our schemes we have assumed the quantum systems $A_k$ that Bob transmits to Alice to have finite Hilbert space dimension, for $k\in[N]$. However, some light sources, like weak coherent sources, or other photon sources with Poissonian statistics,
can emit pulses with a number of photons $J$, where $J$ can tend to infinity, although with a probability tending to zero. This issue is easily solved by fixing a maximum number of photons $J_{\text{max}}$ and assuming that unforgeability is not guaranteed whenever Bob's photon source emits a pulse with more than $J_{\text{max}}$ photons. By fixing $J_{\text{max}}$ to be arbitrarily large, but finite, the probability that among the $N$ emitted pulses there is at least one pulse with more than $J_{\text{max}}$ photons can be made arbitrarily small. Thus, with probability arbitrarily close to unity, honest Bob is guaranteed that each of his $N$ emitted pulses does not have more than $J_{\text{max}}$ photons, i.e., the internal degrees of freedom -- like the polarization degrees of freedom -- of each pulse, represented by the quantum system $A_k$, have a finite Hilbert space dimension.

\subsection{Experimental setup}

Our experimental setup is based on a free space optical quantum key distribution (QKD) system, which can operate at daylight conditions. This setup was developed by one of us (D. L.) during his PhD \cite{DLthesis}, based upon the work of Ref. \cite{DGHMR06}.
The main features of our experimental setup are illustrated in Figs. \ref{setup} and \ref{fig1}.

\begin{figure*}
\includegraphics[scale=0.204]{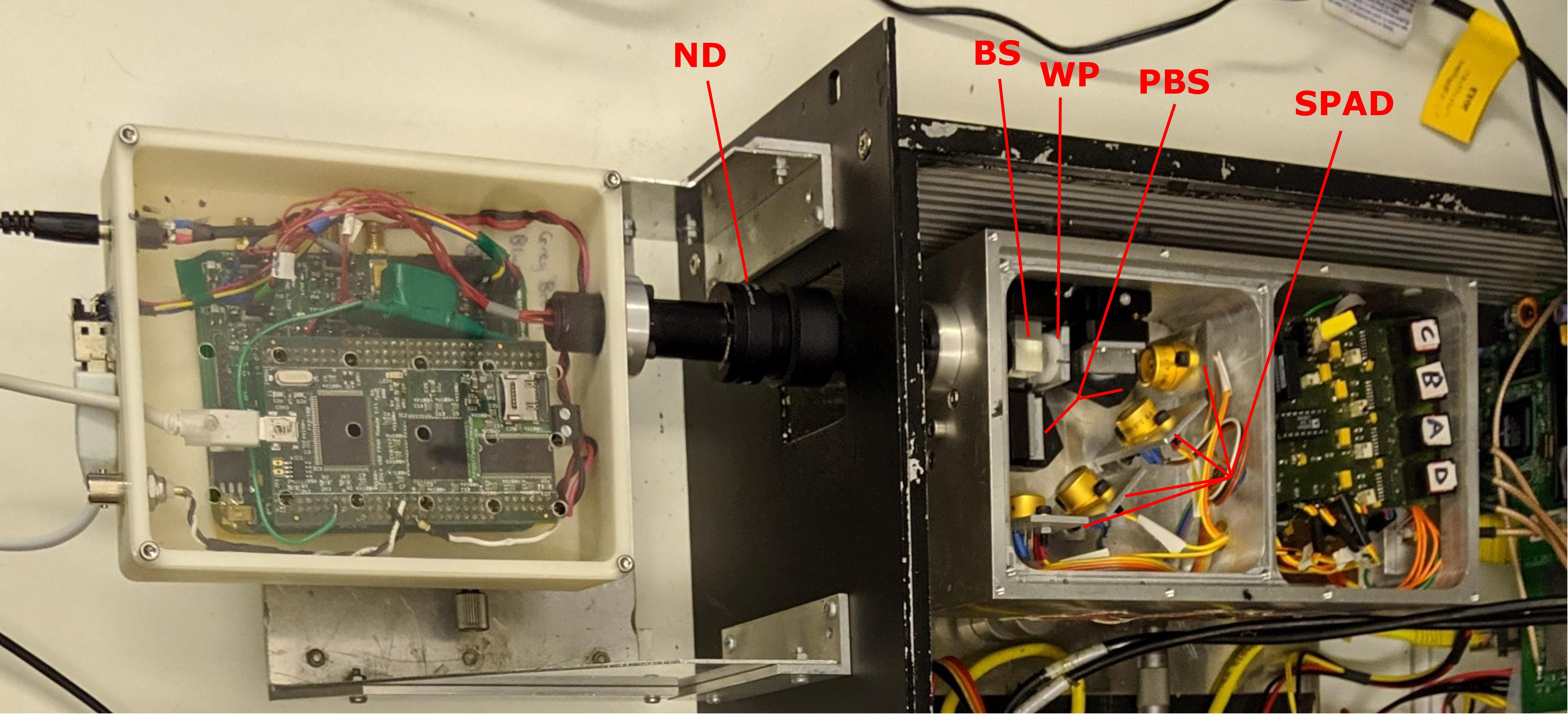}
\caption{\label{fig1} \textbf{Photograph of the experimental setup.} Bob's quantum transmitter (white box in the left) is a small and low-cost hand-held device of approximately $20\times 15\times 5$ cm. Alice's quantum receiver is contained within a box of approximately $20\times 12\times 5$ cm (grey box on the right), with further electronics contained within another box of approximately $30\times 50\times 15$ cm (bigger black box). At Bob's site, the QKD transmitter comprises a field programmable gate array (FPGA) which pulses 4 LEDs, each polarized in one of the horizontal (H), vertical (V), diagonal (D), and anti-diagonal (A) states, corresponding to the $\lvert 0\rangle$, $\lvert 1\rangle$, $\lvert +\rangle$ and $\lvert -\rangle$ BB84 states, respectively. The light from the LEDs is collimated by a diffraction grating and pinholes. The statistics of Bob's photon source is assumed Poissonian \cite{HPLG07,DLthesis}. Neutral-density (ND) filters (small black cylinders) are used to attenuate the pulses down to the required mean photon number, which in our experiment was $\mu=0.09$. Since Bob's photon source consists of LEDs, and LEDs have low spatio-temporal coherence \cite{MSDS10}, no phase randomization is required to satisfy assumption G to a good approximation (see Table \ref{tableassu}).
At Alice's site, the received light pulses from Bob are focused from the transmitter pinhole, through a 50:50 beam splitter (\damiannn{BS}, small transparent cube) which performs basis selection, and wave plate (\damiannn{WP}, thin white cylinder) and polarizing beam splitters (\damiannn{PBS}, small black boxes) which perform the measurement of the polarization. The photons are detected with single photon avalanche diodes (\damiannn{SPAD,} small golden cylinders), which are threshold single photon detectors with efficiency $\eta=0.21$, including the quantum efficiency of the detectors and the transmission efficiency from Bob's setup to the detectors.  An FPGA time tags the detections with 52 bit precision, equivalent to 30.5 ps \cite{NDR11}, and sends them to a PC for processing. Alice's gray and black boxes are closed during operation to decrease noise due to environment light. But they are shown open here for illustration.}
\end{figure*}

Only minor changes to our quantum setup are needed to implement the quantum stage of $\mathcal{QT}_2$. For example, the 50:50 beam splitter in Alice's site can be replaced by a suitably placed mirror directing the received photon pulses to one of the two polarizing beam splitters. This mirror can be set in a movable arm, which positions the mirror in place if $z$ has a specific value (e.g. if $z=1$) and out of place, letting the photon pulses reach the other polarizing beam splitter, if $z$ takes the other value (e.g. $z=0$). The movable arm putting the mirror in place or out of place does not need to move very fast, as it remains in the same position during the transmission of all $N$ pulses from Bob in the case of two presentation points, or during the transmission of each set of 
$N$ pulses from the total of $NM$ in the case of $2^M$ presentation points (see quantum token scheme $\mathcal{QT}_2^M$ in Appendix \ref{manyapplast}).

\subsection{Experimental tests and numerical example}
The quantum stage of the token scheme $\mathcal{QT}_1$ was implemented with the experimental setup described above. Below we describe our experiment and the numerical example of Fig. \ref{fig2}. Unless we consider it necessary or helpful, all values smaller than unity obtained in our experiment and numerical example are given below rounded to two significant figures.

As we explain below, our obtained experimental values for the parameters in Lemmas \ref{robust1} and \ref{correct1} and in Theorem \ref{Alice1} are  $N=4\times 10 ^7$, $P_\text{det}=0.019$, $E=0.058$, $\beta_\text{PB}=2.4\times10^{-3}$, $\beta_\text{PS}=3.6\times10^{-3}$, $P_\text{noqub}=3.8\times10^{-3}$. We assume an angle $\theta\leq 10^{\circ}$ in our experiment.

In the numerical example of Fig. \ref{fig2} we used the previous experimental values, except for $\theta$ and $E$, which were varied as shown in the plots, and for $\beta_\text{PB}$ and $\beta_\text{PS}$. In the plots of Fig. \ref{fig2}, if $\beta_\text{PB}\leq \beta_\text{max}$ and $\beta_\text{PS}\leq \beta_\text{max}$ hold, then we obtain from Lemmas \ref{robust1} and \ref{correct1} and from Theorem \ref{Alice1} that $\epsilon_\text{rob},\epsilon_\text{cor},\epsilon_\text{unf}\leq 10^{-9}$. We do not claim that our numerical example is optimal. In other words, we do not claim that with our experimental parameters every point above the curves of Fig. \ref{fig2} is insecure, in the sense that the conditions $\epsilon_\text{rob},\epsilon_\text{cor},\epsilon_\text{unf}\leq 10^{-9}$ do not hold. Our claim is only that given our experimental parameters, the regions of points below the curves of Fig. \ref{fig2} satisfy the conditions $\epsilon_\text{rob},\epsilon_\text{cor},\epsilon_\text{unf}\leq 10^{-9}$.

For the three curves of Fig. \ref{fig2} we set $\gamma_\text{det}=0.018$. Thus, the condition (\ref{rob1}) of Lemma \ref{robust1} is satisfied, and from (\ref{rob2}), we have $\epsilon_\text{rob}=e^{-1052}<10^{-9}$.

For the three curves of Fig. \ref{fig2} we set $\nu_\text{unf}=3.9\times10^{-3}$. This is the minimum value for which the first term of $\epsilon_\text{unf}$ in (\ref{Al5}) equals $\frac{10^{-9}}{2}$. This is because, as we describe below, we also chose the parameters satisfying that the second term of $\epsilon_\text{unf}$ in (\ref{Al5}) equals $\frac{10^{-9}}{2}$, from which we have $\epsilon_\text{unf}=10^{-9}$. We recognize that although this particular choice seems natural, it probably does not optimize our results.

Then, for each of the three considered error rates $E=0.01$, $E=0.03$ and $E=0.058$, and for each of the angles $\theta=0^\circ,1^\circ,\ldots, 11^\circ$, we set $\beta_{\text{PB}}=\beta_{\text{PS}}=\beta_\text{max}$ and varied $\beta_\text{max}$, $\nu_{\text{cor}}$ and $\gamma_\text{err}$ trying to find the maximum value of $\beta_{\text{max}}$ for which both terms of $\epsilon_\text{cor}$ in (\ref{cor2}) and the second term of $\epsilon_\text{unf}$ in (\ref{Al5}) were as close as possible to $\frac{10^{-9}}{2}$, but not bigger than $\frac{10^{-9}}{2}$, while guaranteeing that the constraints (\ref{cor1}) and (\ref{Al1}) were satisfied. Our results for $\beta_\text{max}$ are plotted in Fig. \ref{fig2}.

We describe how we obtained the experimental parameters presented above. At a repetition rate of $10$ MHz, Bob transmitted photon pulses to Alice during four seconds. Thus, the number of transmitted pulses was $N=4\times 10 ^7$.

Since the photon statistics of Bob's source is assumed Poissonian \cite{HPLG07,DLthesis}, the probability that a photon pulse has two or more photons is $P_\text{noqub}=1-(1+\mu)e^{-\mu}$. Since in our experiment $\mu=0.09$, we obtain $P_\text{noqub}=3.8\times10^{-3}$.

As discussed below, Alice assigned successful measurements using reporting strategy 1. The number of pulses for which Alice assigned successful measurement was $n=742491$. The obtained estimation for the probability $P_\text{det}$, was obtained as $P_\text{det}=\frac{n}{N}=0.019$.

The measured detection efficiency, including the quantum efficiency of the detectors and the transmission probability from Bob's setup to the detectors, was $\eta=0.21$. We note that our obtained value of $P_\text{det}=0.01856$, which we reported above with the less precise value $P_\text{det}=0.019$, is a good approximation to the theoretical prediction in which the photon statistics of Bob's source follow a Poisson distribution with average photon number $\mu=0.09$, Alice uses reporting strategy 1 with her four detectors having the same efficiency $\eta=0.21$, and the dark count probabilities are assumed to be zero.
As follows from (\ref{hhh}) and (\ref{hhh0}) in Lemma \ref{lemmax},
this theoretical prediction for $P_\text{det}$ is given by
\begin{eqnarray}
\label{theoretical}
P_\text{det}^\text{theo}&=&\sum_{k=0}^\infty \frac{e^{-\mu}\mu^k}{k!}\bigl(G^{(1)}_{1,0}(0,0,\eta,k)+G^{(1)}_{1,1}(0,0,\eta,k)\bigr)\nonumber\\
&=&2\sum_{k=0}^\infty \frac{e^{-\mu}\mu^k}{k!}\Bigl[\Bigl(1-\frac{\eta}{2}\Bigr)^k-(1-\eta)^k\Bigr]\nonumber\\
&=&2\bigl(e^{-\frac{\mu\eta}{2}}-e^{-\mu\eta}\bigr)\nonumber\\
&=&0.01863,
\end{eqnarray}
where in the last line we used our experimental parameters $\mu=0.09$ and $\eta=0.21$. This gives a ratio $\frac{P_\text{det}}{P_\text{det}^{\text{theo}}}=0.996$.

As mentioned in Fig. \ref{setup}, Alice applies reporting strategy 1, in order to protect against multi-photon attacks \cite{BCDKPG21} (see Lemma \ref{lemmax}). That is, Alice assigns successful measurement outcomes in the basis $\mathcal{D}_0$ ($\mathcal{D}_1$) with unit probability for the pulses in which at least one of the detectors D$_0$ and D$_1$ (D$_+$ and D$_-$) click and D$_+$ and D$_-$ (D$_0$ and D$_1$) do not click. It is clear that when only the detector D$_i$ clicks, Alice associates the measurement outcome to the BB84 state $\lvert i\rangle$, for $i\in\{0,1,+,-\}$. However, it is not clear how Alice should assign measurement outcomes to the cases in which both D$_0$ and D$_1$ (D$_+$ and D$_-$) click and D$_+$ and D$_-$ (D$_0$ and D$_1$) do not click. The results of Lemma \ref{lemmax} are independent of how Alice assigns these outcomes. In order to make clear this generality of the results of Lemma \ref{lemmax}, we have not included how these outcomes are assigned by Alice in the definition of reporting strategy 1 in Fig. \ref{setup}. Nevertheless, how these outcomes are assigned by Alice plays a role in the error rate $E$, and thus also in the degrees of correctness and unforgeability that can be guaranteed (see Lemma \ref{correct1} and Theorem \ref{Alice1}). In our experiment, Alice assigns a random measurement outcome associated to the state $\lvert 0\rangle$ and $\lvert 1\rangle$ ($\lvert +\rangle$ and $\lvert -\rangle$) when both D$_0$ and D$_1$ (D$_+$ and D$_-$) click and D$_+$ and D$_-$ (D$_0$ and D$_1$) do not click.

As mentioned above, in our experiment we obtained Alice's error rate $E=0.058$, and deviations from the random distributions for basis and state generation by Bob of $\beta_\text{PB}=2.4\times10^{-3}$ and $\beta_\text{PS}=3.6\times10^{-3}$, respectively. These values were computed as we describe below.

\subsection{Statistical information}
In our experimental tests, the number of photon pulses transmitted from Bob to Alice
was $N=4\times 10 ^7$. The number of pulses for which Alice assigned successful measurement was $n=742491$. The obtained estimation for the probability $P_\text{det}$, was obtained as $P_\text{det}=\frac{n}{N}=0.019$.

The error rate $E=0.058$ was computed as follows. From the $n$ pulses that Alice assigned as successful measurements,
$n_{tu}^{\text{same}}$ pulses were prepared by Bob with polarization given by the qubit state $\lvert\phi_{tu}\rangle$ and were measured by Alice in the same basis of preparation by Bob ($\mathcal{D}_u=\{\lvert \phi_{tu}\rangle\}_{t=0}^1$), from which $n_{tu}^{\text{error}}$ gave Alice the outcome opposite to the state prepared by Bob, i.e. an error, for $t,u\in\{0,1\}$. We computed $E_{tu}=\frac{n_{tu}^{\text{error}}}{n_{tu}^{\text{same}}}$, for $t,u\in\{0,1\}$. The estimation for the error rate $E$ was taken as $E=\max\{E_{00},E_{10},E_{01},E_{11}\}$. We obtained $n_{00}^{\text{same}}=80786$, $n_{10}^{\text{same}}=121159$, $n_{01}^{\text{same}}=93618$, $n_{11}^{\text{same}}=80653$, $n_{00}^{\text{error}}=4725$, $n_{10}^{\text{error}}=2250$, $n_{01}^{\text{error}}=1602$ and $n_{11}^{\text{error}}=3851$. From these values, we obtained $E_{00}=0.058$, $E_{10}=0.019$, $E_{01}=0.017$, $E_{11}=0.048$ and $E=0.058$.

Our experimentally obtained estimations $\beta_\text{PB}=2.4\times10^{-3}$ and $\beta_\text{PS}=3.6\times10^{-3}$ were obtained from the number $n$ of pulses that Alice reported as successfully measured. We did not use the whole number $N$ of transmitted pulses for these estimations, because the software integrated in our experimental setup is configured to output data for the pulses that produce a detection event in at least one detector. From the $n=742491$ pulses reported above, Bob produced $n_{tu}$ pulses in the state $\lvert \phi_{tu}\rangle$, for $t,u\in\{0,1\}$. We note that $n=n_{00}+n_{10}+n_{01}+n_{11}$. We computed $\beta_{\text{PB}}=\bigl\lvert\frac{n_{00}+n_{10}}{n}-\frac{1}{2}\bigr\rvert$ and $\beta_{\text{PS}}=\max\bigl\{\bigl\lvert\frac{n_{00}}{n_{00}+n_{10}}-\frac{1}{2}\bigr\rvert,\bigl\lvert\frac{n_{01}}{n_{01}+n_{11}}-\frac{1}{2}\rvert\bigr\}$. We obtained $n_{00}=185166$, $n_{10}=187842$, $n_{01}=184251$, $n_{11}=185232$, $\beta_\text{PB}=2.4\times10^{-3}$ and $\beta_\text{PS}=3.6\times10^{-3}$.

\begin{acknowledgments}
The authors acknowledge financial support from the UK Quantum Communications Hub grants no. EP/M013472/1 and EP/T001011/1, and thank Siddarth Koduru Joshi for helpful conversations. A.K. and D.P.G. also thank Sarah Croke for helpful conversations. A.K. is partially supported  by Perimeter Institute for Theoretical Physics. Research at Perimeter Institute is supported by the Government of Canada through Industry Canada and by the Province of Ontario through the Ministry of Research and Innovation.

\textbf{Author contributions}

A.K and J.R. conceived the project. D.P.G. did the majority of the theoretical work, with input from A.K.
D.L. devised the experimental setup and took the experimental data. D.P.G. analysed the experimental data and did the numerical work. A.K. and D.P.G. wrote the manuscript with input from D.L.

\textbf{Competing interests:} 
A.K. jointly owns the patent ‘A. Kent, Quantum tokens, US Patent No. 10,790,972 (2020)’ and similar patents in
other jurisdictions, and has consulted for and owns shares in a corporate co-owner.

\textbf{Data availability:} the datasets generated and analysed during the current study are available from the corresponding author on reasonable request.


\textbf{Materials and correspondence:} correspondence and requests for materials should be addressed to D.P.G. (email: D.Pitalua-Garcia@damtp.cam.ac.uk).

  \end{acknowledgments}

\appendix

\section{Summary}

These appendices provide the mathematical proofs of the lemmas and theorems given in the main text, as well as some known mathematical results, and new lemmas and a new theorem derived here to prove these results. Appendix \ref{sec2} states some well known mathematical results that are used along this text. The rest of this text is divided in three parts. The first part comprises Appendices \ref{secnewlemma} and \ref{idealunfapp}. Appendix \ref{secnewlemma} provides Lemmas \ref{lemma0} and \ref{eigenvalue}, which are used in other appendices to prove various results. In Appendix \ref{idealunfapp}, Lemma \ref{lemma01} is proved from Lemma \ref{lemma0}. Thus, the first part, along with Appendix \ref{sec2} is all that the reader requires for the security proof in the case of two presentation points (the case $M=1$) in the ideal case where there are not any errors or losses, i.e. Lemma \ref{lemma01}. The second part comprises Appendices \ref{proofoflemmas}, \ref{appnewlemma} and \ref{newapp}, which give mathematical details for the case of two presentation points in the general case that there are errors, losses and other experimental imperfections. The third part comprises Appendix \ref{manyapplast}, which gives mathematical details for the case of $2^M$ presentation points, for any integer $M\geq1$, in the general case that there are errors, losses and other experimental imperfections. 

In the second part, Lemmas \ref{robust1}, \ref{correct1} and \ref{Bob1}, which indicate the robustness, correctness and privacy for the token schemes $\mathcal{QT}_1$ and $\mathcal{QT}_2$ given in the main text, are proved in Appendix \ref{proofoflemmas}. Appendix \ref{appnewlemma} proves Lemma \ref{lemmax}, showing that reporting strategies 1 and 2 with the photonic setup of \damian{Fig. \ref{setup}} guarantee perfect protection against arbitrary multi-photon  attacks \cite{BCDKPG21}, given assumptions E and F of Table \ref{tableassu}. Using Lemma \ref{lemma0}, Appendix \ref{newapp} proves Theorem \ref{Alice1}, which corresponds to unforgeability for the token schemes $\mathcal{QT}_1$ and $\mathcal{QT}_2$ given in the main text.

In the third part, Theorem \ref{manypoints} given in the main text is proved in Appendix \ref{manyapplast} in the following way. First, the quantum token scheme $\mathcal{QT}_a$ is extended to a quantum token scheme $\mathcal{QT}_a^M$ for the case of $2^M$ presentation points, for any integer $M\geq 1$, and for $a\in\{1,2\}$. 
Then, Lemmas \ref{lastrobustM}, \ref{lastcorrectM} and \ref{lastBobM} and Theorem \ref{lastAliceM}, which respectively state the robustness, correctness, privacy and unforgeability for the token schemes $\mathcal{QT}_1^M$ and $\mathcal{QT}_2^M$, are given and proved. 

\section{Mathematical preliminaries}
\label{sec2}
The following results are well known in the literature.  We state them here for completeness. In the following, $\lVert O \rVert$ denotes the Schatten $\infty-$norm of the linear operator $O$, which equals the greatest eigenvalue of $O$, if $O$ is a positive semi definite operator acting on a finite dimensional Hilbert space.

\begin{proposition}
\label{proposition3}
Let $X_1,X_2,\ldots,X_N$ be independent random variables taking values $X_k\in \{0,1\}$, for $k\in[N]$. Let $X=\sum_{k=1}^N {X}_k$, and let $E(X)$ be the average value of $X$. Two Chernoff bounds state that \cite{Mitzenmacherbook}
\begin{eqnarray}
\label{ly}
\text{Pr}[X\leq (1-\epsilon)E(X)]&\leq& e^{-\frac{E(X)}{2}\epsilon^2},\nonumber\\
\text{Pr}[X\geq (1+\epsilon)E(X)]&\leq& e^{-\frac{E(X)}{3}\epsilon^2},
\end{eqnarray}
for $0< \epsilon<1$.
\end{proposition}

\begin{proposition}
\label{proposition1}
For any quantum density matrix $\xi$ and any positive semi definite operator $O$ acting on a finite dimensional Hilbert space $\mathcal{H}$, we have
\begin{equation}
\label{eq0}
\text{Tr}(O \xi )\leq \lVert O\rVert.
\end{equation}
\end{proposition}

\begin{proof}
Since $O$ acts on a finite dimensional Hilbert space $\mathcal{H}$ and it is positive semi definite, it is also Hermitian, hence, from the  spectral theorem there exists an orthonormal basis $\{\lvert e_i\rangle\}_i$ of $\mathcal{H}$ which is an eigenbasis of $O$, with real eigenvalues $\{\mu_i\}_i$. Suppose that $\xi=\lvert \xi\rangle\langle \xi\rvert$ is pure, then we express it in the eigenbasis $\{\lvert e_i\rangle\}_i$ of $O$. We have $\lvert \xi\rangle=\sum_i\alpha_i \lvert e_i\rangle$, where $\sum_i\lvert \alpha_i\rvert^2=1$, hence, $\text{Tr}(O\xi)=\sum_i \mu_i\lvert\alpha_i \rvert^2\leq \lVert O\rVert$. If $\xi$ is not pure, it can be written as the convex combination of pure states, hence, applying the inequality to each of these pure states, the result follows.
\end{proof}

\begin{proposition}
\label{proposition2}
For any finite set of positive semi definite operators $\{D_a\}_{a\in \Omega}$ acting on a finite dimensional Hilbert space $\mathcal{H}_A$ and any projective measurement $\{\Pi_a\}_{a\in \Omega}$ acting on a finite dimensional Hilbert space $\mathcal{H}_B$, it holds that
\begin{equation}
\label{l1.9}
\biggl\lVert \sum_{a\in \Omega} (D_a)_A\otimes (\Pi_a)_B\biggr\rVert =\max_{a\in \Omega} \lVert D_a\rVert.
\end{equation}
\end{proposition}

\begin{proof}
Let $\lvert \psi\rangle$ be the eigenstate of $O=\sum_{a\in \Omega} (D_a)_A\otimes (\Pi_a)_B$ with the greatest eigenvalue, hence, $\lVert O\rVert=\langle \psi\rvert O\lvert \psi\rangle$. We can write $\lvert\psi\rangle=\sum_{a\in\Omega}\alpha_a \lvert \omega_a\rangle$, where $\sum_{a\in\Omega}\lvert \alpha_a\rvert^2=1$, and where $\bigl(\mathds{1}_A\otimes (\Pi_a)_B\bigr)\lvert\omega_b\rangle=\delta_{a,b}\lvert\omega_b\rangle$, for $a,b\in\Omega$. Thus, we obtain 
\begin{equation}
\label{l1.10}
\lVert O\rVert=\sum_{a\in\Omega} \lvert \alpha_a\rvert^2 \langle \omega_a\rvert  \bigl((D_a)_A\otimes \mathds{1}_B\bigr)\lvert\omega_a\rangle.
\end{equation}
We can then write $\lvert\omega_a\rangle=\sum_{i,j}\beta_{i,j}^a\lvert e_i^a\rangle\otimes\lvert j\rangle$, where $\{\lvert e_i^a\rangle\}_i$ is the eigenbasis  of $D_a$, with eigenvalues $\{\mu_i^a\}_i$, $\{\lvert j\rangle\}_j$ is an orthonormal basis of $\mathcal{H}_B$, and where $\sum_{i,j}\lvert \beta_{i,j}^a\rvert^2=1$, for $a\in\Omega$. From (\ref{l1.10}), we have
\begin{eqnarray}
\label{l1.11}
\lVert O\rVert&=&\sum_{a\in\Omega} \lvert \alpha_a\rvert^2\sum_{i,j} \lvert \beta_{i,j}^a\rvert^2 \mu_i^a\nonumber\\
&\leq&\sum_{a\in\Omega} \lvert \alpha_a\rvert^2\sum_{i,j} \lvert \beta_{i,j}^a\rvert^2\lVert D_a\rVert\nonumber\\
&=& \sum_{a\in\Omega} \lvert \alpha_a\rvert^2 \lVert D_a\rVert\nonumber\\
&\leq&\max_{a\in\Omega} \lVert D_a\rVert,
\end{eqnarray}
where in the second line we used that $\lVert D_a\rVert$ is the greatest of the eigenvalues $\{\mu_i^a\}_i$. Now let $\lvert \tau_b\rangle\in\mathcal{H}_A$ be an eigenstate of $D_{b}$ whose corresponding eigenvalue is $\lVert D_b\rVert$, i.e. the greatest eigenvalue of $D_b$, and where $\lVert D_{b}\rVert =\max_{a\in\Omega} \lVert D_a\rVert$,  for some $b\in\Omega$. Let $\lvert\chi_b\rangle\in\mathcal{H}_B$ be a pure state satisfying $\Pi_{a}\lvert \chi_b\rangle=\delta_{a,b}\lvert \chi_b\rangle$, for $a\in\Omega$. It can easily be verified that $\lvert \tau_b\rangle\otimes \lvert \chi_b\rangle$ is an eigenstate of $O$ with eigenvalue $\lVert D_{b}\rVert=\max_{a\in\Omega} \lVert D_a\rVert$, i.e. there is an eigenvalue of $O$ equal to $\max_{a\in\Omega} \lVert D_a\rVert$. Thus, since $O$ is positive semi definite, $\lVert O\rVert$ is the greatest eigenvalue of $O$, hence, the result follows from (\ref{l1.11}).
\end{proof}

\section{Useful mathematical results}
\label{secnewlemma}

In this appendix, we state and prove Lemmas \ref{lemma0} and \ref{eigenvalue}.
Lemma \ref{lemma0} extends results of Ref. \cite{LKBHTKGWZ13etal}, for example in allowing a small deviation from the random distribution, as characterized by the parameters $\beta_\text{PS}>0$ and $\beta_\text{PB}>0$. Lemma \ref{lemma0} is a central mathematical result that we use to prove Lemma \ref{lemma01} and Theorem \ref{Alice1} in Appendices \ref{idealunfapp} and \ref{newapp}, respectively. 
Lemma \ref{eigenvalue} states an upper bound on the maximum eigenvalue of a particular qubit density matrix. It will be useful in the proof of Theorem \ref{Alice1} in Appendix \ref{newapp}.

\begin{lemma}
	\label{lemma0}
	For $r,r',s\in\{0,1\}$ and $k\in[N]$, and for some $O\in\bigl[\frac{1}{\sqrt{2}},1\bigr)$, let $\lvert \phi_{r,s}^k\rangle$ be qubit pure states satisfying $\langle \phi_{0,s}^k\vert \phi_{1,s}^k\rangle=0$ and $\bigl\lvert \langle \phi_{r,0}^k\vert \phi_{r',1}^k\rangle\bigr\rvert \leq O $. Let $\mathbf{h}=(h_1,\ldots,h_N)$ be a $N-$bit string. For any $N-$bit string $\mathbf{x}=(x_1,\ldots,x_N)$ and for $i\in\{0,1\}$, let $\mathbf{x}_i$ denote the restriction of $\mathbf{x}$ to $S_i^\mathbf{h}=\bigl\{k\in[N]\vert s_k=h_k\oplus i\bigr\}$. Let $P_{\mathbf{s}}=\prod_{k=1}^N P_{s_k}^k$ be the probability distribution for $\mathbf{s}=(s_1,\ldots,s_N)\in\{0,1\}^N$, where $\{P_{0}^k,P_{1}^k\}$ is a binary probability distribution satisfying $P_0^{\text{LB}}\leq P_0^k\leq P_0^{\text{UB}}$, for $k\in[N]$. Let $d(\cdot,\cdot)$ denote the Hamming distance, and let $\bigl(\phi_{\mathbf{r},\mathbf{s}}\bigr)_B=\bigotimes _{k=1}^N \bigl(\lvert \phi_{r_k,s_k}^k\rangle\langle\phi_{r_k,s_k}^k\rvert\bigr)_{B_k}$, where $B=B_1B_2\cdots B_N$ denotes a quantum system of $N$ qubits. For $\mathbf{a},\mathbf{b}\in\{0,1\}^N$, we define
\begin{equation}
\label{l0.1}
\bigl(D_{\mathbf{a},\mathbf{b}}\bigr)_B
=\sum_{\mathbf{s}\in\{0,1\}^N} P_{\mathbf{s}}\!\!\!\!\sum_{\substack{\mathbf{r}\in\{0,1\}^N\\ d(\mathbf{a}_0,\mathbf{r}_0)+d(\mathbf{b}_1,\mathbf{r}_1)\leq N\gamma_\text{err}}}\!\!\!\!\!\!\!\!\! \bigl(\phi_{\mathbf{r},\mathbf{s}}\bigr)_{B},
\end{equation}
for some $\gamma_\text{err}\geq 0$. Let $\lambda$ be a lower bound on the minimum of the function $B(P_0,O)$ evaluated over the range $P_0 \in[P_0^{\text{LB}},P_0^{\text{UB}}]$, with 
\begin{equation}
\label{1003}
B(P_0,O)= \frac{1-\sqrt{1-4(1-O^2)P_0(1-P_0)}}{2}.
\end{equation}
In particular, if $P_0^{\text{LB}}=\frac{1}{2}-\beta_\text{PB}$ and $P_0^{\text{UB}}=\frac{1}{2}+\beta_\text{PB}$ for some $\beta_\text{PB}\geq 0$ then 
\begin{equation}
\label{1004}
\lambda=\frac{1}{2}\Bigl(1-\sqrt{1-(1-O^2)(1-4\beta_\text{PB}^2)}\Bigr).
\end{equation}
If $\gamma_{\text{err}}=0$ then it holds that
\begin{equation}
	\label{l0.1.2}
	\max_{\mathbf{a},\mathbf{b}\in\{0,1\}^N}\lVert D_{\mathbf{a},\mathbf{b}}\rVert \leq\bigl(1-\lambda\bigr)^N.
	\end{equation}
If $0<\gamma_\text{err}<\lambda$ then it holds that
	\begin{equation}
	\label{l0.2}
	\max_{\mathbf{a},\mathbf{b}\in\{0,1\}^N}\lVert D_{\mathbf{a},\mathbf{b}}\rVert \leq e^{-\frac{N\lambda}{2}\bigl(1-\frac{\gamma_\text{err}}{\lambda}\bigr)^2}.
	\end{equation}	
	\end{lemma}

\begin{proof}
For $\mathbf{a},\mathbf{b},\mathbf{s}\in\{0,1\}^N$, we define $\mathbf{u}\in\{0,1\}^N$ satisfying $\mathbf{u}_0=\mathbf{a}_0$ and $\mathbf{u}_1=\mathbf{b}_1$. Then, in (\ref{l0.1}), we change variables $\mathbf{r}= \mathbf{x}\oplus\mathbf{u}$ and we sum over $\mathbf{x}\in\{0,1\}^N$, where `$\oplus$' denotes bit-wise sum modulo 2. We obtain
\begin{equation}
\label{l1.6}
D_{\mathbf{a},\mathbf{b}}=\sum_{\mathbf{s}\in\{0,1\}^N}P_{\mathbf{s}} \sum_{\substack{\mathbf{x}\in\{0,1\}^N\\ w(\mathbf{x})\leq N\gamma_\text{err}}} \bigl(\phi_{\mathbf{x}\oplus \mathbf{u},\mathbf{s}}\bigr)_{B},
\end{equation}
where $w(\mathbf{x})$ denotes the Hamming weight of $\mathbf{x}$.

We note that $D_{\mathbf{a},\mathbf{b}}$ is a positive semi definite operator, hence, $\lVert D_{\mathbf{a},\mathbf{b}}\rVert$ corresponds to the greatest eigenvalue of $D_{\mathbf{a},\mathbf{b}}$, for $\mathbf{a},\mathbf{b}\in\{0,1\}^N$. 
For given $\mathbf{a},\mathbf{b}\in\{0,1\}^N$, in order to compute $\bigl\lVert D_{\mathbf{a},\mathbf{b}}\bigr\rVert$, we first evaluate the sum over $\mathbf{s}\in\{0,1\}^N$ in (\ref{l1.6}). We obtain
\begin{eqnarray}
\label{l1.12}
\sum_{\mathbf{s}\in\{0,1\}^N} P_{\mathbf{s}} \phi_{\mathbf{x}\oplus \mathbf{u},\mathbf{s}}&=&\bigotimes_{k=1}^N \Bigl(\sum_{s_k=0}^1 P^k_{s_k}\lvert \phi^k_{x_k\oplus u_k,s_k}\rangle\langle \phi^k_{x_k\oplus u_k,s_k}\rvert\Bigr)\nonumber\\
&=&\bigotimes_{k=1}^N \rho^k_{x_k\oplus a_k,x_k\oplus b_k},
\end{eqnarray}
where
\begin{equation}
\label{l1.13}
\rho_{b,c}^k=\Bigl(P_{h_k}^k \lvert \phi^k_{b,h_k}\rangle\langle \phi^k_{b,h_k}\rvert + P_{h_k\oplus 1}^k \lvert \phi^k_{c,h_k\oplus1}\rangle\langle \phi^k_{c,h_k\oplus 1}\rvert \Bigr),
\end{equation}
for $b,c\in\{0,1\}$ and $k\in[N]$. We note that $\rho_{b,c}^k+\rho_{b\oplus1,c\oplus1}^k=\mathds{1}$ since $\{\lvert \phi_{r,s}^k\rangle\}_{r\in\{0,1\}}$ is a qubit orthonormal basis for $s\in\{0,1\}$ and since $\{P_0^k,P_1^k\}$ is a probability distribution, for $k\in[N]$. Thus, $\rho_{b,c}^k$ and $\rho_{b\oplus1,c\oplus1}^k$ are diagonal in the same basis, for $b,c\in\{0,1\}$ and $k\in[N]$. Therefore, without loss of generality, we can write
\begin{equation}
\label{l1.14}
\rho_{b,c}^k=\sum_{t=0}^1 \lambda_{t\oplus b,b\oplus c}^k \bigl\lvert \mu_{t,b\oplus c}^k\bigr\rangle\bigl\langle \mu_{t,b\oplus c}^k\bigr\rvert,
\end{equation}
where $\bigl\{\bigl\lvert \mu_{t,b\oplus c}^k\bigr\rangle\bigr \}_{t=0}^1$ is the eigenbasis of $\rho_{b,c}^k$ with real non-negative eigenvalues $\bigl\{\lambda_{t\oplus b,b\oplus c}^k \bigr\}_{t=0}^1$, and where 
\begin{equation}
\label{lx}
\lambda_{0,c}^k+\lambda_{1,c}^k=1,
\end{equation}
for $b,c\in\{0,1\}$ and $j\in[N]$. Thus, we have
\begin{eqnarray}
\label{l1.15}
&&\Biggl(\bigotimes_{k=1}^N \rho^k_{x_k\oplus a_k,x_k\oplus b_k}\Biggr)\Biggl(\bigotimes_{k=1}^N  \bigl\lvert \mu_{t_k,a_k\oplus b_k}^k\bigr\rangle \Biggr)=\nonumber\\
&&\quad\Biggl(\prod_{k=1}^N \lambda_{t_k\oplus x_k\oplus a_k,a_k\oplus b_k}^k\Biggr) \Biggl(\bigotimes_{k=1}^N  \bigl\lvert \mu_{t_k,a_k\oplus b_k}^k\bigr\rangle \Biggr),
\end{eqnarray}
for $\mathbf{t}\in\{0,1\}^N$. Importantly, we see from (\ref{l1.15}) that the eigenbasis of $\bigotimes_{k=1}^N \rho^k_{x_k\oplus a_k,x_k\oplus b_k}$ is the same for all $\mathbf{x}\in\{0,1\}^N$. Thus, from (\ref{l1.6}), (\ref{l1.12}) and (\ref{l1.15}), we see that the eigenbasis of $D_{\mathbf{a},\mathbf{b}}$ is $\Bigl\{\bigotimes_{k=1}^N  \bigl\lvert \mu_{t_k,a_k\oplus b_k}^k\bigr\rangle\Bigr\}_{\mathbf{t}\in\{0,1\}^N}$, with eigenvalues 
\begin{equation}
\label{l1.16}
\sum_{\substack{\mathbf{x}\in\{0,1\}^N\\ w(\mathbf{x})\leq N\gamma_\text{err} }} \Biggl(\prod_{k=1}^N \lambda_{t_k\oplus x_k\oplus a_k,a_k\oplus b_k}^k\Biggr),\nonumber
\end{equation}
for $\mathbf{t}\in\{0,1\}^N$. It follows that 
\begin{eqnarray}
\label{l1.17}
&&\max_{\mathbf{a},\mathbf{b}\in\{0,1\}^N}\lVert D_{\mathbf{a},\mathbf{b}}\rVert\nonumber\\
&&\qquad\qquad=\max_{\mathbf{a},\mathbf{b},\mathbf{t}\in\{0,1\}^N} \sum_{\substack{\mathbf{x}\in\{0,1\}^N\\ w(\mathbf{x})\leq N\gamma_\text{err} }} \Biggl(\prod_{k=1}^N \lambda_{t_k\oplus x_k\oplus a_k,a_k\oplus b_k}^k\Biggr)\nonumber\\
&&\qquad\qquad = \max_{\mathbf{\alpha},\mathbf{\beta}\in\{0,1\}^N} \sum_{\substack{\mathbf{x}\in\{0,1\}^N\\ w(\mathbf{x})\leq N\gamma_\text{err} }} \Biggl(\prod_{k=1}^N \lambda_{x_k\oplus \beta_k,\alpha_k}^k\Biggr),
\end{eqnarray}
by taking the change of variables $\alpha_k=a_k\oplus b_k$ and $\beta_k=a_k\oplus t_k$, for $k\in[N]$.

Below we compute the maximum given by the second line of (\ref{l1.17}). We consider two cases separately, the case $\gamma_{\text{err}}=0$, and the case $0<\gamma_{\text{err}}<\lambda$. Within the second case we consider the subcases $0<\gamma_{\text{err}}<\frac{1}{N}$ and $\gamma_{\text{err}}\geq \frac{1}{N}$. We use the following definitions:
\begin{eqnarray}
\label{l1.20}
\lambda_0^k&=&\max\{\lambda_{\beta,\alpha}^k\}_{\alpha,\beta\in\{0,1\}}\nonumber\\
\lambda_1^k&=&1-\lambda_0^k,
\end{eqnarray}
where in the second line we used (\ref{lx}), for $k\in[N]$. We also define parameters $\lambda_0\leq 1$ and $\lambda_1$ that satisfy
\begin{eqnarray}
\label{l1.21}
\lambda_0^k&\leq& \lambda_0,\nonumber\\
\lambda_1&=&1-\lambda_0,
\end{eqnarray}
for $k\in[N]$.

In the case $\gamma_{\text{err}}=0$, we have from (\ref{l1.17}) that
\begin{eqnarray}
\label{l1.17.1}
\max_{\mathbf{a},\mathbf{b}\in\{0,1\}^N}\lVert D_{\mathbf{a},\mathbf{b}}\rVert&=&  \max_{\mathbf{\alpha},\mathbf{\beta}\in\{0,1\}^N}  \prod_{k=1}^N \lambda_{\beta_k,\alpha_k}^k\nonumber\\
&=&  \prod_{k=1}^N \lambda_{0}^k\nonumber\\
&\leq&  (1-\lambda_1)^N\nonumber\\
&=&(1-\lambda)^N,
\end{eqnarray}
where in the second line we used (\ref{l1.20}), in the third line we used (\ref{l1.21}), and in the last line we used that $\lambda_1=\lambda$, which is shown below. The bound (\ref{l0.1.2}) follows from (\ref{l1.17.1}).

In the case $0<\gamma_{\text{err}}<\frac{1}{N}$, we note that since $\lambda\in(0,1)$, we have $\ln (1-\lambda)\leq -\lambda<-\frac{\lambda}{2}$. It follows that
\begin{equation}
\label{l1.17.2}
(1-\lambda)^N < e^{-\frac{N\lambda}{2}\bigl(1-\frac{\gamma_{\text{err}}}{\lambda}\bigr)^2}.
\end{equation}
Thus, from (\ref{l1.17.1}) and (\ref{l1.17.2}), it follows that in the case that the conditions $0<\gamma_{\text{err}}<\lambda$ and $0<\gamma_{\text{err}}<\frac{1}{N}$ hold, the bound (\ref{l0.2}) is satisfied.

We show below that the bound (\ref{l0.2}) is satisfied too in the case that $\frac{1}{N}\leq\gamma_{\text{err}}<\lambda$ holds. It follows that (\ref{l0.2}) holds if $0<\gamma_{\text{err}}<\lambda$, as stated in the lemma.

We consider $\frac{1}{N}\leq\gamma_{\text{err}}<\lambda$. For any $\mathbf{\alpha},\mathbf{\beta}\in\{0,1\}^N$ and for any $l\in [N]$, we define $\tilde{\mathbf{x}}_l=(x_1,x_2,\ldots,x_{l-1},x_{l+1},x_{l+2},\ldots,x_N)$, and we can write
\begin{eqnarray}
\label{l1.18}
&&\sum_{\substack{\mathbf{x}\in\{0,1\}^N\\ w(\mathbf{x})\leq N\gamma_\text{err} }} \prod_{k=1}^N \lambda_{x_k\oplus \beta_k,\alpha_k}^k=\nonumber\\
&&\qquad \lambda_{\beta_l\oplus 1,\alpha_l}^l \Biggl(  \sum_{\substack{\tilde{\mathbf{x}}_l\in\{0,1\}^{N-1}\\ w(\tilde{\mathbf{x}}_l)\leq N\gamma_\text{err} -1}} \prod_{\substack{k=1\\k\neq l}}^N \lambda_{x_k\oplus \beta_k,\alpha_k}^k \Biggr)\nonumber\\
&&\qquad\!\!\quad+\lambda_{\beta_l,\alpha_l}^l\Biggl( \sum_{\substack{\tilde{\mathbf{x}}_l\in\{0,1\}^{N-1}\\ w(\tilde{\mathbf{x}}_l)\leq N\gamma_\text{err} }} \prod_{\substack{k=1\\ k\neq l}}^N \lambda_{x_k\oplus \beta_k,\alpha_k}^k \Biggr).
\end{eqnarray}
We see that the term inside the second bracket cannot be smaller than the term inside the first one. Since this holds for any choice of $l\in[N]$, in order to maximize the quantity on the left-hand side, we need to maximize $\lambda_{\beta_l,\alpha_l}^l$ for $l\in[N]$. Thus, we obtain from (\ref{l1.17}), (\ref{l1.20}) and (\ref{l1.18}) that
\begin{equation}
\label{l1.19}
\max_{\mathbf{a},\mathbf{b}\in\{0,1\}^N}\lVert D_{\mathbf{a},\mathbf{b}}\rVert= \sum_{\substack{\mathbf{x}\in\{0,1\}^N\\ w(\mathbf{x})\leq N\gamma_\text{err} }} \Biggl(\prod_{k=1}^N \lambda_{x_k}^k\Biggr).
\end{equation}
Similarly, reasoning as in the previous lines, it is straightforward to obtain from (\ref{l1.21}) and (\ref{l1.19}) that
\begin{eqnarray}
\label{l1.23}
\max_{\mathbf{a},\mathbf{b}\in\{0,1\}^N}\lVert D_{\mathbf{a},\mathbf{b}}\rVert&\leq& \sum_{\substack{\mathbf{x}\in\{0,1\}^N\\ w(\mathbf{x})\leq N\gamma_\text{err} }} \Biggl(\prod_{k=1}^N \lambda_{x_k}\Biggr)\nonumber\\
&=&\!\sum_{n=0}^{\lfloor N\gamma_\text{err}\rfloor}\!\!\! \bigl(\begin{smallmatrix} N\\ n\end{smallmatrix}\bigr)(\lambda_0)^{N-n}(\lambda_1)^n.
\end{eqnarray}

We upper bound the right-hand side of (\ref{l1.23}) using the Chernoff bound given by Proposition \ref{proposition3}. Let $X_k$ be a random variable taking value $X_k=i$ with probability $\lambda_i$, for $i\in\{0,1\}$ and $k\in[N]$. Let $X=\sum_{k=1}^N {X}_k$, whose average value is $E(X)=N\lambda_1$. We have
\begin{eqnarray}
\label{l1.24}
\sum_{n=0}^{\lfloor N\gamma_\text{err}\rfloor} \bigl(\begin{smallmatrix} N\\ n\end{smallmatrix}\bigr)(\lambda_0)^{N-n}(\lambda_1)^n&\leq&  \text{Pr}[X\leq N\gamma_\text{err} ]\nonumber\\
&\leq& e^{-\frac{N\lambda_1}{2}\bigl(1-\frac{\gamma_\text{err}}{\lambda_1}\bigr)^2},
\end{eqnarray}
for $0<\gamma_\text{err} < \lambda_1$, where in the second line we used the Chernoff bound (\ref{ly}), by taking $\epsilon=1-\frac{\gamma_\text{err}}{\lambda_1}$. By taking $\lambda_1=\lambda$, it follows from (\ref{l1.23}) and (\ref{l1.24}) that
\begin{equation}
\label{l1.25}
\max_{\mathbf{a},\mathbf{b}\in\{0,1\}^N}\lVert D_{\mathbf{a},\mathbf{b}}\rVert\leq e^{-\frac{N\lambda}{2}\bigl(1-\frac{\gamma_\text{err}}{\lambda}\bigr)^2},
\end{equation}
for $0<\gamma_\text{err} < \lambda$, as claimed.

We complete the proof below by showing that $\lambda_{1}=\lambda$ satisfies the condition (\ref{l1.21}). We write $\langle \phi_{b,h_k}^k\vert \phi_{c,h_k\oplus 1}^k\rangle=\omega^k_{b,c}e^{i\chi_{b,c}^k}$, with $\omega^k_{b,c}=\bigl\lvert  \langle \phi_{b,h_k}^k\vert \phi_{c,h_k\oplus 1}^k\rangle\bigr\rvert$, for $b,c\in\{0,1\}$ and $k\in[N]$. We define 
\begin{equation}
\label{1000}
R^k_{\pm,b,c}=\frac{P_{h_k\oplus 1}^k-P_{h_k}^k\pm \sqrt{(P_{1}^k-P_{0}^k)^2+4(\omega_{b,c}^k)^2P_{0}^kP_{1}^k}}{2\omega_{b,c}^k P_{h_k}^k},
\end{equation}
for $b,c\in\{0,1\}$ and $k\in[N]$. It is straightforward to verify that the density matrix $\rho_{b,c}^k$ given by (\ref{l1.13}) has eigenstates
\begin{eqnarray}
\label{l1.26}
&&\bigl\lvert e_{\pm ,b, c}^k\bigr\rangle=\nonumber\\
&&\qquad \frac{1}{\sqrt{1+\bigl(R^k_{\pm,b,c}\bigr)^2}}\Bigl(\bigl\lvert \phi_{b,h_k}^k\bigr\rangle+ R_{\pm,b,c}e^{-i\chi_{b,c}^k}  \bigl\lvert \phi_{c,h_k\oplus 1}^k\bigr\rangle \Bigr)\nonumber\\
\end{eqnarray}
with eigenvalues
\begin{equation}
\label{l1.27}
\lambda_{\pm ,b, c}^k= P_{h_k}^k\Bigl(1+ \omega^k_{b,c} R^k_{\pm,b,c}\Bigr),
\end{equation}
for $b,c\in\{0,1\}$ and $k\in[N]$. Thus, from the definition (\ref{l1.20}) and from (\ref{l1.27}), we obtain
\begin{equation}
\label{l1.28}
2\lambda_0^k = P_0^k+P_1^k+\sqrt{\bigl(P_1^k\!-\!P_0^k\bigr)^2\!+\!4P_0^kP_1^k\max_{b,c\in\{0,1\}}\{(\omega_{b,c}^k)^2\}},\\
\end{equation}
for $k\in[N]$. Since by assumption of the lemma, $\bigl\lvert  \langle \phi_{b,0}^k\vert \phi_{c,1}^k\rangle\bigr\rvert\leq O$ for some $O\in\bigl[\frac{1}{\sqrt{2}},1\bigr)$, using $P_1^k=1-P_0^k$ and $\omega^k_{b,c}=\bigl\lvert  \langle \phi_{b,h_k}^k\vert \phi_{c,h_k\oplus 1}^k\rangle\bigr\rvert$, we have from (\ref{l1.28}) that
\begin{equation}
\label{1001}
\lambda_0^k\leq A(P_0^k,O),
\end{equation}
for $k\in[N]$, where
\begin{equation}
\label{1002}
A(P_0,O)= \frac{1+\sqrt{1-4[1-O^2]P_0(1-P_0)}}{2}.
\end{equation}
Thus, since $P_0^k\in[P_0^{\text{LB}},P_0^{\text{UB}}]$ for $k\in[N]$, by defining $\lambda_0$ as an upper bound on the maximum of the function $A(P_0,O)$ evaluated over the range $P_0 \in[P_0^{\text{LB}},P_0^{\text{UB}}]$, with $\lambda_0 <1-\gamma_\text{err}$, and by defining $\lambda_1=1-\lambda_0$, the conditions given by (\ref{l1.21}) hold. We see from (\ref{1002}) that the function  $B(P_0,O)$ given by (\ref{1003}) satisfies $B(P_0,O)=1-A(P_0,O)$. Thus, we define $\lambda_1$ as a lower bound on the minimum of the function $B(P_0,O)$ evaluated over the range $P_0 \in[P_0^{\text{LB}},P_0^{\text{UB}}]$, and we define $\lambda=\lambda_1$.
In the case that $P_0^{\text{LB}}=\frac{1}{2}-\beta_\text{PB}$ and $P_0^{\text{UB}}=\frac{1}{2}+\beta_\text{PB}$ for $\beta_\text{PB}\geq 0$,
we define $\lambda_1$ as the minimum of the function $B(P_0,O)$ evaluated over the range $P_0 \in[P_0^{\text{LB}},P_0^{\text{UB}}]$,  and we define $\lambda=\lambda_1$. It is straightforward to see from (\ref{1002}) that in this case $\lambda_1=\frac{1}{2}\bigl(1-\sqrt{1-[1-O^2](1-4\beta_\text{PB}^2)}\bigr)$. The result follows by noting that, as stated in the lemma, $\lambda=\frac{1}{2}\bigl(1-\sqrt{1-[1-O^2](1-4\beta_\text{PB}^2)}\bigr)$.
\end{proof}

\begin{lemma}
\label{eigenvalue}
Let $\rho$ be a qubit density matrix given by
\begin{equation}
\label{N1018000}
\rho=\sum_{u=0}^1 \sum_{t=0}^1P_{\text{PB}}(u)P_{\text{PS}}(t)\lvert \phi_{tu}\rangle\langle\phi_{tu}\rvert,
\end{equation}
where $\{\lvert \phi_{tu}\rangle\}_{t,u\in\{0,1\}}$ is a set of qubit states satisfying $\langle \phi_{0u}\vert \phi_{1u}\rangle=0$ for $u\in\{0,1\}$, where the qubit orthonormal basis $\mathcal{D}_u=\{\lvert \phi_{tu}\rangle\}_{t=0}^1$ is the computational (Hadamard) basis within an uncertainty angle $\theta\in\bigl(0,\frac{\pi}{4}\bigr)$ on the Bloch sphere if $u=0$ ($u=1$); and where the binary probability distributions $\{P_{\text{PB}}(u)\}_{u=0}^1$ and $\{P_{\text{PS}}(t)\}_{t=0}^1$ satisfy $\frac{1}{2}-\beta_{\text{PB}}\leq P_{\text{PB}}(u)\leq \frac{1}{2}+\beta_{\text{PB}}$ and $\frac{1}{2}-\beta_{\text{PS}}\leq P_{\text{PS}}(u)\leq \frac{1}{2}+\beta_{\text{PS}}$ for $u\in\{0,1\}$, and for given parameters $\beta_{\text{PB}},\beta_{\text{PS}}\in\bigl(0,\frac{1}{2}\bigr)$. Let $\mu_+$ be the greatest eignevalues of $\rho$. It holds that
\begin{equation}
\label{0007}
\mu_+\leq \frac{1}{2}\bigl(1+h(\beta_\text{PS},\beta_\text{PB},\theta)\bigr),
\end{equation}
where 
\begin{equation}
\label{0006}
h(\beta_\text{PS},\beta_\text{PB},\theta)=2\beta_\text{PS}\sqrt{\frac{1}{2}\!+\!2\beta_\text{PB}^2\!+\!\Bigl(\!\frac{1}{2}-2\beta_\text{PB}^2\Bigr)\sin(2\theta)}.\\
\end{equation}
\end{lemma}

\begin{proof}
To simplify notation, we define $P=P_{\text{PB}}(0)$, $1-P=P_{\text{PB}}(1)$, $R=P_{\text{PS}}(0)$, and $1-R=P_{\text{PS}}(1)$. Since  applying a unitary operation $U$ on $\rho$ does not change its eigenvalues, we define $\rho'= U\rho U^\dagger$ and we compute an upperbound on the greatest eigenvalue of $\rho'$. Since $\{\lvert \phi_{tu}\rangle\}_{t=0}^1$ is a qubit orthonormal basis, for $u\in\{0,1\}$, we can choose $U$ such that $U\lvert \phi_{t0}\rangle=\lvert t\rangle$ and $U\lvert \phi_{t1}\rangle=\lvert \tilde{t}\rangle$, where $\{\lvert 0\rangle,\lvert 1\rangle\}$ is the computational basis, which has Bloch vectors in the $z$ axis, and where $\{\lvert \tilde{0}\rangle,\lvert \tilde{1}\rangle\}$ is another orthonormal basis with Bloch vectors on the $z-x$ plane. Thus, we see that from the statement of the lemma, we can choose $U$ such that $\lvert \tilde{0}\rangle$ has a Bloch vector with angle $\xi$ above the $x$ axis, towards the $z$ axis; hence, $\lvert \tilde{1}\rangle$ has a Bloch vector with angle $\xi$ below the $-x$ axis, towards the $-z$ axis; where $\xi\in[-2\theta,2\theta]$ for some $\theta\in\bigl(0,\frac{\pi}{4}\bigr)$.

Thus, using this notation, from (\ref{N1018000}), we obtain
\begin{equation}
\label{N1038000}
\rho'=P \rho_0+(1-P)\rho_1,
\end{equation}
where 
\begin{eqnarray}
\label{N1039000}
\rho_0&=&R\lvert 0\rangle\langle 0\rvert +(1-R)\lvert 1\rangle\langle 1\rvert,\nonumber\\
\rho_1&=&R\lvert \tilde{0}\rangle\langle \tilde{0}\rvert +(1-R)\lvert \tilde{1}\rangle\langle \tilde{1}\rvert.
\end{eqnarray}
The Bloch vector of $\rho_0$ is $(2R-1)\hat{z}$ and the Bloch vector of $\rho_1$ is $(2R-1)\bigl(\cos\xi\hat{x}+\sin\xi\hat{z}\bigr)$, where $\hat{x}$ and $\hat{z}$ are unit vectors pointing along the $x$ and $z$ axes, respectively. Thus, the Bloch vector of $\rho'$ is 
\begin{equation}
\label{N1040000}
\vec{r}=(2R-1)\bigl[\bigl(P+(1-P)\sin\xi\bigr)\hat{z}+(1-P)\cos\xi\hat{x}\bigr].
\end{equation}
The eigenvalues of $\rho'$, hence also of $\rho$, are given by $\mu_{\pm}=\frac{1}{2}\bigl(1\pm\lvert \vec{r}\rvert\bigr)$. Thus, from (\ref{N1040000}), the greatest eigenvalue is given by
\begin{equation}
\label{N1041000}
\mu_+=\frac{1}{2}\bigl(1+\lvert\vec{r}\rvert\bigr),
\end{equation}
where
\begin{equation}
\label{N1042000}
\lvert\vec{r}\rvert = \lvert 2R-1\rvert\sqrt{\bigl(P+(1-P)\sin\xi\bigr)^2+(1-P)^2\cos^2\xi}.
\end{equation}
Since from the statement of the lemma we have $\frac{1}{2}-\beta_\text{PS}\leq R\leq \frac{1}{2}+\beta_\text{PS}$, we see from (\ref{N1041000}) and (\ref{N1042000}) that for fixed values of $P$ and $\xi$, $\mu_+$ achieves its maximum if $R=\frac{1}{2}\pm \beta_\text{PS}$. Thus, it holds that 
\begin{equation}
\label{N1043000}
\mu_+ \leq \frac{1}{2}\bigl(1+2\beta_\text{PS} \sqrt{g(P,\xi)}\bigr),
\end{equation}
where 
\begin{equation}
\label{N1044000}
g(P,\xi)=\bigl(P+(1-P)\sin\xi\bigr)^2+(1-P)^2\cos^2\xi.
\end{equation}
We write $P=\frac{1}{2}+d$. From the statement of the lemma, we have $d\in[-\beta_\text{PB},\beta_\text{PB}]$ for some $\beta_\text{PB}\in\bigl(0,\frac{1}{2}\bigr)$. It follows from (\ref{N1044000}) that
\begin{equation}
\label{N1045000}
g(P,\xi)=\frac{1}{2}(1+\sin\xi)+2(1-\sin\xi)d^2.
\end{equation}
Since $\xi\in[-2\theta,2\theta]$ with $0<\theta<\frac{\pi}{4}$, we have $(1-\sin\xi)>0$. Thus, $g(P,\xi)$ is maximum when $d^2$ is maximum. Since $d\in[-\beta_\text{PB},\beta_\text{PB}]$, it follows that
\begin{eqnarray}
\label{N1046000}
g(P,\xi)&\leq&\frac{1}{2}(1+\sin\xi)+2(1-\sin\xi)\beta_\text{PB}^2\nonumber\\
&=&\frac{1}{2}+2\beta_\text{PB}^2+\Bigl(\frac{1}{2}-2\beta_\text{PB}^2\Bigr)\sin\xi.
\end{eqnarray}
Since $\beta_\text{PB}\in\bigl(0,\frac{1}{2}\bigr)$, we have $\frac{1}{2}-2\beta_\text{PB}^2>0$. Thus, since $\xi\in[-2\theta,2\theta]$ with $0<\theta<\frac{\pi}{4}$, the second line of (\ref{N1046000}) is maximum when $\xi=2\theta$. It follows that
\begin{equation}
\label{N1047000}
g(P,\xi)\leq \frac{1}{2}+2\beta_\text{PB}^2+\Bigl(\frac{1}{2}-2\beta_\text{PB}^2\Bigr)\sin(2\theta).
\end{equation}
Thus, from (\ref{N1043000}) and (\ref{N1047000}), we obtain (\ref{0007}):
\begin{equation}
\label{N1048000}
\mu_+\leq \frac{1}{2}\bigl(1+h(\beta_\text{PS},\beta_\text{PB},\theta)\bigr),
\end{equation}
where $h(\beta_\text{PS},\beta_\text{PB},\theta)$ is given by (\ref{0006}), as claimed.
\end{proof}

\section{Proof of Lemma \ref{lemma01}}
\label{idealunfapp}

\begin{lemma1*}
\label{lemma1repeated}
The quantum token schemes $\mathcal{IQT}_1$ and $\mathcal{IQT}_2$ are $\epsilon_{\text{unf}}-$unforgeable with
	\begin{equation}
	\label{000repeated}
	\epsilon_{\text{unf}}=\Bigl(\frac{1}{2}+\frac{1}{2\sqrt{2}}\Bigr)^N.
\end{equation}
\end{lemma1*}

\begin{proof}

In summary, the proof comprises  reducing a general cheating strategy by Alice in the token schemes $\mathcal{IQT}_1$ and $\mathcal{IQT}_2$ to the task of producing the $N-$bit strings $\mathbf{a}$ and $\mathbf{b}$ given in Lemma \ref{lemma0}, with $\gamma_{\text{err}}=0$, $\beta_{\text{PB}}=0$, and $\theta=0$, i.e. $O(\theta)\equiv O=\frac{1}{\sqrt{2}}$. As we show, then Alice's success probability is upper bounded by the quantity $ \max_{\mathbf{a},\mathbf{b}\in\{0,1\}^N}\lVert D_{\mathbf{a},\mathbf{b}}\rVert$, which for these parameters is upper bounded by $\bigl(\frac{1}{2}+\frac{1}{2\sqrt{2}}\bigr)^N$.

Consider the token schemes $\mathcal{IQT}_1$ and $\mathcal{IQT}_2$. In these token schemes Alice gives Bob a $N-$bit strings $\mathbf{d}=(d_1,\ldots,d_N)$ and a bit $c$. Using this information, honest Bob computes the $N-$bit string  $\tilde{\mathbf{d}}_i=(\tilde{d}_{i,1},\ldots,\tilde{d}_{i,N})$ in the causal past of the presentation point $Q_i$, where $\tilde{d}_{i,k}=d_k\oplus c \oplus i$, for $k\in[N]$ and $i\in\{0,1\}$. In a general cheating strategy $\mathcal{S}$, Alice applies a joint projective measurement on the quantum system $A$ of $N-$qubits in the state $\lvert \phi_{\mathbf{t}\mathbf{u}}\rangle_A=\bigotimes_{k=1}^N\lvert \phi_{t_ku_k}\rangle_{A_k}$ received from Bob and an ancilla $E$ of arbitrary finite Hilbert space dimension in a quantum state $\vert\chi\rangle_E$, and obtains the classical message $x=(\mathbf{d},c)$ of $N+1$ bits that she gives Bob within the causal pasts of $Q_0$ and $Q_1$ and two $N-$bit (token) strings $\mathbf{a}=(a_1\ldots,a_N)$ and $\mathbf{b}=(b_1,\ldots,b_N)$ that she gives to Bob at $Q_0$ and $Q_1$, respectively. Alice succeeds in making Bob validate these token strings at $Q_0$ and $Q_1$ if $\mathbf{a}_0=\mathbf{t}_0$ and $\mathbf{b}_1=\mathbf{t}_1$, where $\mathbf{x}_i$ is a restriction of the string $\mathbf{x}\in\{\mathbf{a},\mathbf{b},\mathbf{t}\}$ to the bit entries $x_k\in\Delta_i$, where $\Delta_i=\{k\in[N]\vert \tilde{d}_{i,k}=u_k\}$, for $i\in\{0,1\}$. Since $\tilde{d}_{i,k}=d_k\oplus c \oplus i$, for $k\in[N]$, we have that $\Delta_i=\{k\in[N]\vert u_k=d_k\oplus c \oplus i\}$, for $i\in\{0,1\}$.

Now consider the task of Lemma \ref{lemma0} with the following parameters. The states $\lvert\phi_{r,s}^k\rangle$ are the BB84 states: $\lvert \phi_{0,0}^k\rangle\equiv\lvert \phi_{00}\rangle=\lvert0\rangle$, $\lvert \phi_{1,0}^k\rangle\equiv\lvert \phi_{10}\rangle=\lvert 1\rangle$, $\lvert \phi_{0,1}^k\rangle\equiv\lvert \phi_{01}\rangle=\lvert +\rangle$, $\lvert \phi_{1,1}^k\rangle\equiv\lvert \phi_{11}\rangle=\lvert -\rangle$, for $k\in[N]$. It follows that $O=\frac{1}{\sqrt{2}}$. We also consider that $P_{\mathbf{s}}=\bigl(\frac{1}{2}\bigr)^N$ for $\mathbf{s}\in\{0,1\}^N$, i.e. $\beta_{\text{PB}}=0$. It follows from Lemma \ref{lemma0} that $\lambda=\frac{1}{2}-\frac{1}{2\sqrt{2}}$ and that 
		\begin{equation}
	\label{001}
	\max_{\mathbf{a},\mathbf{b}\in\{0,1\}^N}\lVert D_{\mathbf{a},\mathbf{b}}\rVert \leq\Bigl(\frac{1}{2}+\frac{1}{2\sqrt{2}}\Bigr)^N.
	\end{equation}
	We define the $N-$bit strings $\mathbf{r}=(r_1,\ldots,r_N)$, $\mathbf{s}=(s_1,\ldots,s_N)$ and $\mathbf{h}=(h_1,\ldots,h_N)$ in terms of the strings $\mathbf{t}$, $\mathbf{u}$ and $\mathbf{d}$ and of the bit $c$ of the token schemes $\mathcal{IQT}_1$ and $\mathcal{IQT}_2$ as follows: $r_k=t_k$, $s_k=u_k$ and $h_k=d_k\oplus c$, for $k\in[N]$. Thus, the set  $S_i^{\mathbf{h}}$ in Lemma \ref{lemma0} is the set $\Delta_i$ in the token schemes $\mathcal{IQT}_1$ and $\mathcal{IQT}_2$: $S_i^{\mathbf{h}}=\Delta_i$, for $i\in\{0,1\}$. It follows that the operator $D_{\mathbf{a},\mathbf{b}}$ in Lemma \ref{lemma0} can be associated to Alice's cheating strategy in the token schemes $\mathcal{IQT}_1$ and $\mathcal{IQT}_2$. We deduce this connection below.

We consider an entanglement-based version of the token schemes $\mathcal{IQT}_1$ and $\mathcal{IQT}_2$. Bob prepares a pair of qubits $B_kA_k$ in the Bell state $\lvert\Phi^+\rangle_{B_kA_k}=\frac{1}{\sqrt{2}}\bigl(\lvert 0\rangle\lvert0\rangle+\lvert1\rangle\lvert1\rangle\bigr)_{B_kA_k}$, sends the qubit $A_k$ to Alice, chooses $u_k\in\{0,1\}$ with probability $\frac{1}{2}$ and then measures the qubit $B_k$ in the basis $\mathcal{D}_{u_k}=\{\lvert \phi_{tu_k}\rangle\}_{t=0}^1$, obtaining the outcome $\lvert \phi_{t_ku_k}\rangle$ randomly, with Alice's qubit $A_k$ projecting into the same state, for $t_k\in\{0,1\}$. In a general cheating strategy $\mathcal{S}$, Alice introduces an ancillary quantum system $E$ of arbitrary finite Hilbert space dimension in a pure state $\lvert \chi\rangle_E$ and then applies a projective measurement on $AE$, with projector operators $\Pi_{x\mathbf{a}\mathbf{b}}$, where the measurement outcomes $x=(\mathbf{d},c)\in\{0,1\}^{N+1}$ and $\mathbf{a},\mathbf{b}\in\{0,1\}^{N}$ correspond to the classical messages that Alice gives Bob, and where $A=A_1\cdots A_N$. The probability that Alice obtains outcomes $x, \mathbf{a}$ and $\mathbf{b}$ following her strategy $\mathcal{S}$, for given values of  $\mathbf{u}$ and $\mathbf{t}$, is given by
\begin{equation}
\label{N1009000}
P_{\mathcal{S}}[x\mathbf{a}\mathbf{b}\vert \mathbf{t}\mathbf{u}]=\text{Tr}\Bigl[\Phi_{\mathbf{t}\mathbf{u}}\Pi_{x\mathbf{a}\mathbf{b}}\Bigr],
\end{equation}
where $\Phi_{\mathbf{t}\mathbf{u}}=\bigl(\lvert\phi_{\mathbf{t}\mathbf{u}}\rangle \langle\phi_{\mathbf{t}\mathbf{u}}\rvert\bigr)_{A}\otimes \bigl(\lvert \chi\rangle\langle\chi\rvert\bigr)_E$. We define the sets
\begin{eqnarray}
\label{0003}
\Gamma^x_{\mathbf{a}\mathbf{b}\mathbf{u}}\!&=&\!\bigl\{\!(\mathbf{a},\mathbf{b})\!\in\!\{0,1\}^N\!\!\times\!\!\{0,1\}^N\!\vert \mathbf{a}_0\!=\!\mathbf{t}_0,\!\mathbf{b}_1\!=\!\mathbf{t}_1\!\bigr\}\!,
\\
\label{0004}
\xi^x_{\mathbf{a}\mathbf{b}\mathbf{u}}&=&\bigl\{\mathbf{t}\in\{0,1\}^N\vert \mathbf{a}_0=\mathbf{t}_0,\mathbf{b}_1=\mathbf{t}_1\bigr\}.
\end{eqnarray}
It follows that Alice's success probability $P_{\mathcal{S}}$ satisfies
\begin{eqnarray}
\label{N1007000}
P_{\mathcal{S}}&=&\Bigl(\frac{1}{4}\Bigr)^N\sum_{x,\mathbf{u},\mathbf{t}}~\sum_{ (\mathbf{a},\mathbf{b})\in \Gamma^x_{\mathbf{t}\mathbf{u}}} P_{\mathcal{S}}[x\mathbf{a}\mathbf{b}\vert \mathbf{t}\mathbf{u}]\nonumber\\
&=&\Bigl(\frac{1}{4}\Bigr)^N\sum_{\mathbf{a},\mathbf{b},x,\mathbf{u}}~\sum_{ \mathbf{t}\in \xi^x_{\mathbf{a}\mathbf{b}\mathbf{u}}} P_{\mathcal{S}}[x\mathbf{a}\mathbf{b}\vert \mathbf{t}\mathbf{u}],
\end{eqnarray}
where in the first line we used (\ref{N1009000}) and (\ref{0003}); and where in the second line we used (\ref{0003}) and (\ref{0004}), and the fact that the string $\mathbf{z}_i$ has bit entries with labels from the set $\Delta_i$ satisfying $\Delta_0\cap\Delta_1=\emptyset$ and $\Delta_0\cup\Delta_1=[N]$, for $i\in\{0,1\}$ and $\mathbf{z}\in\{\mathbf{a},\mathbf{b},\mathbf{t}\}$.

We define the quantum state
\begin{equation}
\label{Nl1.1000}
\rho=\bigl(\Phi^+\bigr)_{BA}\otimes\bigl(\lvert \chi\rangle\langle \chi\rvert\bigr)_{E},
\end{equation}
where $B$ denotes the system held by Bob and where $\bigl(\Phi^+\bigr)_{BA}=\bigotimes_{k\in[N]}\bigl(\lvert\Phi^+\rangle\langle\Phi^+\rvert\bigr)_{B_kA_k}$. We define the positive semi definite (and Hermitian) operators
\begin{eqnarray}
\label{Nl1.7x000}
D_{x\mathbf{a}\mathbf{b}}&=&\Bigl(\frac{1}{2}\Bigr)^N\sum_{\mathbf{u}}\sum_{\mathbf{t}\in\xi^x_{\mathbf{a}\mathbf{b}\mathbf{u}}}\bigl(\phi_{\mathbf{t}\mathbf{u}}\bigr)_{B},\\
\label{Nl1.7000}
\tilde{P}&=&\sum_{x,\mathbf{a},\mathbf{b}} \bigl(D_{x\mathbf{a}\mathbf{b}}\bigr)_{B}\otimes \bigl(\Pi_{x\mathbf{a}\mathbf{b}}\bigr)_{AE},
\end{eqnarray}
where $\bigl(\phi_{\mathbf{t}\mathbf{u}}\bigr)_{B}=\bigotimes _{k\in[N]} \bigl(\lvert \phi_{t_ku_k}\rangle\langle\phi_{t_ku_k}\rvert\bigr)_{B_k}$ and where $\mathbf{u}$ runs over $\{0,1\}^N$, $x$ runs over $\{0,1\}^{N+1}$, and $\mathbf{a}$ and $\mathbf{b}$ run over $\{0,1\}^{N}$. It follows straightforwardly from (\ref{N1009000}) -- (\ref{Nl1.7000}) that
\begin{eqnarray}
\label{N1010000}
P_{\mathcal{S}}&=& \text{Tr} \bigl( \tilde{P}\rho\bigr)\nonumber\\
&\leq&\lVert \tilde{P}\rVert\nonumber\\
&=&\max_{x,\mathbf{a},\mathbf{b}} \bigl\lVert D_{x\mathbf{a}\mathbf{b}}\bigr\rVert,
\end{eqnarray}
where in the second line we used Proposition \ref{proposition1}; and where in the third line we used (\ref{Nl1.7000}) and Proposition \ref{proposition2}, since $\{\Pi_{x\mathbf{a}\mathbf{b}}\}_{x,\mathbf{a},\mathbf{b}}$ is a projective measurement acting on a finite dimensional Hilbert space and $\{D_{x\mathbf{a}\mathbf{b}}\}_{x,\mathbf{a},\mathbf{b}}$ is a finite set of positive semi definite operators acting on a finite dimensional Hilbert space. We note that the operator $D_{x\mathbf{a}\mathbf{b}}$ defined by (\ref{Nl1.7x000}) equals the operator $D_{\mathbf{a}\mathbf{b}}$ given in Lemma \ref{lemma0}, for the parameters $\gamma_{\text{err}}=0$, $O=\frac{1}{\sqrt{2}}$ and $\beta_{\text{PB}}=0$ that we are considering here. Thus, from (\ref{001}) and (\ref{N1010000}),
and because this bound does not depend on $x$, we obtain
\begin{equation}
\label{0005}
P_{\mathcal{S}}\leq \Bigl(\frac{1}{2}+\frac{1}{2\sqrt{2}}\Bigr)^N.
\end{equation}
Thus, the quantum token schemes $\mathcal{IQT}_1$ and $\mathcal{IQT}_2$ are $\epsilon_{\text{unf}}-$unforgeable with $\epsilon_{\text{unf}}=\bigl(\frac{1}{2}+\frac{1}{2\sqrt{2}}\bigr)^N$, as claimed.	
	\end{proof}

\section{Proofs of Lemmas \ref{robust1}, \ref{correct1} and \ref{Bob1}}
\label{proofoflemmas}

We recall that Lemmas \ref{robust1}, \ref{correct1} and \ref{Bob1} consider parameters $\gamma_\text{det},\gamma_\text{err}\in(0,1)$, allow for the experimental imperfections of Table \ref{tableimp} and make the assumptions of Table \ref{tableassu}.

\subsection{Proof of Lemma \ref{robust1}}
\begin{lemma2*}
\label{lemma2repeated}
If
\begin{equation}
				\label{rob1repeated}
		0<\gamma_\text{det} <P_{\text{det}},
	\end{equation}
then $\mathcal{QT}_1$ and $\mathcal{QT}_2$ are $\epsilon_{\text{rob}}-$robust with
\begin{equation}
			\label{rob2repeated}
			\epsilon_{\text{rob}}=e^{-\frac{P_{\text{det}}N}{2}\bigl(1-\frac{\gamma_\text{det} }{P_{\text{det}}}\bigr)^2.
			}
			\end{equation}
\end{lemma2*}

We note that the condition (\ref{rob1repeated}) is necessary to guarantee robustness, as in the limit $N\rightarrow \infty$ the number $n$ of quantum states $\lvert \psi_k\rangle$ reported by Alice as being successfully measured tends to its expectation value $E(n)=P_{\text{det}}N$ with probability tending to unity. Thus, if $P_{\text{det}}<\gamma_\text{det} $ then $n<\gamma_\text{det} N$ and Bob aborts with probability tending to unity for $N\rightarrow \infty$.

\begin{proof}[Proof of Lemma 2]
Let $P_{\text{abort}}$ be the probability that Bob aborts the token scheme if Alice and Bob follow the token scheme honestly. 
By definition of the token schemes $\mathcal{QT}_1$ and $\mathcal{QT}_2$, we have
\begin{equation}
\label{aaax1}
P_{\text{abort}}=\text{Pr}[n<\gamma_\text{det} N].
\end{equation}
We note that the expectation value of $n$ is $E(n)=NP_\text{det}$. From (\ref{rob1repeated}), we have that $0<1-\frac{\gamma_\text{det}}{P_{\text{det}}}<1$. Thus, we obtain from a Chernoff bound of Proposition \ref{proposition3} that
 \begin{equation}
 \label{aaax3}
 \text{Pr}[n<\gamma_\text{det} N]<e^{-\frac{P_{\text{det}}N}{2}\bigl(1-\frac{\gamma_\text{det} }{P_{\text{det}}}\bigr)^2}.
 \end{equation}
It follows from (\ref{aaax1}) and (\ref{aaax3}) that
 \begin{equation}
 \label{aaax4}
 P_{\text{abort}}<\epsilon_{\text{rob}},
 \end{equation}
 with $\epsilon_{\text{rob}}$ given by (\ref{rob2repeated}). 
It follows from (\ref{aaax4}) that the token schemes $\mathcal{QT}_1$ and $\mathcal{QT}_2$ are $\epsilon_{\text{rob}}$-robust with $\epsilon_{\text{rob}}$ given by (\ref{rob2repeated}).

\end{proof}

\subsection{Proof of Lemma \ref{correct1}}
\begin{lemma3*}
\label{lemma3repeated}
If
\begin{eqnarray}
				\label{cor1repeated}
				0&<&\frac{\gamma_\text{err}}{2}<E<\gamma_\text{err},\nonumber\\
		0&<&\nu_\text{cor}<\frac{P_{\text{det}}(1-2\beta_\text{PB})}{2},
	\end{eqnarray}
then $\mathcal{QT}_1$ and $\mathcal{QT}_2$ are $\epsilon_{\text{cor}}-$correct with
\begin{equation}
			\label{cor2repeated}
			\epsilon_{\text{cor}}=e^{-\frac{P_{\text{det}}(1-2\beta_\text{PB})N}{4}\bigl(1-\frac{2\nu_\text{cor}}{P_{\text{det}}(1-2\beta_\text{PB})}\bigr)^2}+e^{-\frac{E\nu_\text{cor} N}{3}\bigl(\frac{\gamma_\text{err}}{E}-1\bigr)^2}.
			\end{equation}
\end{lemma3*}

We recall that $E=\max_{t,u}\{E_{tu}\}$, where $E_{tu}$ is Alice's error rate when Bob prepares states $\lvert\phi_{tu}^k\rangle$ and Alice measures in the basis of preparation by Bob, for $t,u\in\{0,1\}$. The condition
\begin{equation}
E_{\text{min}}<\gamma_\text{err},
\end{equation}
with $E_{\text{min}}=\min_{t,u} \{E_{tu}\}$, is necessary to guarantee correctness. To see this, suppose that $E_{\text{min}}>\gamma_\text{err}$. In the limit $N\rightarrow \infty$, we have $\lvert \Delta_b\rvert \rightarrow \infty$, in which case the number of error outcomes $n_{\text{errors}}$ when Alice measures in the same basis of preparation by Bob satisfies $n_{\text{errors}}\geq E_{\text{min}} \lvert \Delta_b\rvert>\gamma_\text{err} \lvert \Delta_b\rvert $ with probability tending to unity. Thus, with probability tending to unity, Bob does not accept Alice's token as valid, for $N\rightarrow \infty$, if $E_{\text{min}}>\gamma_\text{err}$.

\begin{proof}[Proof of Lemma 3]
Let $P_{\text{error}}$ be the probability that Bob does not accept Alice's token as valid if Alice and Bob follow the token scheme honestly. 
By definition of the token schemes $\mathcal{QT}_1$ and $\mathcal{QT}_2$, we have
\begin{equation}
\label{aaax1.2}
P_{\text{error}}=\sum_{\lvert \Delta_b\rvert=0}^N P_{\text{error}}(\lvert \Delta_b\rvert)\text{Pr}(\lvert \Delta_b\rvert),
\end{equation}
where
\begin{equation}
\label{aaax1.1}
P_{\text{error}}(\lvert \Delta_b\rvert)=\text{Pr}\Bigl[n_{\text{errors}}>\lvert \Delta_b\rvert\gamma_\text{err}\big\vert \lvert \Delta_b\rvert\Bigr],
\end{equation}
and where $n_{\text{errors}}$ is the number of bit errors in the substring $\mathbf{x}_b$ of the the token $\mathbf{x}$ that Alice presents to Bob at $Q_b$, compared to the bits of the substring $\mathbf{r}_b$ of $\mathbf{r}$ encoded by Bob.
From 
(\ref{aaax1.2}),
we have
\begin{eqnarray}
\label{aaay3.1}
P_{\text{error}}&=&\sum_{\lvert \Delta_b\rvert<\nu_\text{cor} N}P_{\text{error}}(\lvert \Delta_b\rvert)\text{Pr}(\lvert \Delta_b\rvert)\nonumber\\
&&\quad+ \sum_{\lvert \Delta_b\rvert\geq \nu_\text{cor} N}P_{\text{error}}(\lvert \Delta_b\rvert)\text{Pr}(\lvert \Delta_b\rvert)\nonumber\\
&\leq&\text{Pr}\big[ \lvert \Delta_b\rvert < \nu_\text{cor} N\bigr]\nonumber\\
&&\quad+ \sum_{\lvert \Delta_b\rvert\geq \nu_\text{cor} N}P_{\text{error}}(\lvert \Delta_b\rvert)\text{Pr}(\lvert \Delta_b\rvert).
\end{eqnarray}
We show below that
\begin{equation}
\label{aaay3.2x}
P_{\text{error}}(\lvert \Delta_b\rvert )<e^{-\frac{E \lvert \Delta_b\rvert }{3}\bigl(\frac{\gamma_\text{err}}{E}-1\bigr)^2},
\end{equation}
and that
\begin{equation}
\label{aaay3.2}
\text{Pr}[\lvert \Delta_b\rvert\!<\!\nu_\text{cor} N]\leq e^{-\frac{P_{\text{det}}(1-2\beta_\text{PB})N}{4}\bigl(1-\frac{2\nu_\text{cor}}{P_{\text{det}}(1-2\beta_\text{PB})}\bigr)^2}.
\end{equation}
From (\ref{aaay3.1}) -- (\ref{aaay3.2}), and noting that $e^{-\frac{E \lvert \Delta_b\rvert}{3}\bigl(\frac{\gamma_\text{err}}{E}-1\bigr)^2}$ decreases with increasing $\lvert \Delta_b\rvert$, we obtain
\begin{equation}
			\label{neweq2}
			P_\text{error}<\epsilon_\text{cor},
						\end{equation}
			with $\epsilon_\text{cor}$ given by (\ref{cor2repeated}). Thus, the token schemes $\mathcal{QT}_1$ and $\mathcal{QT}_2$ are $\epsilon_\text{cor}-$correct with $\epsilon_\text{cor}$ given by (\ref{cor2repeated}), as claimed.

We show (\ref{aaay3.2x}). 
Let us assume for now that $E_{tu}=E$ for $t,u\in\{0,1\}$. Given $\lvert \Delta_b\rvert$, we note that the expectation value of $n_\text{error}$ equals $E\lvert \Delta_b\rvert$. From (\ref{cor1repeated}), we have $0<\frac{\gamma_{\text{err}}}{E}-1<1$. Thus, from a Chernoff bound of Proposition \ref{proposition3}, we have
\begin{equation}
\label{aaay2}
\text{Pr}\bigl[n_{\text{errors}}>\lvert \Delta_b\rvert \gamma_\text{err}\big\vert \lvert \Delta_b\rvert \bigr]<e^{-\frac{E \lvert \Delta_b\rvert }{3}\bigl(\frac{\gamma_\text{err}}{E}-1\bigr)^2}.
\end{equation}

The function $f(E)=E\bigl(\frac{\gamma_\text{err}}{E}-1\bigr)^2$ is decreasing with increasing $E$, because from (\ref{cor1repeated}) we have that $f'(E)=1-(\frac{\gamma_{\text{error}}}{E})^2<0$. 
Let $E_{\text{max}}\geq E$. Thus, from 
(\ref{aaay2}), we have
\begin{equation}
\label{aaay2.2}
\text{Pr}\bigl[n_{\text{errors}}\!>\!\lvert \Delta_b\rvert \gamma_\text{err}\big\vert \lvert \Delta_b\rvert \bigr]\!<\!e^{-\frac{E_{\text{max}} \lvert \Delta_b\rvert }{3}\bigl(\!\frac{\gamma_\text{err}}{E_{\text{max}}}-1\!\bigr)^2}.
\end{equation}
It follows from (\ref{aaax1.1}) and (\ref{aaay2.2}) that 
\begin{equation}
\label{aaay3.2xx}
P_{\text{error}}(\lvert \Delta_b\rvert )<e^{-\frac{E_{\text{max}} \lvert \Delta_b\rvert }{3}\bigl(\frac{\gamma_\text{err}}{E_{\text{max}}}-1\bigr)^2}.
\end{equation}
Since in general we have $E_{tu}\leq E$, for $t,u\in\{0,1\}$, we can replace $E_\text{max}$ by $E$ in (\ref{aaay3.2xx}) and obtain (\ref{aaay3.2x}).

We show (\ref{aaay3.2}). Since for the quantum state $\lvert \psi_k\rangle$, with $g(k)=j$, for $k\in\Lambda $ and $j\in[n]$, $\mathcal{B}$ encodes the bit $t_k=r_j$ in the basis labelled by $u_k=s_j$, with $u_k$ chosen with probability $P^{k}_{\text{PB}}(u_k)$ satisfying $\frac{1}{2}-\beta_\text{PB}\leq P^{k}_{\text{PB}}(u_k)\leq \frac{1}{2}+\beta_\text{PB}$ for $t_k,u_k\in\{0,1\}$, the expectation value $E(\lvert \Delta_b\rvert)$ of $\lvert \Delta_b\rvert$ satisfies
\begin{equation}
\label{neweq3}
E(\lvert \Delta_b\rvert)\geq P_{\text{det}}N\Bigl(\frac{1}{2}-\beta_\text{PB}\Bigr).
\end{equation}
This is easily seen as follows. By the definition of $\Delta_b$ given in the token schemes $\mathcal{QT}_1$ and $\mathcal{QT}_2$, we see that $\lvert\Delta_b\rvert$ corresponds to the number of labels $k\in\Lambda$ satisfying $g(k)=j\in[n]$ for which it holds that $y_j=s_j$, where we recall $y_j$ and $s_j$ are the bits labelling the qubit measurement basis by Alice and the preparation basis by Bob, respectively. Thus, $E(\lvert \Delta_b\rvert)=P_{\text{det}}N\text{Pr}[y_j=s_j]=P_{\text{det}}N\sum_{a=0}^1 \text{Pr}[s_j=a]\text{Pr}[y_j=a] \geq P_{\text{det}}N\Bigl(\frac{1}{2}-\beta_\text{PB}\Bigr)$, as claimed. We define
\begin{equation}
\label{aaay3.5}
 \epsilon=1-\frac{2\nu_\text{cor}}{P_{\text{det}}(1-2\beta_\text{PB})}.
\end{equation}
From the condition (\ref{cor1repeated}), we have $0<\epsilon<1$. It follows that 
\begin{equation}
\label{aaay3.4}
\nu_\text{cor} N=(1-\epsilon)P_{\text{det}}N\Bigl(\frac{1}{2}-\beta_\text{PB}\Bigr)=(1-\epsilon')E(\lvert \Delta_b\rvert),
\end{equation}
for some $\epsilon'$ satisfying $0<\epsilon\leq \epsilon'<1$. Thus, from the Chernoff bound of Proposition \ref{proposition3}, we have
\begin{eqnarray}
\label{aaay3.3}
\text{Pr}[\lvert \Delta_b\rvert<\nu_\text{cor} N]&=&\text{Pr}[\lvert \Delta_b\rvert<(1-\epsilon')E(\lvert \Delta_b\rvert)]\nonumber\\
&\leq&e^{-\frac{E(\lvert \Delta_b\rvert)}{2}\epsilon'^2}\nonumber\\
&\leq&e^{-\frac{P_{\text{det}}N(1-2\beta_\text{PB})}{4}\epsilon^2}\nonumber\\
&=&e^{-\frac{P_{\text{det}}(1-2\beta_\text{PB})N}{4}\bigl(1-\frac{2\nu_\text{cor}}{P_{\text{det}}(1-2\beta_\text{PB})}\bigr)^2},\nonumber\\
\end{eqnarray}
where in the first line we used (\ref{aaay3.4}); in the second line we used the Chernoff bound of Proposition \ref{proposition3}; in the third line we used (\ref{neweq3}) and $0<\epsilon\leq\epsilon'$ ; and in the last line we used (\ref{aaay3.5}).

\end{proof}

\subsection{Proof of Lemma \ref{Bob1}}

\begin{lemma4*}
\label{lemma4repeated}
$\mathcal{QT}_1$ and  $\mathcal{QT}_2$ are $\epsilon_{\text{priv}}-$private with
\begin{equation}
\label{Bo1repeated}
\epsilon_{\text{priv}}=\beta_{\text{E}}.
\end{equation}
\end{lemma4*}

\begin{proof}
From assumption C (see Table \ref{tableassu}), the set $\Lambda$ of labels transmitted to $\mathcal{B}$ in step 2 of $\mathcal{QT}_1$
and $\mathcal{QT}_1$ gives $\mathcal{B}$ no information about the string $W$
and the bit $z$. Furthermore, from assumption E (see Table \ref{tableassu}), 
$\mathcal{B}$ cannot use degrees of freedom
not previously agreed for the transmission of the quantum states to affect, or obtain information about, the statistics of the quantum measurement devices of $\mathcal{A}$. Moreover, in our setting, we assume that Alice's laboratories are secure and that communication among Alice's agents is made through secure and authenticated classical channels. It follows from these assumptions that the only way in which Bob can obtain information about Alice's bit $b$ before she presents the token is via the message $c=z\oplus b$.

In order to prove our result, let us assume that Bob knows Alice's probability distributions $P_{\text{E}}(z)$. Since this cannot make it more difficult for Bob to guess Alice's bit $b$, we can assume this without loss of generality. In the ideal case that the probability distribution $P_{\text{E}}(z)$ is totally random Bob cannot obtain any information about $b$. However, as stated by our allowed experimental imperfection 7 (see Table \ref{tableimp}), this probability distribution is only close to random: 
\begin{eqnarray}
\label{aaaaaa}
\frac{1}{2}-\beta_{\text{E}}\leq P_{\text{E}}(z)&\leq& \frac{1}{2}+\beta_{\text{E}},
\end{eqnarray}
for a small parameters $\beta_{\text{E}}>0$, for $z\in\{0,1\}$. Thus, Bob can guess $b$ with some probability greater than $\frac{1}{2}$.

Let $P_{\text{bit}}^{(i)}(c)$ be the probability distribution for the bit $c$ that $\mathcal{A}$ sends $\mathcal{B}$, when $b=i$, for $i,c\in\{0,1\}$. Since $c=b\oplus z$, we have
\begin{equation}
\label{aaaaaa1}
P_{\text{bit}}^{(i)}(c)=P_{\text{E}}(z=i\oplus c),
\end{equation}
for $c,i\in\{0,1\}$.

For any two probability distributions $P(x)$ and $Q(x)$ over a set of values $x\in\mathcal{X}$, the maximum probability $P_{\text{max}}$ to distinguish them is given by $P_{\text{max}}=\frac{1}{2}+\frac{1}{2}\lVert P-Q\rVert$, where $\lVert P - Q\rVert $ is their variational distance. Thus, Bob's probability $P_{\text{Bob}}$ to guess Alice's bit $b$ is upper bounded by
\begin{eqnarray}
\label{aaaaaaz22}
P_{\text{Bob}}&\leq& \frac{1}{2}+\frac{1}{2}\bigl\lVert P_{\text{bit}}^{(0)} - P_{\text{bit}}^{(1)}\bigr\rVert \nonumber\\
&=&\frac{1}{2}+\frac{1}{4}\sum_{c=0}^1\bigl\lvert P_{\text{bit}}^{(0)}(c)-P_{\text{bit}}^{(1)}(c)\bigr\rvert\nonumber\\
&=&\frac{1}{2}+\frac{1}{4}\sum_{c=0}^1\bigl\lvert P_{\text{E}}(z=c) - P_{\text{E}}(z= c\oplus 1)\bigr\rvert\nonumber\\
&\leq&\frac{1}{2}+\frac{1}{4}\sum_{c=0}^1\bigl\lvert 2\beta_\text{E}\bigr\rvert\nonumber\\
&=&\frac{1}{2}+\beta_\text{E},
\end{eqnarray}
where in the second line we used the definition of the variational distance; in the third line we used (\ref{aaaaaa1}); and in the fourth line we used (\ref{aaaaaa}).

It follows from (\ref{aaaaaaz22}) that the token schemes $\mathcal{QT}_1$ and $\mathcal{QT}_2$ are $\epsilon_{\text{priv}}$-private, with $\epsilon_{\text{priv}}$ given by (\ref{Bo1repeated}), as claimed.
\end{proof}

\section{Proof of Lemma \ref{lemmax}}
\label{appnewlemma}

\begin{lemma5*}
\label{lemmaxrepeated}
Suppose that Bob sends Alice $N$ photon pulses, labelled by $k\in[N]$. Let the $k$th pulse have $L_k$ photons. Let $\rho$ be an arbitrary quantum state prepared by Bob in the polarization degrees of freedom of the photons sent to Alice, which can be arbitrarily entangled among all photons in all pulses and can also be arbitrarily entangled with an ancilla held by Bob. Let $\mathcal{D}_0$ and $\mathcal{D}_1$ be two arbitrary qubit orthogonal bases. Suppose that either Alice uses the setup of Fig. \ref{setup} with reporting strategy 1 to implement the quantum token scheme $\mathcal{QT}_1$ (see Table \ref{real1}), or Alice uses the setup of \damian{Fig. \ref{setup}} with reporting strategy 2 to implement the quantum token scheme $\mathcal{QT}_2$ (see Table \ref{real2}). Suppose also that assumptions E and F (see Table \ref{tableassu}) hold.
For $k\in[N]$, let $m_k=1$ if Alice assigns a successful measurement to the $k$th pulse and $m_k=0$ otherwise; let $w_k=0$ ($w_k=1$) if Alice assigns a measurement basis to the $k$th pulse in the basis $\mathcal{D}_{0}$ ($\mathcal{D}_{1}$). If Alice uses the setup of Fig. \ref{setup} and reporting strategy 1 to implement the scheme $\mathcal{QT}_1$, without loss of generality, suppose also that Alice sets $w_k=0$ with unit probability, if $m_k=0$, for $k\in[N]$. Let $m=(m_1,\ldots,m_N)$, $w=(w_1,\ldots,w_N)$ and $L=(L_1,\ldots,L_N)$.

If Alice uses the setup of Fig. \ref{setup} with reporting strategy 1 to implement the scheme $\mathcal{QT}_1$, then the probability that Alice reports the string  $m$ to Bob and assigns the string of measurement bases $w$, given $\rho$  and $L$, is
\begin{equation}
\label{hhhrepeated}
P^{(1)}_{\text{rep}}(m,w\lvert \rho, L)=\prod_{k=1}^N G^{(1)}_{m_k,w_k}(d_0,d_1,\eta,L_k),
\end{equation}
where
\begin{eqnarray}
\label{hhh0repeated}
G^{(1)}_{1,b}(d_0,d_1,\eta,a)&=&(1-d_0)(1-d_1)\Bigl(1-\frac{\eta}{2}\Bigr)^{a}\nonumber\\
&&\qquad-(1-d_0)^2(1-d_1)^2(1-\eta)^{a},\nonumber\\
G^{(1)}_{0,0}(d_0,d_1,\eta,a)&=&1-2G^{(1)}_{1,0}(d_0,d_1,\eta,a),\nonumber\\
G^{(1)}_{0,1}(d_0,d_1,\eta,a)&=&0,
\end{eqnarray}
for $b\in\{0,1\}$, $m,w\in\{0,1\}^N$ and $a,L_1,\ldots,L_N\in\{0,1,2,\ldots\}$. Furthermore, the probability $P_\text{MB}(w_k)$ that Alice assigns a measurement in the basis $\mathcal{D}_{w_k}$, conditioned on the value $m_k=1$, for the $k$th pulse, satisfies  
\begin{equation}
\label{h000.1repeated}
P_\text{MB}(w_k)=\frac{1}{2},
\end{equation}
for $w_k\in\{0,1\}$ and $k\in[N]$.

If Alice uses the setup of \damian{Fig. \ref{setup}} with reporting strategy 2 to implement the scheme $\mathcal{QT}_2$, then the probability that Alice reports the string  $m$ to Bob, given $\rho$, $w$ and $L$, is
\begin{equation}
\label{gggrepeated}
P^{(2)}_{\text{rep}}(m\lvert w, \rho, L)=\prod_{k=1}^N G^{(2)}_{m_k}(d_0,d_1,\eta,L_k),
\end{equation}
where
\begin{eqnarray}
\label{ggg0repeated}
G^{(2)}_0(d_0,d_1,\eta,a)&=&(1\!-\!d_0)(1\!-\!d_1)(1\!-\!\eta)^{a},\nonumber\\
G^{(2)}_1(d_0,d_1,\eta,a)&=&1-(1\!-\!d_0)(1\!-\!d_1)(1\!-\!\eta)^{a},
\end{eqnarray}
for $m,w\in\{0,1\}^N$ and $a,L_1,\ldots,L_N\in\{0,1,2,\ldots\}$.

In any of the two cases, the message $m$ gives Bob no information about the bit entries $w_k$ for which $m_k=1$. Equivalently, the set $\Lambda\subset[N]$ of labels transmitted to Bob in step 2 of $\mathcal{QT}_1$ and $\mathcal{QT}_2$ gives Bob no information about the string $W$ and the bit $z$.

\end{lemma5*}

\begin{proof}
We note from (\ref{hhhrepeated}) and (\ref{hhh0repeated}) that if Alice uses the setup of Fig. \ref{setup} with reporting strategy 1 to implement the scheme $\mathcal{QT}_1$, then Alice's probability $P^{(1)}_{\text{rep}}(m,w\lvert \rho, L)$ to report the message $m$ to Bob and assign measurement basis with string of labels $w$ is the same for all strings $w$ satisfying that $w_k=0$ if $m_k=0$, for arbitrary fixed given values of $m$, $\rho$ and $L$. It follows from this and from assumption E (see Table \ref{tableassu}) that the message $m$ gives Bob no information about the bit entries $w_k$ for which $m_k=1$.

Similarly, we note from (\ref{gggrepeated}) and (\ref{ggg0repeated}) that if Alice uses the setup of \damian{Fig. \ref{setup}} with reporting strategy 2 to implement the scheme $\mathcal{QT}_2$, then Alice's probability $P^{(2)}_{\text{rep}}(m\lvert w,\rho, L)$ to report the message $m$ to Bob, given that she
applied quantum measurements with string of labels $w$, is the same for all strings $w$, for arbitrary fixed given values of $m$, $\rho$ and $L$. It follows from this and from assumption E (see Table \ref{tableassu}) that the message $m$ gives Bob no information about any bit entries $w_k$ of $w$. We note that in the scheme $\mathcal{QT}_2$, $w_k=z$ for $k\in[N]$, and for some bit $z$ chosen by Alice. Thus, it follows that the message $m$ gives Bob no information about the bit $z$.

Alice sending the message $m$ to Bob is equivalent to Alice sending the set $\Lambda$ of labels $k\in[N]$ for which $m_k=1$, i.e. the labels of pulses that were successfully measured by Alice. Furthermore, the string $W$ is defined on the set of labels $\Lambda$, and has entries equal to $w_k$, for $k\in\Lambda$, i.e. for $k\in[N]$ satisfying $m_k=1$. It follows that the set $\Lambda\subset[N]$ of labels transmitted to Bob in step 2 of $\mathcal{QT}_1$ and $\mathcal{QT}_2$ gives Bob no information about the string $W$ and the bit $z$.

We prove (\ref{hhhrepeated}). We suppose that Alice uses the setup of Fig. \ref{setup} with reporting strategy 1 to implement the scheme $\mathcal{QT}_1$, and that assumptions E and F of Table \ref{tableassu} hold.
It is straightforward to obtain
\begin{equation}
\label{hhh1}
P^{(1)}_{\text{rep}}(m,w\lvert \rho, L)=\prod_{k=1}^N P^{(1,k)}_{\text{rep}}(m_k,w_k\lvert  \rho, L,\tau_k),
\end{equation}
where $\tau_k=(m_0,w_0,\ldots,m_{k-1},w_{k-1})$, $P^{(1,k)}_{\text{rep}}(m_k,w_k\lvert \rho, L,\tau_k)$ is the probability that Alice reports the bit message $m_k$ to Bob and assigns a measurement in the basis $\mathcal{D}_{w_k}$ for the $k$th pulse, given $\rho$, $L$,  and $\tau_k$, for $k\in[N]$; and where without loss of generality we define $m_0=w_0=1$.

From Lemma 12 of Ref. \cite{BCDKPG21} and the definition (\ref{hhh0repeated}), after measuring the first pulse received from Bob, Alice reports the message $m_1=1$ to Bob and assigns a measurement outcome in the basis $\mathcal{D}_{w_1}$ with probability 
\begin{equation}
\label{hhh2}
P^{(1,1)}_{\text{rep}}(1,w_1\lvert  \rho, L,\tau_1)=G^{(1)}_{1,w_1}(d_0,d_1,\eta,L_1),
\end{equation}
for $w_1\in\{0,1\}$, which only depends on the dark count probabilities $d_0$ and $d_1$, on the detector efficiency $\eta$, and on the number of photons $L_1$ of the first pulse, for an arbitrary quantum state $\rho$, which can be arbitrarily entangled with an ancilla held by Bob. We can consider this ancilla to include the pulses labelled by $2,3,\ldots,N$.

Similarly, for $k\in\{2,3,\ldots,N\}$, after measuring the pulses with labels $1,2,\ldots,k-1$, Alice obtains a value $\tau_{k}$ according to her obtained detection statistics for these pulses, and the joint quantum state of the pulses with labels $k,k+1,\ldots,N$ and any ancilla held by Bob changes to some quantum state $\rho_k$. Then, from Lemma 12 of Ref. \cite{BCDKPG21} and the definition (\ref{hhh0repeated}), after measuring the $k$th pulse, Alice reports the message $m_k=1$ to Bob and assigns a measurement outcome in the basis $\mathcal{D}_{w_k}$ with probability 
\begin{equation}
\label{hhh3}
P^{(1,k)}_{\text{rep}}(1,w_{k}\lvert  \rho, L,\tau_{k})=G^{(1)}_{1,w_k}(d_0,d_1,\eta,L_{k}),
\end{equation}
for $w_{k}\in\{0,1\}$, which only depends on the dark count probabilities $d_0$ and $d_1$, on the detector efficiency $\eta$, and on the number of photons $L_{k}$ of the $k$th pulse; in particular, $P^{(1,k)}_{\text{rep}}(1,w_{k}\lvert  \rho, L,\tau_{k})$ does not depend on the quantum state $\rho_k$, which can be arbitrarily entangled with an ancilla held by Bob. In this case, we can consider this ancilla to include the pulses labelled by $k+1,k+2,\ldots,N$. We note that since $P^{(1,k)}_{\text{rep}}(1,w_{k}\lvert  \rho, L,\tau_{k})$ does not depend on $\rho_k$, it does not depend on $\rho$, apart from the number of photons $L_{k}$ of the $k$th pulse, and it does not depend on $\tau_k$ either.

By definition of Alice's reporting strategy 1, we have
\begin{eqnarray}
\label{hhh4}
P^{(1,k)}_{\text{rep}}(0,0\lvert  \rho, L,\tau_{k})&=&1-P^{(1,k)}_{\text{rep}}(1,0\lvert  \rho, L,\tau_{k})\nonumber\\
&&\qquad-P^{(1,k)}_{\text{rep}}(1,1\lvert  \rho, L,\tau_{k})\nonumber\\
&=&1-2G^{(1)}_{1,0}(d_0,d_1,\eta,L_{k})\nonumber\\
&=&G^{(1)}_{0,0}(d_0,d_1,\eta,L_{k}),
\end{eqnarray}
for $k\in[N]$, where in the second line we used (\ref{hhh2}) and (\ref{hhh3}), and the definition (\ref{hhh0repeated}); and in the third line we used (\ref{hhh0repeated}) again. Similarly, by definition of Alice's reporting strategy and from the definition (\ref{hhh0repeated}), we have
\begin{equation}
\label{hhh5}
P^{(1,k)}_{\text{rep}}(0,1\lvert  \rho, L,\tau_{k})=G^{(1)}_{0,1}(d_0,d_1,\eta,L_{k}),
\end{equation}
for $k\in[N]$. Thus, the claimed result (\ref{hhhrepeated}) follows straightforwardly from (\ref{hhh1}) -- (\ref{hhh4}).

We prove (\ref{h000.1repeated}). Let $P^{(1,k)}_\text{MB}(w_k\lvert m_k=1, \rho, L,\tau_{k})$ be the probability that Alice assigns a measurement in the basis $\mathcal{D}_{w_k}$, conditioned on the value $m_k=1$, for the $k$th pulse, given the values of $\rho$, $L$ and $\tau_k$, for $k\in[N]$. We have
\begin{eqnarray}
\label{hhh6}
&&P^{(1,k)}_\text{MB}(w_k\lvert m_k=1,  \rho, L,\tau_{k})\nonumber\\
&&\qquad\qquad=\frac{P^{(1,k)}_{\text{rep}}(1,w_k\lvert  \rho, L,\tau_{k})}{P^{(1,k)}_{\text{rep}}(1,0\lvert  \rho, L,\tau_{k})+P^{(1,k)}_{\text{rep}}(1,1\lvert  \rho, L,\tau_{k})}\nonumber\\
&&\qquad\qquad=\frac{G^{(1)}_{1,w_k}(d_0,d_1,\eta,L_k)}{2G^{(1)}_{1,w_k}(d_0,d_1,\eta,L_k)}\nonumber\\
&&\qquad\qquad=\frac{1}{2},
\end{eqnarray}
for $w_k\in\{0,1\}$ and $k\in[N]$; where in the second line we used (\ref{hhh2}) and (\ref{hhh3}), and the definition (\ref{hhh0repeated}); and in the third line we used that $G^{(1)}_{1,w_k}(d_0,d_1,\eta,L_k)>0$ from (\ref{hhh0repeated}) and from the fact that $d_0,d_1,\eta\in(0,1)$, as stated in assumption F (see Table \ref{tableassu}). From (\ref{hhh6}), since $P^{(1,k)}_\text{MB}(w_k\lvert m_k=1,  \rho, L,\tau_{k})$ does not depend on $k$, $\rho$, $L$ or $\tau_k$, we have that 
\begin{equation}
\label{hhh7}
P^{(1,k)}_\text{MB}(w_k\lvert m_k=1, \rho, L,\tau_{k})=P_\text{MB}(w_k),
\end{equation}
for $w_k\in\{0,1\}$ and $k\in[N]$, and the claimed result (\ref{h000.1repeated}) follows.

We prove (\ref{gggrepeated}). We suppose that Alice uses the setup of \damian{Fig. \ref{setup}} with reporting strategy 2 to implement the scheme $\mathcal{QT}_2$, and that assumptions E and F of Table \ref{tableassu} hold. We note that in the scheme $\mathcal{QT}_2$, the string $w$ has bit entries $w_k=z$, for a bit $z$ chosen by Alice, and for $k\in[N]$. However, the analysis below is more general, and works for arbitrary $w\in\{0,1\}^N$. It is straightforward to obtain
\begin{equation}
\label{ggg1}
P^{(2)}_{\text{rep}}(m\lvert w, \rho, L)=\prod_{k=1}^N P^{(2,k)}_{\text{rep}}(m_k\lvert w, \rho, L,\tilde{\tau}_k),
\end{equation}
where $\tilde{\tau}_k=(m_0,\ldots,m_{k-1})$ and $P^{(2,k)}_{\text{rep}}(m_k\lvert w, \rho, L,\tilde{\tau}_k)$ is the probability that Alice reports the bit message $m_k$ to Bob for the $k$th pulse, given $\rho$, $L$, $w$ and $\tilde{\tau}_k$, for $k\in[N]$; and where without loss of generality we define $m_0=1$. 

From Lemma 1 of Ref. \cite{BCDKPG21} and the definition (\ref{ggg0repeated}), after measuring the first pulse received from Bob in the basis $\mathcal{D}_{w_1}$, Alice reports the message $m_1=1$ to Bob with probability 
\begin{equation}
\label{ggg2}
P^{(2,1)}_{\text{rep}}(1\lvert  w,\rho, L,\tilde{\tau}_1)=G^{(2)}_{1}(d_0,d_1,\eta,L_1),
\end{equation}
which only depends on the dark count probabilities $d_0$ and $d_1$, on the detector efficiency $\eta$, and on the number of photons $L_1$ of the first pulse, for an arbitrary quantum state $\rho$, which can be arbitrarily entangled with an ancilla held by Bob. We can consider this ancilla to include the pulses labelled by $2,3,\ldots,N$.

Similarly, for $k\in\{2,3,\ldots,N\}$, after measuring the pulses with labels $1,2,\ldots,k-1$ in the bases $\mathcal{D}_{w_1},\ldots,\mathcal{D}_{w_{k-1}}$, Alice obtains a value $\tilde{\tau}_{k}$ according to her obtained detection statistics for these pulses, and the joint quantum state of the pulses with labels $k,k+1,\ldots,N$ and any ancilla held by Bob changes to some quantum state $\rho_k$. Then, from Lemma 1 of Ref. \cite{BCDKPG21} and the definition (\ref{ggg0repeated}), after measuring the $k$th pulse in the basis $\mathcal{D}_{w_k}$, Alice reports the message $m_k=1$ to Bob  with probability 
\begin{equation}
\label{ggg3}
P^{(2,k)}_{\text{rep}}(1\lvert  w,\rho, L,\tilde{\tau}_{k})=G^{(2)}_{1}(d_0,d_1,\eta,L_{k}),
\end{equation}
which only depends on the dark count probabilities $d_0$ and $d_1$, on the detector efficiency $\eta$, and on the number of photons $L_{k}$ of the $k$th pulse; in particular, $P^{(2,k)}_{\text{rep}}(1\lvert w,  \rho, L,\tilde{\tau}_{k})$ does not depend on the quantum state $\rho_k$, which can be arbitrarily entangled with an ancilla held by Bob. In this case, we can consider this ancilla to include the pulses labelled by $k+1,k+2,\ldots,N$. We note that since $P^{(2,k)}_{\text{rep}}(1\lvert  w,\rho, L,\tilde{\tau}_{k})$ does not depend on $\rho_k$, it does not depend on $\rho$, apart from the number of photons $L_{k}$ of the $k$th pulse, and it does not depend on $\tilde{\tau}_k$ either.

By definition of Alice's reporting strategy 2, we have
\begin{eqnarray}
\label{ggg4}
&&P^{(2,k)}_{\text{rep}}(0\lvert w, \rho, L,\tilde{\tau}_k)\nonumber\\
&&\qquad\qquad=1-P^{(2,k)}_{\text{rep}}(1\lvert w, \rho, L,\tilde{\tau}_k)\nonumber\\
&&\qquad\qquad=1-G^{(2)}_{1}(d_0,d_1,\eta,L_{k})\nonumber\\
&&\qquad\qquad=G^{(2)}_{0}(d_0,d_1,\eta,L_{k}),
\end{eqnarray}
for $k\in[N]$, where in the second line we used (\ref{ggg2}) and (\ref{ggg3}), and in the third line we used the definition (\ref{ggg0repeated}). Thus, the claimed result (\ref{gggrepeated}) follows straightforwardly from (\ref{ggg1}) -- (\ref{ggg4}).

\end{proof}

\section{Proof of Theorem \ref{Alice1}}
\label{newapp}

\begin{theorem1*}
\label{Alice1repeated}
Consider the constraints
\begin{eqnarray}
\label{Al1repeated}
0&<&\gamma_\text{err}<\lambda(\theta,\beta_{\text{PB}}),\nonumber\\
0&<&P_\text{noqub}<\nu_\text{unf}<\min\biggl\{\!2P_\text{noqub},\gamma_\text{det}\biggl(\!1\!-\!\frac{\gamma_\text{err}}{\lambda(\theta,\beta_{\text{PB}})}\!\biggr)\!\biggr\},\nonumber\\
		0&<&\beta_\text{PS}<\frac{1}{2}\biggl[e^{\frac{\lambda(\theta,\beta_{\text{PB}})}{2}\bigl(1-\frac{\delta}{\lambda(\theta,\beta_{\text{PB}})}\bigr)^2}-1\biggr].
	\end{eqnarray}
	We define the function
\begin{eqnarray}
\label{Al3repeated}	
	&&f(\!\gamma_\text{err},\!\beta_\text{PS},\!\beta_\text{PB},\!\theta,\!\nu_\text{unf},\!\gamma_\text{det})\nonumber\\
	&&\quad=\!(\gamma_\text{det}\!-\!\nu_\text{unf})\!\biggl[\!\frac{\lambda(\theta,\beta_{\text{PB}}) }{2}\!\biggl(\!1\!-\!\frac{\delta}{\lambda(\theta,\beta_{\text{PB}})}\!\biggr)^2\!\!-\!\ln(1\!+\!2\beta_\text{PS})\!\biggr]\nonumber\\
	&&\quad\qquad-\bigl(1\!-\!(\gamma_\text{det}\!-\!\nu_\text{unf})\bigr)\ln\bigl[1+h(\beta_\text{PS},\beta_\text{PB},\theta)\bigr],
	\end{eqnarray}	
	 where
	 \begin{eqnarray}
\label{Al4repeated}
h(\beta_\text{PS},\beta_\text{PB},\theta)&=&2\beta_\text{PS}\sqrt{\frac{1}{2}\!+\!2\beta_\text{PB}^2\!+\!\Bigl(\frac{1}{2}\!-\!2\beta_\text{PB}^2\Bigr)\sin(2\theta)},\nonumber\\
\delta&=&\frac{\gamma_\text{det}\gamma_\text{err}}{\gamma_\text{det}-\nu_\text{unf}}.
\end{eqnarray}
There exist parameters satisfying the constraints (\ref{Al1repeated}), for which $f(\!\gamma_\text{err},\!\beta_\text{PS},\!\beta_\text{PB},\!\theta,\!\nu_\text{unf},\!\gamma_\text{det})>0$. For these parameters, $\mathcal{QT}_1$ and $\mathcal{QT}_2$ are $\epsilon_{\text{unf}}-$unforgeable with
\begin{equation}
\label{Al5repeated}
\epsilon_{\text{unf}}=e^{-\frac{P_\text{noqub}N}{3}\bigl(\frac{\nu_\text{unf}}{P_\text{noqub}}-1\bigr)^2} +e^{-Nf(\gamma_\text{err},\beta_\text{PS},\beta_\text{PB},\theta,\nu_\text{unf},\gamma_\text{det})}.	
\end{equation}
\end{theorem1*}

We recall that Theorem \ref{Alice1} considers parameters $\gamma_\text{det},\gamma_\text{err}\in(0,1)$, allows for the experimental imperfections of Table \ref{tableimp} and makes the assumptions of Table \ref{tableassu}.

\subsection{Summary of the Proof}
We allow Alice to have arbitrarily advanced quantum technology. In particular, we assume that Alice knows the set $\Omega_{\text{qub}}$ and $\Omega_{\text{noqub}}$ and that she receives all quantum systems $A_k$, for $k\in[N]$. Let $\lvert \psi\rangle_A=\otimes_{k\in[N]}\lvert \psi_k\rangle_{A_k}$ be the quantum state that Alice receives from Bob, where $A$ is the global quantum system received from Bob.

Alice's more general cheating strategy in the token scheme $\mathcal{QT}_1$ is as follows. In the intersection of the causal pasts of $Q_0$ and $Q_1$, Alice adds an ancillary system $E$ of arbitrary finite Hilbert space dimension in a quantum state $\lvert\chi\rangle_E$ and applies an arbitrary projective measurement on $AE$, which may depend on $\Omega_{\text{qub}}$ and $\Omega_{\text{noqub}}$, and obtains an outcome that includes $\Lambda,g,\mathbf{d}$ and $c$ satisfying the required constraints, as well as respective tokens $\mathbf{a}$ and $\mathbf{b}$ to give at the presentation points $Q_0$ and $Q_1$. Alice sends $ \Lambda$, $g$, $\mathbf{d}$, and $c$ to Bob, as required by the task. Alice gives Bob tokens $\mathbf{a}$ at $Q_0$ and $\mathbf{b}$ at $Q_1$. Alice's more general cheating strategy in the token scheme $\mathcal{QT}_2$ is equivalent, with the only difference that the string $\mathbf{d}$ is not required in Alice's measurement outcome. We recall that the $j$th entry of $\mathbf{d}$ in $\mathcal{QT}_1$ is $d_j=y_j\oplus z$, for $j\in[n]$. On the other hand, in $\mathcal{QT}_2$ we have $y_j=z$, for $j\in[n]$. Thus, without loss of generality, in Alice's general cheating strategy  in the token scheme $\mathcal{QT}_2$ we simply set $\mathbf{d}$ to be a fixed string with all bit entries being zero.

We note that if  $Q_1$ is in the causal future of $Q_0$ Alice's agent $\mathcal{A}_0$ can send a signal to $\mathcal{A}_1$ indicating whether her token $\mathbf{a}$ was validated or not at $Q_0$, and $\mathcal{A}_1$ can in principle use this information to adapt her strategy at $Q_1$. However, this possibility cannot increase Alice's probability to have tokens validated at both $Q_0$ and $Q_1$. If $\mathcal{A}_0$ fails in having $\mathbf{a}$ validated at $Q_0$ then Alice fails in her attempt to have Bob validating her tokens at both $Q_0$ and $Q_1$. Thus, without loss of generality we assume that the signal sent from $\mathcal{A}_0$ to $\mathcal{A}_1$ indicates that $\mathbf{a}$ was successfully validated at $Q_0$, which is equivalent to $\mathcal{A}_0$ not sending any signal to $\mathcal{A}_1$ and $\mathcal{A}_1$ always acting as if $\mathbf{a}$ were successfully validated at $Q_0$. The same reasoning applies if $Q_0$ is in the causal future of $Q_1$. Furthermore, if $Q_0$ and $Q_1$ are spacelike separated $\mathcal{A}_1$ cannot receive any signals informing her whether the token $\mathbf{a}$ was successfully validated at $Q_0$ before $\mathcal{A}_1$ presents the token $\mathbf{b}$ at $Q_1$. This means that the strategy outlined in the previous paragraph can be considered as the most general one.

If $\lvert \Omega_{\text{noqub}}\rvert>\nu_\text{unf} N$ then we assume that Alice can succeed in giving valid tokens $\mathbf{a}$ and $\mathbf{b}$ at $Q_0$ and $Q_1$. This is because, for example, if $\lvert \Omega_{\text{noqub}}\rvert$ is too large, then there is not any constraint on the quantum states $\lvert\psi_k\rangle$ for a large number of labels $k$, i.e for $k\in \Omega_{\text{noqub}}$. For example, if Bob's quantum state source is a
Poissonian photon source (e.g. weak coherent) with average photon number $\mu<<1$ then we associate pulses with two or more photons to have labels from the set $\Omega_{\text{noqub}}$, giving $P_\text{noqub}=1-(1+\mu)e^{-\mu}$. In this case, the states $\lvert\psi_k\rangle$ with $k\in\Omega_{\text{noqub}}$ consist of two or more copies of quantum states $\lvert\phi_{t_ku_k}^k\rangle$ and Alice can measure each copy in the corresponding basis, being able to present correct bit outcomes at $Q_0$ and $Q_1$, for bit entries with labels $k\in\Omega_{\text{noqub}}$. But, since the probability $P_\text{noqub}$ that Bob prepares states $\lvert\psi_k\rangle$ with labels $k\in\Omega_{\text{noqub}}$ is bounded, the probability that $\lvert \Omega_{\text{noqub}}\rvert > \nu_\text{unf} N$ is bounded by the first term of $ \epsilon_{\text{unf}}$ in (\ref{Al5repeated}).

On the other hand, we show that there exist parameters satisfying the constraints (\ref{Al1repeated}) for which $f(\gamma_\text{err},\beta_\text{PS},\beta_\text{PB},\theta,\nu_\text{unf},\gamma_\text{det})>0$. We show that, for these parameters, and for the case $\lvert \Omega_{\text{noqub}}\rvert\leq \nu_\text{unf} N$, the probability that Alice gives valid tokens $\mathbf{a}$ and $\mathbf{b}$ at $Q_0$ and $Q_1$ is upper bounded by $e^{-Nf(\gamma_\text{err},\beta_\text{PS},\beta_\text{PB},\theta,\nu_{\text{unf}},\gamma_\text{det})}$. This analysis gives the second term of $\epsilon_{\text{unf}}$ in (\ref{Al5repeated}).

If $\lvert \Omega_{\text{noqub}}\rvert\leq \nu_\text{unf} N$, from the conditions (\ref{Al1repeated}) it follows that Alice must obtain the correct bits in two token strings $\mathbf{a}$ and $\mathbf{b}$, given respectively at $Q_0$ and $Q_1$, for a sufficiently large number of entries $j\in\Delta_0$ with $j=g(k)$ for some $k\in\Omega_{\text{qub}}$, for $\mathbf{a}$, and $j\in\Delta_1$ with $j=g(k)$ for some $k\in\Omega_{\text{qub}}$, for $\mathbf{b}$. That is, Alice could in principle learn perfectly the bit entries of both strings  $\mathbf{a}_0$ and $\mathbf{b}_1$ with labels $k\in\Omega_{\text{noqub}}$, but the number of these entries is small and thus Alice must be able to obtain the correct bit entries of the strings $\mathbf{a}_0$ and $\mathbf{b}_1$ for a sufficiently large number of labels $k\in\Omega_{\text{qub}}$ for which such bits are encoded in single qubits, and not in multiple copies of the same quantum states. The probability that Alice succeeds in this task is upper bounded using Lemma \ref{lemma0}, which considers the ideal situation in which Bob encodes each bit in a single qubit.

We see that the $n-$bit strings $\tilde{\mathbf{d}}_0=(\tilde{d}_{0,1},\ldots,\tilde{d}_{0,n})$ and $\tilde{\mathbf{d}}_1=(\tilde{d}_{1,1},\ldots,\tilde{d}_{1,n})$ are the complement of each other. In the scheme $\mathcal{QT}_1$, we have $\tilde{d}_{1,j}=d_j\oplus c\oplus 1=\tilde{d}_{0,j}\oplus 1$, for $j\in[n]$. Similarly, in the scheme $\mathcal{QT}_2$, we have $\tilde{d}_{1,j}= c\oplus 1= \tilde{d}_{0,j}\oplus 1$, for $j\in[n]$. Thus, in both schemes $\mathcal{QT}_1$ and $\mathcal{QT}_2$, we define the sets of labels $\Delta_0=\{j\in[n] \vert \tilde{d}_{0,j}=s_j\}$ and $\Delta_1=\{j\in[n] \vert \tilde{d}_{1,j}=s_j\}$, which do not intersect, implying that $\lvert \Delta_0\rvert +\lvert \Delta_1\rvert=n$. We define the sets $\underline{\Delta}_0$ and $\underline{\Delta}_1$ as the sets of labels $j\in \Delta_0$ and $j\in \Delta_1$ with $j=g(k)$ for some $k\in\Omega_{\text{qub}}$, respectively. We also define the strings $\underline{\mathbf{a}}_0$ and $\underline{\mathbf{b}}_1$ as the restrictions of the strings $\mathbf{a}_0$ and $\mathbf{b}_1$ to bit entries with labels $j\in \underline{\Delta}_0$ and $j\in\underline{\Delta}_1$, respectively. The strings $\underline{\mathbf{r}}_0$ and $\underline{\mathbf{r}}_1$ are defined similarly. Then, we can upper bound Alice's success probability by the probability that the string $\underline{\mathbf{a}}_0$ and the string $\underline{\mathbf{b}}_1$  satisfy $d(\underline{\mathbf{a}}_0,\underline{\mathbf{r}}_0)+d(\underline{\mathbf{b}}_1,\underline{\mathbf{r}}_1)\leq (\lvert \underline{\Delta}_0\rvert+\lvert \underline{\Delta}_1\rvert)\delta$, where $0<\tilde{\delta}\leq \delta<\lambda(\theta,\beta_{\text{PB}})$, and where $\tilde{\delta}=\frac{(\lvert \Delta_0\rvert+\lvert \Delta_1\rvert)\gamma_\text{err}}{(\lvert \underline{\Delta}_0\rvert+\lvert \underline{\Delta}_1\rvert)}$. Finally, this probability is upper bounded using Lemma \ref{lemma0}.

\subsection{Preliminaries}

We consider that Alice performs an arbitrary cheating strategy $\mathcal{S}$ trying to have Bob validating tokens at $Q_0$ and $Q_1$. Let $P_\mathcal{S}$ be Alice's success probability. We show an upper bound on $P_\mathcal{S}$, for any strategy $\mathcal{S}$. We assume that Alice has arbitrarily advanced quantum technology. In particular, we assume that Alice knows the set $\Omega_{\text{qub}}$, hence also the set $\Omega_{\text{noqub}}$, and that she receives all quantum systems $A_k$ transmitted by Bob, for $k\in[N]$.

Let $P_{\mathcal{S}}^{\Omega_{\text{qub}}}$ be Alice's success probability following the strategy $\mathcal{S}$ given a set $\Omega_{\text{qub}}$, and let $P_{\Omega_{\text{qub}}}$ be the probability that the set $\Omega_{\text{qub}}$ is generated. We have
\begin{eqnarray}
	\label{Nxy3}
	P_{\mathcal{S}}&=&\sum_{m=0}^N \;\: \sum_{\Omega_{\text{qub}} : \lvert \Omega_{\text{noqub}}\rvert =m}P_{\mathcal{S}}^{\Omega_{\text{qub}}} P_{\Omega_{\text{qub}}}\nonumber\\
	&=&\sum_{m\leq \nu_\text{unf} N} \;\: \sum_{\Omega_{\text{qub}} : \lvert \Omega_{\text{noqub}}\rvert =m}P_{\mathcal{S}}^{\Omega_{\text{qub}}} P_{\Omega_{\text{qub}}}\nonumber\\
	&&\qquad+\sum_{m>\nu_\text{unf} N} \;\: \sum_{\Omega_{\text{qub}} : \lvert \Omega_{\text{noqub}}\rvert =m}P_{\mathcal{S}}^{\Omega_{\text{qub}}} P_{\Omega_{\text{qub}}}\nonumber\\
	&\leq& \sum_{m\leq \nu_\text{unf} N} \;\: \sum_{\Omega_{\text{qub}} : \lvert \Omega_{\text{noqub}}\rvert =m}P_{\mathcal{S}}^{\Omega_{\text{qub}}} P_{\Omega_{\text{qub}}}\nonumber\\
	&&\qquad+\text{Pr}[\lvert \Omega_{\text{noqub}}\rvert>\nu_\text{unf} N],
	\end{eqnarray}
where in the third line we used $\text{Pr}[\lvert \Omega_{\text{noqub}}\rvert>\nu_\text{unf} N]=\sum_{m>\nu_\text{unf} N} \sum_{\Omega_{\text{qub}} : \lvert \Omega_{\text{noqub}}\rvert =m} P_{\Omega_{\text{qub}}}$ and the tivial bound $P_{\mathcal{S}}^{\Omega_{\text{qub}}} \leq 1$ for $\lvert \Omega_{\text{noqub}}\rvert>\nu_\text{unf} N$.

Since for each element $k\in[N]$, Bob's agent $\mathcal{B}$ asigns it to be an element of the set $\Omega_{\text{noqub}}$ with probability $P_\text{noqub}$, the probability  that $\Omega_{\text{noqub}}$ has $m$ elements is
\begin{equation}
\sum_{\Omega_{\text{qub}} : \lvert \Omega_{\text{noqub}}\rvert =m}P_{\Omega_{\text{qub}}}=\bigl(\begin{smallmatrix} N\\ m\end{smallmatrix}\bigr)(P_\text{noqub})^m(1-P_\text{noqub})^{N-m},
\end{equation}
for $m\in\{0,1,\ldots,N\}$. To simplify notation, below we write $m=\lvert \Omega_{\text{noqub}}\rvert$. We use the Chernoff bound of Proposition \ref{proposition3} to show below that
	\begin{equation}
	\label{Nxy4}
	\text{Pr}[m>\nu_\text{unf} N]< e^{-\frac{P_{\text{noqub}}N}{3}\bigl(\frac{\nu_\text{unf}}{P_\text{noqub}}-1\bigr)^2} .
	\end{equation}
Let $Z_1,Z_2,\ldots, Z_N$ be independent random variables, where $Z_k\in\{0,1\}$ and $\text{Pr}[Z_k=1]=P_\text{noqub}$ for $k\in[N]$. We can then write $m=\sum_{k=1}^NZ_k$. The expectation value of $m$ is $E(m)=P_\text{noqub}N$. Let $\epsilon=\frac{\nu_\text{unf}}{P_\text{noqub}}-1$. It follows from (\ref{Al1repeated}) that $0<\epsilon<1$. Thus, we have
	\begin{eqnarray}
	\label{Nxy5}
	\text{Pr}[m>\nu_\text{unf} N]&=&\text{Pr}[m>(1+\epsilon)E(m)]\nonumber\\
	&<& e^{-\frac {P_{\text{noqub}}N}{3}\bigl(\frac{\nu_\text{unf}}{P_\text{noqub}}-1\bigr)^2},
	\end{eqnarray}
	as claimed, where in the second line we used $E(m)=P_{\text{noqub}}N$,  $0<\epsilon=\frac{\nu_\text{unf}}{P_\text{noqub}}-1<1$ and the Chernoff bound of Proposition \ref{proposition3}.
	
From (\ref{Nxy3}) and (\ref{Nxy4}), we have
\begin{eqnarray}
	\label{Nq2}
	P_{\mathcal{S}}&<& e^{-\frac{P_\text{noqub}N}{3}\bigl(\frac{\nu_\text{unf}}{P_\text{noqub}}-1\bigr)^2}\nonumber\\
	&&\quad+ \sum_{m\leq \nu_\text{unf} N} \; \sum_{\Omega_{\text{qub}} : \lvert \Omega_{\text{noqub}}\rvert = m}P_{\mathcal{S}}^{\Omega_{\text{qub}}} P_{\Omega_{\text{qub}}}.
	\end{eqnarray}
	Let 
\begin{equation}
\label{Nq1}
\lvert \Psi_{\mathbf{t}\mathbf{u}}^{\Omega_{\text{qub}}}\rangle_A=\bigotimes_{k\in\Omega_{\text{qub}}}\lvert \phi_{t_ku_k}^k\rangle_{A_k}\bigotimes_{k\in\Omega_{\text{noqub}}}\lvert \Phi_{t_ku_k}^k\rangle_{A_k}
\end{equation}
be the quantum state that Alice's agent $\mathcal{A}$ receives from Bob's agent $\mathcal{B}$,
where the global quantum system received by $\mathcal{A}$ is $A=A_1\cdots A_N$, and where $\mathbf{t}=(t_1,\ldots,t_N)$ and $\mathbf{u}=(u_1,\ldots,u_N)$. We recall that $\Omega_{\text{noqub}}=[N]\setminus\Omega_{\text{qub}}$. For a given set $\Omega_{\text{qub}}\subseteq [N]$, let $\underline{\mathbf{x}}$ and $\overline{\mathbf{x}}$ denote the sub-strings of the $N-$bit string $\mathbf{x}$ with bit entries $k\in\Omega_{\text{qub}}$ and $k\in\Omega_{\text{noqub}}$, respectively, for $\mathbf{x}\in\{\mathbf{u},\mathbf{t}\}$. Thus, from (\ref{Nq1}), we can write
\begin{equation}
\label{N1006}
\lvert \Psi_{\mathbf{t}\mathbf{u}}^{\Omega_{\text{qub}}}\rangle_A=\lvert \phi_{\underline{\mathbf{t}}\underline{\mathbf{u}}}\rangle_{\underline{A}}\otimes \lvert \Phi_{\overline{\mathbf{t}}\overline{\mathbf{u}}}\rangle_{\overline{A}},
\end{equation}
where
\begin{eqnarray}
\label{neweq28}
\lvert\phi_{\underline{\mathbf{t}}\underline{\mathbf{u}}}\rangle_{\underline{A}}&=&\bigotimes_{k\in\Omega_{\text{qub}}}\lvert \phi_{t_ku_k}^k\rangle_{A_k},\nonumber\\
\lvert \Phi_{\overline{\mathbf{t}}\overline{\mathbf{u}}}\rangle_{\overline{A}}&=& \bigotimes_{k\in\Omega_{\text{noqub}}}\lvert \Phi^k_{t_ku_k}\rangle_{A_k},
\end{eqnarray}
and where $\underline{A}=\bigotimes_{k\in\Omega_{\text{qub}}} A_k$ and $\overline{A}=\bigotimes_{k\in\Omega_{\text{noqub}}} A_k$. Let $P_{\mathbf{t}\mathbf{u}}$ be the probability distribution for the variables $(\mathbf{t},\mathbf{u})\in\{0,1\}^N\times\{0,1\}^N$. By the statement of the theorem, we have
\begin{equation}
\label{N1034}
P_{\mathbf{t}\mathbf{u}}=\prod_{k=1}^NP^k_{\text{PS}}(t_k)P^k_{\text{PB}}(u_k),
\end{equation}
where $\bigl\{P^k_{\text{PB}}(0),P^k_{\text{PB}}(1)\bigr\}$ and $\bigl\{P^k_{\text{PS}}(0),P^k_{\text{PS}}(1)\bigr\}$ are binary probability distributions, for $k\in[N]$. Thus, for any sets $F_0\subseteq [N]$ and $F_1=[N]\setminus F_0$ with $\mathbf{u}_0$ and $\mathbf{u}_1$ being substrings of $\mathbf{u}$, and with $\mathbf{t}_0$ and $\mathbf{t}_1$ being substrings of $\mathbf{t}$, with bit entries with labels from the sets $F_0$ and $F_1$, respectively, we use the notation $P_{\mathbf{t}\mathbf{u}}=P_{\mathbf{t}_0\mathbf{t}_1\mathbf{u}_0\mathbf{u}_1}=P_{\mathbf{t}_0\mathbf{t}_1}P_{\mathbf{u}_0\mathbf{u}_1}$. Thus, we can express $
P_{\mathcal{S}}^{\Omega_{\text{qub}}}$ by
\begin{equation}
\label{N1005}
P_{\mathcal{S}}^{\Omega_{\text{qub}}}=\sum_{\overline{\mathbf{t}},\overline{\mathbf{u}}}P_{\overline{\mathbf{t}}\overline{\mathbf{u}}}P_{\mathcal{S}}^{\Omega_{\text{qub}}\overline{\mathbf{t}}\overline{\mathbf{u}}},
\end{equation}
where $P_{\overline{\mathbf{t}}\overline{\mathbf{u}}}$ is the probability distribution for the variables $\overline{\mathbf{t}},\overline{\mathbf{u}}$; and where $P_{\mathcal{S}}^{\Omega_{\text{qub}}\overline{\mathbf{t}}\overline{\mathbf{u}}}$ is the probability that Alice succeeds in giving Bob valid tokens at the presentation points $Q_0$ and $Q_1$ by following the strategy $\mathcal{S}$, given a set $\Omega_{\text{qub}}$ and given variables $\overline{\mathbf{t}}$ and $\overline{\mathbf{u}}$. 

We show below that there exist parameters satisfying the constraints (\ref{Al1repeated}) and satisfying
\begin{equation}
\label{N1028}
 f(\!\gamma_\text{err},\!\beta_\text{PS},\!\beta_\text{PB},\!\theta,\!\nu_\text{unf},\!\gamma_\text{det})>0,
\end{equation}
where $f(\!\gamma_\text{err},\!\beta_\text{PS},\!\beta_\text{PB},\!\theta,\!\nu_\text{unf},\!\gamma_\text{det})$ is given by (\ref{Al3repeated}). We show below that, for the parameters satisfying (\ref{Al1repeated}) and (\ref{N1028}), it holds that
\begin{equation}
\label{Nq3}
P_{\mathcal{S}}^{\Omega_{\text{qub}}\overline{\mathbf{t}}\overline{\mathbf{u}}}\leq e^{-Nf(\gamma_\text{err},\beta_\text{PS},\beta_\text{PB},\theta,\nu_\text{unf},\gamma_\text{det})},
	\end{equation}
	for any finite dimensional quantum state $\lvert \Phi_{\overline{\mathbf{t}}\overline{\mathbf{u}}}\rangle_{\overline{A}}$ and for any set $\Omega_{\text{qub}}$ satisfying $m=\lvert \Omega_{\text{noqub}}\rvert \leq \nu_\text{unf} N$. It follows from (\ref{Nq2}) and from (\ref{N1005}) -- (\ref{Nq3}) that Alice's probability to succeed in her cheating strategy is not greater than $\epsilon_{\text{unf}}$, with $\epsilon_{\text{unf}}$ given by (\ref{Al5repeated}). Thus, $\mathcal{QT}_1$ and $\mathcal{QT}_2$ are $\epsilon_{\text{unf}}-$unforgeable, with $\epsilon_{\text{unf}}$ given by (\ref{Al5repeated}), as claimed.

We show that there exist parameters satisfying the constraints (\ref{Al1repeated}) for which (\ref{N1028}) holds. Consider parameters satisfying (\ref{Al1repeated}) and the following constraint:
\begin{equation}
\label{N1035}
0<\beta_\text{PS}<\frac{1}{2}\Bigl[e^{\frac{(\gamma_\text{det}-\nu_\text{unf})\lambda(\theta, \beta_{\text{PB}})}{2}\bigl(1-\frac{\delta}{\lambda(\theta, \beta_{\text{PB}})}\bigr)^2}-1\Bigr],
\end{equation}
which is equivalent to
\begin{equation}
\label{neweq9}
0\!<\!\ln\bigl(1+2\beta_\text{PS}\bigr)\!<\!\frac{(\gamma_\text{det}\!-\!\nu_\text{unf})\lambda(\theta, \beta_{\text{PB}})}{2}\biggl(\!1-\frac{\delta}{\lambda(\theta, \beta_{\text{PB}})}\!\biggr)^2.
\end{equation}
From (\ref{Al1repeated}), we have
\begin{equation}
\label{neweq4}
0<\gamma_\text{det}-\nu_\text{unf}<1.
\end{equation}
Thus, from (\ref{N1035}) and (\ref{neweq4}), we have
 \begin{equation}
 \label{neweq5}
 0<\beta_\text{PS}<\frac{1}{2}\Bigl[e^{\frac{\lambda(\theta, \beta_{\text{PB}})}{2}\bigl(1-\frac{\delta}{\lambda(\theta, \beta_{\text{PB}})}\bigr)^2}-1\Bigr],
 \end{equation} 
  as required by (\ref{Al1repeated}). Then, since we have $0< \theta<\frac{\pi}{4}$, $0<\beta_\text{PB}<\frac{1}{2}$ and $0<\beta_\text{PS}<\frac{1}{2}$, we obtain from the definition of $h(\beta_\text{PS},\beta_{\text{PB}},\theta)$ in (\ref{Al4repeated}) that 
   \begin{equation}
 \label{neweq6}
 0<h(\beta_\text{PS},\beta_\text{PB},\theta)<2\beta_\text{PS}.
 \end{equation} 
  From 
   (\ref{neweq6}), we have
    \begin{equation}
 \label{neweq7}
 0<\ln\big[1+h(\beta_\text{PS},\beta_\text{PB},\theta)\bigr]<\ln\bigl(1+2\beta_\text{PS}\bigr).
 \end{equation} 
From (\ref{Al3repeated}), we have
\begin{eqnarray}
\label{neweq8}
&& f(\!\gamma_\text{err},\!\beta_\text{PS},\!\beta_\text{PB},\!\theta,\!\nu_\text{unf},\!\gamma_\text{det})\nonumber\\
&&\!\quad=\!\frac{(\gamma_\text{det}\!-\!\nu_\text{unf})\lambda(\theta,\beta_{\text{PB}}) }{2}\!\biggl(\!1\!-\!\frac{\delta}{\lambda(\theta,\beta_{\text{PB}})}\!\biggr)^2\!-\ln(1\!+\!2\beta_\text{PS}).\nonumber\\
&&\qquad+\bigl(\!1\!-\!(\gamma_\text{det}\!-\!\nu_\text{unf})\!\bigr)\!\Bigl[\!\ln\bigl(\!1\!+\!2\beta_\text{PB}\!\bigr)\!-\!\ln\bigl[\!1\!+\!h(\!\beta_{\text{PS}},\beta_{\text{PB}},\theta\!)\!\bigr]\!\Bigr]\nonumber\\
&&\!\quad>\!\frac{(\gamma_\text{det}\!-\!\nu_\text{unf})\lambda(\theta,\beta_{\text{PB}}) }{2}\!\biggl(\!1\!-\!\frac{\delta}{\lambda(\theta,\beta_{\text{PB}})}\!\biggr)^2\!-\ln(1\!+\!2\beta_\text{PS}).\nonumber\\
&&\!\quad>0,
\end{eqnarray}
where in the second line we used (\ref{neweq4}) and (\ref{neweq7}), and in the third line we used (\ref{neweq9}).

A general cheating strategy $\mathcal{S}$ by Alice is as follows. Alice receives the quantum system $A$ from Bob in the quantum state $\lvert \Psi_{\mathbf{t}\mathbf{u}}^{\Omega_{\text{qub}}}\rangle_A$ given by (\ref{N1006}), she adds an ancillary system $E$ of arbitrary finite Hilbert space dimension in a quantum state $\lvert\chi\rangle_E$ and applies an arbitrary projective measurement $\{\Pi_{x\mathbf{a}\mathbf{b}}\}_{x,\mathbf{a},\mathbf{b}}$ on $AE$, which may depend on $\Omega_{\text{qub}}$, and obtains an outcome $(x,\mathbf{a},\mathbf{b})$, where 
 $x=(\Lambda,g,\mathbf{d},c,\zeta)$, with $\Lambda,g,\mathbf{d}$ and $c$ comprising the information that Alice must send Bob before token presentation satisfying the required constraints, $\mathbf{a},\mathbf{b}\in\{0,1\}^{n}$ are token strings to present at the respective presentation points $Q_0$ and $Q_1$, and $\zeta$ is any other classical variable obtained by Alice's measurement.
This means  that $\Lambda\subseteq [N]$ with $\lvert \Lambda\rvert=n$, for some integer $n\geq \gamma_\text{det} N$ and for a predetermined $\gamma_\text{det}\in(0,1)$, $g$ is a one-to-one function of the form $g(k)=j$ for $k\in\Lambda$ and $j\in[n]$, $\mathbf{d}\in\{0,1\}^{n}$ and $c\in\{0,1\}$. In the intersection of the causal pasts of $Q_0$ and $Q_1$, Alice sends $\Lambda,g,c$ and $\mathbf{d}$ to Bob, as required by the task. Bob does not abort as he receives the set $\Lambda$ satisfying the required condition, as well as the $g,\mathbf{d}$ and $c$ of the required form. Alice sends the tokens to her respective agents who present them at the corresponding presentation points.

Below, we show the bound (\ref{Nq3}) on the probability $P_{\mathcal{S}}^{\Omega_{\text{qub}}\overline{\mathbf{t}}\overline{\mathbf{u}}}$ that Alice succeeds in giving Bob valid tokens at the presentation points $Q_0$ and $Q_1$ by following the strategy $\mathcal{S}$, given a set $\Omega_{\text{qub}}$ and given variables $\overline{\mathbf{t}}$ and $\overline{\mathbf{u}}$. Thus, we consider that Alice gives tokens $\mathbf{a}\in\{0,1\}^{n}$ at $Q_0$ and $\mathbf{b}\in\{0,1\}^{n}$ at $Q_1$.

\subsection{Notation and useful relations}

We recall notation. Using the set $\Lambda$ and the function $g:\Lambda\rightarrow[n]$, we have the following relations between $\mathbf{t},\mathbf{u}\in\{0,1\}^N$ and $\mathbf{r},\mathbf{s}\in\{0,1\}^{n}$. We have $r_j=t_k$, and $s_j=u_k$, where $j=g(k)$, for $j\in[n]$ and $k\in\Lambda$. We define $\mathbf{r}=(r_1,\ldots,r_n)$ and $\mathbf{s}=(s_1,\ldots,s_n)$. In the token scheme $\mathcal{QT}_1$, we have defined 
\begin{equation}
\label{neweq13}
\Delta_i=\{j\in [n]\vert \tilde{d}_{i,j}=s_j\},
\end{equation}
and $\mathbf{a}_i$ as the restriction of a string $\mathbf{a}$ to entries $a_j$ with $j\in\Delta_i$, for $i\in\{0,1\}$. These variables are defined similarly in the token scheme $\mathcal{QT}_2$ by simply setting $\mathbf{d}$ as a string whose bit entries are only zero. 

By definition of $\Lambda$, we have
\begin{equation}
\label{neweq11}
\Lambda\subseteq [N],\qquad \lvert\Lambda\rvert= n\leq N.
\end{equation}
By definition of $\Omega_\text{qub}$ and $\Omega_\text{noqub}$, we have
\begin{equation}
\label{neweq12}
\Omega_\text{qub}\cap\Omega_\text{noqub}=\emptyset,\qquad \Omega_\text{qub}\cup\Omega_\text{noqub}=[N].
\end{equation}

We note that the sets $\Delta_i$ depend on $x=(\Lambda,g,\mathbf{d},c,\zeta)$, for $i\in\{0,1\}$, but we do not write this dependence explicitly, in order to simplify the notation. Since $\tilde{d}_{1,j}=\tilde{d}_{0,j}\oplus 1$, for $j\in[n]$, we have 
\begin{equation}
\label{Nq5}
\Delta_0\cap\Delta_1=\emptyset,\qquad \Delta_0\cup\Delta_1=[n].
\end{equation}

We define 
\begin{equation}
\label{neweq14}
\underline{\Lambda}=\Lambda\cap\Omega_{\text{qub}},\qquad\overline{\Lambda}=\Lambda\cap\Omega_{\text{noqub}}.
\end{equation}
It follows that 
\begin{equation}
\label{neweq15}
\underline{\Lambda}\cap \overline{\Lambda}=\emptyset,\qquad \underline{\Lambda}\cup \overline{\Lambda}=\Lambda.
\end{equation}
 Similarly, we define 
 \begin{eqnarray}
 \label{neweq16}
 \underline{\Delta}&=&\{j\in[n]\vert ~\exists k\in\underline{\Lambda} \text{~s.~t.~} g(k)=j\},\nonumber\\
 \overline{\Delta}&=&\{j\in[n]\vert ~\exists k\in\overline{\Lambda} \text{~s.~t.~} g(k)=j\}.
  \end{eqnarray}
 Since the function $g: \Lambda\rightarrow [n]$ is one-to-one, there is a one-to-one correspondence between the elements of $\underline{\Lambda}$ ($\overline{\Lambda}$) and the elements of $\underline{\Delta}$ ($\overline{\Delta}$). Thus, from (\ref{neweq15}) and (\ref{neweq16}), we have
 \begin{equation}
\label{neweq17}
\underline{\Delta}\cap \overline{\Delta}=\emptyset,\qquad \underline{\Delta}\cup \overline{\Delta}=[n].
\end{equation}
 We define 
 \begin{equation}
\label{neweq18}
\underline{\Delta}_i=\Delta_i\cap \underline{\Delta},
\end{equation}
 for $i\in\{0,1\}$. It follows that $\underline{\Delta}_i\subseteq \Delta_i$, for $i\in\{0,1\}$. 
From the definitions of $\underline{\Delta}_0$, $\underline{\Delta}_1$ and $\underline{\Delta}$, we have
\begin{equation}
\label{neweq10}
\underline{\Delta}_0\cap\underline{\Delta}_1=\emptyset,\qquad
\underline{\Delta}_0\cup\underline{\Delta}_1=\underline{\Delta}.
\end{equation}
 We define
 \begin{eqnarray}
 \label{neweq19}
 \Lambda_i&=&\{k\in\Lambda\vert g(k)\in \Delta_i \},\nonumber\\\underline{\Lambda}_i&=&\{k\in\Lambda\vert g(k)\in\underline{\Delta}_i\},
 \end{eqnarray}
for $i\in\{0,1\}$. Since the function $g:\Lambda\rightarrow[n]$ is one-to-one, we have from (\ref{Nq5}), (\ref{neweq16}), (\ref{neweq10}) and (\ref{neweq19}) that
\begin{eqnarray}
\label{neweq20}
\Lambda_0\cap\Lambda_1&=&\emptyset,\qquad
\Lambda_0\cup\Lambda_1=\Lambda,\\
\label{neweq21}
\underline{\Lambda}_0\cap\underline{\Lambda}_1&=&\emptyset,\qquad
\underline{\Lambda}_0\cup\underline{\Lambda}_1=\underline{\Lambda}.
\end{eqnarray}

We define $\mathbf{\underline{\underline{e}}}$, $\mathbf{\overline{\overline{e}}}$, $\underline{\mathbf{e}}_0$ and $\underline{\mathbf{e}}_1$ to be the restrictions of the string $\mathbf{e}\in\{0,1\}^{n}$ to the bit entries $e_j$ with labels $j\in\underline{\Delta}$, $j\in\overline{\Delta}$, $j\in\underline{\Delta}_0$ and $j\in\underline{\Delta}_1$, respectively, for $\mathbf{e}\in\{\mathbf{a},\mathbf{b},\mathbf{r},\mathbf{s}\}$. We have chosen the notation $\mathbf{\underline{\underline{e}}}$ and $\mathbf{\overline{\overline{e}}}$ instead of the more obvious $\mathbf{\underline{e}}$ and $\mathbf{\overline{e}}$ for consistency with the notation chosen below, which simplifies the notation in various equations that follow. From (\ref{neweq18}), we have $\underline{\Delta}_i\subseteq \Delta_i$, for $i\in\{0,1\}$. Thus,  and since $\mathbf{e}_0$ and $\mathbf{e}_1$ are restrictions of the string $\mathbf{e}$ to the bit entries $e_j$ with labels $j\in\Delta_0$ and $j\in\Delta_1$, respectively,
we have that $\underline{\mathbf{e}}_0$ and $\underline{\mathbf{e}}_1$ are substrings of the strings $\mathbf{e}_0$ and $\mathbf{e}_1$, respectively, for $\mathbf{e}\in\{\mathbf{a},\mathbf{b},\mathbf{r},\mathbf{s}\}$. From (\ref{neweq17}) and (\ref{neweq10}), we have
\begin{equation}
\label{Nq7}
\mathbf{e}=(\underline{\underline{\mathbf{e}}},\overline{\overline{\mathbf{e}}})=(\underline{\mathbf{e}}_0,\underline{\mathbf{e}}_1,\overline{\overline{\mathbf{e}}}),\qquad
\underline{\underline{\mathbf{e}}}=(\underline{\mathbf{e}}_0,\underline{\mathbf{e}}_1),
\end{equation}
where $\underline{\underline{\mathbf{e}}}\in\{0,1\}^{\underline{\Delta}}$,  $\overline{\overline{\mathbf{e}}}\in\{0,1\}^{\overline{\Delta}}$, $\underline{\mathbf{e}}_0\in\{0,1\}^{\underline{\Delta}_0}$ and $\underline{\mathbf{e}}_1\in\{0,1\}^{\underline{\Delta}_1}$, for $\mathbf{e}\in\{\mathbf{a},\mathbf{b},\mathbf{r},\mathbf{s}\}$.

 We define $\underline{\mathbf{e}}_i$ to be the restrictions of the string $\mathbf{e}\in\{0,1\}^{N}$ to the bit entries $e_k$ with labels $k\in\underline{\Lambda}_i$, for $i\in\{0,1\}$ and $\mathbf{e}\in\{\mathbf{u},\mathbf{t}\}$. Since the function $g:\Lambda\rightarrow [n]$ is one to one, there is a one to one correspondence between $\underline{\mathbf{t}}_i$ and $\underline{\mathbf{r}}_i$, and between $\underline{\mathbf{u}}_i$ and $\underline{\mathbf{s}}_i$, for $i\in\{0,1\}$. We define the string $\underline{\underline{\mathbf{e}}}$ to be the restriction of the string $\mathbf{e}\in\{0,1\}^N$ to the bit entries $e_k$ with labels $k\in\underline{\Lambda}$, for $\mathbf{e}\in\{\mathbf{u},\mathbf{t}\}$. We define $\underline{\mathbf{e}}$ to be the restriction of $\mathbf{e}\in\{0,1\}^N$ to the bit entries $e_k$ with labels $k\in\Omega_{\text{qub}}$, for $\mathbf{e}\in\{\mathbf{u},\mathbf{t}\}$. Thus, from (\ref{neweq14}), $\underline{\underline{\mathbf{e}}}$ is a sub-string of $\underline{\mathbf{e}}$, for $\mathbf{e}\in\{\mathbf{u},\mathbf{t}\}$. From (\ref{neweq21}), we can write
\begin{equation}
\label{Nq7.1}
\underline{\underline{\mathbf{e}}}=(\underline{\mathbf{e}}_0,\underline{\mathbf{e}}_1),
\end{equation}
where $\underline{\underline{\mathbf{e}}}\in\{0,1\}^{\underline{\Lambda}}$, $\underline{\mathbf{e}}_0\in\{0,1\}^{\underline{\Lambda}_0}$ and $\underline{\mathbf{e}}_1\in\{0,1\}^{\underline{\Lambda}_1}$, for $\mathbf{e}\in\{\mathbf{t},\mathbf{u}\}$. We define $\underline{\mathbf{e}}'$ as the restriction of the string $\mathbf{e}\in\{0,1\}^N$ to the bit entries $e_k$ with labels $k\in\Omega_{\text{qub}}\setminus\underline{\Lambda}$, for $\mathbf{e}\in\{\mathbf{u},\mathbf{t}\}$. It follows that we can write 
\begin{equation}
\label{N1037}
\underline{\mathbf{e}}=(\underline{\underline{\mathbf{e}}},\underline{\mathbf{e}}')=(\underline{\mathbf{e}}_0,\underline{\mathbf{e}}_1,\underline{\mathbf{e}}'),
\end{equation}
where $\underline{\mathbf{e}}\in\{0,1\}^{\Omega_{\text{qub}}}$, $\underline{\underline{\mathbf{e}}}\in\{0,1\}^{\underline{\Lambda}}$, $\underline{\mathbf{e}}'\in\{0,1\}^{\Omega_{\text{qub}}\setminus\underline{\Lambda}}$, $\underline{\mathbf{e}}_0\in\{0,1\}^{\underline{\Lambda}_0}$ and $\underline{\mathbf{e}}_1\in\{0,1\}^{\underline{\Lambda}_1}$, for $\mathbf{e}\in\{\mathbf{t},\mathbf{u}\}$.

As mentioned above, given our notation, there is a one-to-one correspondence between $\underline{\mathbf{t}}_i$ and $\underline{\mathbf{r}}_i$, and between $\underline{\mathbf{u}}_i$ and $\underline{\mathbf{s}}_i$, for $i\in\{0,1\}$. Similarly, there is a one-to-one correspondence between $\underline{\underline{\mathbf{t}}}$ and $\underline{\underline{\mathbf{r}}}$, and between $\underline{\underline{\mathbf{u}}}$ and $\underline{\underline{\mathbf{s}}}$. We express these, and previously mentioned, one-to-one correspondences as follows
\begin{eqnarray}
\label{neweq23}
\Lambda&\leftrightarrow &[n],\nonumber\\
\underline{\Lambda}&\leftrightarrow &\underline{\Delta},\nonumber\\
\overline{\Lambda}&\leftrightarrow &\overline{\Delta},\nonumber\\
\Lambda_i&\leftrightarrow &\Delta_i,\nonumber\\
\underline{\Lambda}_i&\leftrightarrow &\underline{\Delta}_i,\nonumber\\
\underline{\underline{\mathbf{u}}}&\leftrightarrow &\underline{\underline{\mathbf{s}}},\nonumber\\
\underline{\underline{\mathbf{t}}}&\leftrightarrow &\underline{\underline{\mathbf{r}}},\nonumber\\
\underline{\mathbf{u}}_i&\leftrightarrow &\underline{\mathbf{s}}_i,\nonumber\\
\underline{\mathbf{t}}_i&\leftrightarrow &\underline{\mathbf{r}}_i,
\end{eqnarray}
for $i\in\{0,1\}$.

In the rest of this proof we assume 
\begin{equation}
\label{neweq22}
m=\lvert \Omega_{\text{noqub}}\rvert\leq \nu_\text{unf} N.
\end{equation}
We consider the $n-$bit strings  $\tilde{\mathbf{d}}_i=(\tilde{d}_{i,1},\ldots,\tilde{d}_{i,n})$ for $i\in\{0,1\}$. By definition of the token scheme, we have $\tilde{d}_{i,j}=d_{j}\oplus i\oplus c$, for $j\in[n]$ and $i\in\{0,1\}$. Thus, we have $\tilde{d}_{1,j}=\tilde{d}_{0,j}\oplus 1$, for $j\in[n]$, from which follows that the $n-$bit strings $\tilde{\mathbf{d}}_0$ and $\tilde{\mathbf{d}}_1$ are the complement of each other. We have
\begin{eqnarray}
\label{Nq6}
\lvert \underline{\Delta}_0\cup \underline{\Delta}_1\rvert &=&\lvert \underline{\Delta}\rvert\nonumber\\
&=&\lvert\underline{\Lambda}\rvert\nonumber\\
&\geq& \lvert\Lambda\rvert-\lvert\Omega_{\text{noqub}}\rvert\nonumber\\
&\geq&n - \nu_\text{unf} N,
\end{eqnarray}
where in the first line we used (\ref{neweq10}); in the second line we used (\ref{neweq16}) and the fact that $g:\Lambda\rightarrow[n]$ is a one-to-one function; in the third line we used (\ref{neweq12}) and (\ref{neweq14}); and in the last line we used (\ref{neweq11}) and (\ref{neweq22}).

For a given possible outcome $x$, we define the sets
\begin{eqnarray}
\label{N1008}
\Gamma^x_{\mathbf{t}\mathbf{u}}&=&\{(\mathbf{a},\mathbf{b})\in \{0,1\}^{n}\times \{0,1\}^{n}\vert d(\mathbf{a}_0,\mathbf{r}_0)\nonumber\\
&&\qquad\leq\lvert \Delta_0\rvert \gamma_\text{err},d(\mathbf{b}_1,\mathbf{r}_1)\leq \lvert \Delta_1\rvert \gamma_\text{err}\},\\
\label{N1008.1}
\underline{\Gamma}^x_{\mathbf{t}\mathbf{u}}&=&\{(\mathbf{a},\mathbf{b})\in \{0,1\}^{n}\times \{0,1\}^{n}\vert d(\underline{\mathbf{a}}_0,\underline{\mathbf{r}}_0)\nonumber\\
&&\qquad\leq\lvert \Delta_0\rvert \gamma_\text{err},d(\underline{\mathbf{b}}_1,\underline{\mathbf{r}}_1)\leq \lvert \Delta_1\rvert \gamma_\text{err}\},\\
\label{N1008.2}
\underline{\tilde{\Gamma}}^x_{\mathbf{t}\mathbf{u}}&=&\{(\mathbf{a},\mathbf{b})\in \{0,1\}^{n}\times \{0,1\}^{n}\vert d(\underline{\mathbf{a}}_0,\underline{\mathbf{r}}_0)\nonumber\\
&&\qquad+d(\underline{\mathbf{b}}_1,\underline{\mathbf{r}}_1)\leq \bigl(\lvert \Delta_0\rvert +\lvert \Delta_1\rvert\bigr)\gamma_\text{err}\},\\
\label{N1008.4}
\underline{\overline{\Gamma}}^x_{\mathbf{t}\mathbf{u}}&=&\{(\mathbf{a},\mathbf{b})\in \{0,1\}^{n}\times \{0,1\}^{n}\vert d(\underline{\mathbf{a}}_0,\underline{\mathbf{r}}_0)\nonumber\\
&&\qquad+d(\underline{\mathbf{b}}_1,\underline{\mathbf{r}}_1)\leq \bigl(\lvert \underline{\Delta}_0\rvert +\lvert \underline{\Delta}_1\rvert\bigr)\delta\},\\
\label{N1008.3}
\xi^x_{\mathbf{a}\mathbf{b}\underline{\mathbf{u}}}&=&\{\underline{\mathbf{t}}\in \{0,1\}^{\Omega_{\text{qub}}}\vert d(\underline{\mathbf{a}}_0,\underline{\mathbf{r}}_0)+d(\underline{\mathbf{b}}_1,\underline{\mathbf{r}}_1)\nonumber\\
&&\quad\qquad\qquad\qquad \leq  \bigl(\lvert \underline{\Delta}_0\rvert +\lvert \underline{\Delta}_1\rvert\bigr)\delta\},
\end{eqnarray}
where
\begin{equation}
\label{N1012}
\tilde{\delta}=\frac{\bigl(\lvert \Delta_0\rvert+\lvert \Delta_1\rvert\bigr)\gamma_\text{err}}{\lvert\underline{\Delta}_0\rvert+\lvert \underline{\Delta}_1\rvert},
\end{equation}
and where $\delta$ is given by (\ref{Al4repeated}). As this is useful below, we have clarified with the chosen notation that $\xi^x_{\mathbf{a}\mathbf{b}\underline{\mathbf{u}}}$ does not depend on $\mathbf{t}\overline{\mathbf{u}}$.
We show below that 
\begin{equation}
\label{Nj0}
0< \tilde{\delta}\leq \delta <\lambda(\theta,\beta_{\text{PB}}).
\end{equation}
It follows straightforwardly that
\begin{equation}
\label{N1011}
 \Gamma^x_{\mathbf{t}\mathbf{u}}\subseteq\underline{\Gamma}^x_{\mathbf{t}\mathbf{u}}\subseteq\underline{\tilde{\Gamma}}^x_{\mathbf{t}\mathbf{u}}\subseteq \underline{\overline{\Gamma}}^x_{\mathbf{t}\mathbf{u}}.
\end{equation}

We show (\ref{Nj0}). From the condition $n\geq \gamma_\text{det} N$ for Bob not aborting and from (\ref{Al1repeated}), we have 
\begin{equation}
\label{neweq26}
n-\nu_\text{unf} N\geq (\gamma_\text{det}-\nu_\text{unf})N>0.
\end{equation}
Thus, from $\gamma_\text{err}>0$, (\ref{Nq5}), (\ref{neweq10}), (\ref{Nq6}), (\ref{N1012}) and (\ref{neweq26}), we have $\tilde{\delta}> 0$. From (\ref{Nq5}), (\ref{neweq10}), (\ref{Nq6}) and (\ref{N1012}), we have
\begin{eqnarray}
\label{Nq11}
\tilde{\delta}&\leq&\frac{n\gamma_\text{err}}{n-\nu_\text{unf} N}\nonumber\\
&=&\gamma_\text{err}\Bigl(1+\frac{\nu_\text{unf} N}{n-\nu_\text{unf} N}\Bigr)\nonumber\\
&\leq&\gamma_\text{err}\Bigl(1+\frac{\nu_\text{unf}}{\gamma_\text{det}-\nu_\text{unf}}\Bigr)\nonumber\\
&<&\lambda(\theta, \beta_{\text{PB}}),
\end{eqnarray}
where in the third line we used the condition $n\geq\gamma_\text{det} N$ for Bob not aborting, and in the last line we used (\ref{Al1repeated}). Since $\delta= \frac{\gamma_\text{err}\gamma_\text{det}}{\gamma_\text{det}-\nu_\text{unf}}$, as defined by (\ref{Al4repeated}), (\ref{Nj0}) follows from (\ref{Nq11}).

The probability that Alice obtains outcomes $x, \mathbf{a}$ and $\mathbf{b}$ following her strategy $\mathcal{S}$ for given values of $\Omega_{\text{qub}}$, $\mathbf{u}$ and $\mathbf{t}$ is given by
\begin{equation}
\label{N1009}
P_{\mathcal{S}}^{\Omega_{\text{qub}}\overline{\mathbf{t}}\overline{\mathbf{u}}}[x\mathbf{a}\mathbf{b}\vert \underline{\mathbf{t}}\underline{\mathbf{u}}]=\text{Tr}\Bigl[\Bigl(\bigl(\phi_{\underline{\mathbf{t}}\underline{\mathbf{u}}}\bigr)_{\underline{A}}\otimes \bigl(\Phi_{\overline{\mathbf{t}}\overline{\mathbf{u}}}\bigr)_{\overline{A}E}\Bigr)\Pi_{x\mathbf{a}\mathbf{b}}\Bigr],
\end{equation}
where
\begin{eqnarray}
\label{neweq27}
\bigl(\phi_{\underline{\mathbf{t}}\underline{\mathbf{u}}}\bigr)_{\underline{A}}&=&\bigl(\lvert\phi_{\underline{\mathbf{t}}\underline{\mathbf{u}}}\rangle \langle\phi_{\underline{\mathbf{t}}\underline{\mathbf{u}}}\rvert\bigr)_{\underline{A}},\nonumber\\
\bigl(\Phi_{\overline{\mathbf{t}}\overline{\mathbf{u}}}\bigr)_{\overline{A}E}&=&\bigl(\lvert\Phi_{\overline{\mathbf{t}}\overline{\mathbf{u}}}\rangle \langle\Phi_{\overline{\mathbf{t}}\overline{\mathbf{u}}}\rvert\bigr)_{\overline{A}}\otimes \bigl(\lvert \chi\rangle\langle\chi\rvert\bigr)_E,
\end{eqnarray}
and where $\lvert\phi_{\underline{\mathbf{t}}\underline{\mathbf{u}}}\rangle_{\underline{A}}$ and $\lvert\Phi_{\overline{\mathbf{t}}\overline{\mathbf{u}}}\rangle_{\overline{A}}$ are defined by (\ref{neweq28}).
Thus, Alice's success probability $P_{\mathcal{S}}^{\Omega_{\text{qub}}\overline{\mathbf{t}}\overline{\mathbf{u}}}$ satisfies
\begin{eqnarray}
\label{N1007}
P_{\mathcal{S}}^{\Omega_{\text{qub}}\overline{\mathbf{t}}\overline{\mathbf{u}}}&=&\sum_{\underline{\mathbf{t}},\underline{\mathbf{u}},x}\sum_{ (\mathbf{a},\mathbf{b})\in \Gamma^x_{\mathbf{t}\mathbf{u}}} P_{\underline{\mathbf{t}}\underline{\mathbf{u}}}P_{\mathcal{S}}^{\Omega_{\text{qub}}\overline{\mathbf{t}}\overline{\mathbf{u}}}[x\mathbf{a}\mathbf{b}\vert \underline{\mathbf{t}}\underline{\mathbf{u}}]\nonumber\\
&\leq&\sum_{\underline{\mathbf{t}},\underline{\mathbf{u}},x}~\sum_{ (\mathbf{a},\mathbf{b})\in \underline{\overline{\Gamma}}^x_{\mathbf{t}\mathbf{u}}} P_{\underline{\mathbf{t}}\underline{\mathbf{u}}}P_{\mathcal{S}}^{\Omega_{\text{qub}}\overline{\mathbf{t}}\overline{\mathbf{u}}}[x\mathbf{a}\mathbf{b}\vert \underline{\mathbf{t}}\underline{\mathbf{u}}]\nonumber\\
&=&\sum_{\underline{\mathbf{u}},x}~\sum_{\substack{\underline{\mathbf{t}}_0,\underline{\mathbf{t}}_1,\underline{\mathbf{t}}'\\ \underline{\mathbf{b}}_0,\underline{\mathbf{a}}_1,\overline{\overline{\mathbf{a}}},\overline{\overline{\mathbf{b}}}}}~\sum_{ (\underline{\mathbf{a}}_0,\underline{\mathbf{b}}_1): C} P_{\underline{\mathbf{t}}\underline{\mathbf{u}}}P_{\mathcal{S}}^{\Omega_{\text{qub}}\overline{\mathbf{t}}\overline{\mathbf{u}}}[x\mathbf{a}\mathbf{b}\vert \underline{\mathbf{t}}\underline{\mathbf{u}}]\nonumber\\
&=&\sum_{\underline{\mathbf{u}},x}~\sum_{\substack{\underline{\mathbf{a}}_0,\underline{\mathbf{b}}_1,\underline{\mathbf{t}}'\\ \underline{\mathbf{b}}_0,\underline{\mathbf{a}}_1,\overline{\overline{\mathbf{a}}},\overline{\overline{\mathbf{b}}}}}~\sum_{ (\underline{\mathbf{t}}_0,\underline{\mathbf{t}}_1): C} P_{\underline{\mathbf{t}}\underline{\mathbf{u}}}P_{\mathcal{S}}^{\Omega_{\text{qub}}\overline{\mathbf{t}}\overline{\mathbf{u}}}[x\mathbf{a}\mathbf{b}\vert \underline{\mathbf{t}}\underline{\mathbf{u}}]\nonumber\\
&=&\sum_{\mathbf{a},\mathbf{b},x,\underline{\mathbf{u}}}~\sum_{ \underline{\mathbf{t}}\in \xi^x_{\mathbf{a}\mathbf{b}\underline{\mathbf{u}}}} P_{\underline{\mathbf{t}}\underline{\mathbf{u}}}P_{\mathcal{S}}^{\Omega_{\text{qub}}\overline{\mathbf{t}}\overline{\mathbf{u}}}[x\mathbf{a}\mathbf{b}\vert \underline{\mathbf{t}}\underline{\mathbf{u}}],
\end{eqnarray}
where $C$ denotes the constraint 
\begin{equation}
\label{neweq24}
d (\underline{\mathbf{a}}_0,\underline{\mathbf{r}}_0)+d(\underline{\mathbf{b}}_1,\underline{\mathbf{r}}_1)\leq (\lvert\underline{\Delta}_0\rvert+\lvert\underline{\Delta}_1\rvert)\delta;
\end{equation}
where in the first line we used (\ref{N1008}) and (\ref{N1009}); in the second line we used (\ref{N1011}); in the third line we used (\ref{Nq7}), (\ref{N1037}), and (\ref{N1008.4}); in the fourth line we used that $\underline{\mathbf{t}}_i$ is in one to one correspondence with  $\underline{\mathbf{r}}_i$, for $i\in\{0,1\}$; and in the last line we used (\ref{Nq7}), (\ref{N1037}) and (\ref{N1008.3}).

\subsection{Entanglement-based version}

We use an entanglement-based version of the task to re-write the last line of (\ref{N1007}). For $k\in\Omega_{\text{qub}}$, Bob first prepares a pair of qubits $B_kA_k$ in the Bell state $\lvert\Phi^+\rangle_{B_kA_k}=\frac{1}{\sqrt{2}}\bigl(\lvert 0\rangle\lvert0\rangle+\lvert1\rangle\lvert1\rangle\bigr)_{B_kA_k}$, sends the qubit $A_k$ to Alice, chooses $u_k\in\{0,1\}$ with probability $P_{\text{PB}}^k(u_k)$ and then measures the qubit $B_k$ in the the basis $\mathcal{D}_{u_k}^k=\{\lvert \phi_{tu_k}^k\rangle\}_{t=0}^1$, obtaining the outcome $\lvert \phi_{t_ku_k}^k\rangle$ randomly, with Alice's qubit $A_k$ projecting into the same state, for $t_k\in\{0,1\}$. We note that in order to deal with the fact that the probability $P_{\underline{\mathbf{t}}}$ does not necessarily correspond to the random distribution, for $\underline{\mathbf{t}}\in\{0,1\}^{\Omega_{\text{qub}}}$, we will need to introduce a factor of $P_{\underline{\mathbf{t}}}2^{\lvert\Omega_{\text{qub}}\rvert}$. For $k\in\Omega_{\text{noqub}}$, Bob generates the bits $t_k$ and $u_k$ in such a way that the strings $\overline{\mathbf{t}}$ and $\overline{\mathbf{u}}$ are generated with the probability distribution $P_{\overline{\mathbf{t}}\overline{\mathbf{u}}}$. Given the obtained values for $\overline{\mathbf{t}}$ and $\overline{\mathbf{u}}$, Bob prepares a finite dimensional quantum state $\lvert \Phi_{\overline{\mathbf{t}}\overline{\mathbf{u}}}\rangle_{\overline{A}}$ and sends it to Alice. Alice introduces an ancillary quantum system $E$ of arbitrary finite Hilbert space dimension in a pure state $\lvert \chi\rangle_E$ and then applies a projective measurement on $AE$, with projector operators $\Pi_{x\mathbf{a}\mathbf{b}}$, where the possible measurement outcomes $x=(\Lambda,g,\mathbf{d},c,\zeta)$ and $\mathbf{a},\mathbf{b}\in\{0,1\}^{n}$, run over the set of values satisfying the constraints, and where $A=A_1\cdots A_N=\underline{A}\overline{A}$. We define the quantum state
\begin{equation}
\label{Nl1.1}
\rho_{\overline{\mathbf{t}}\overline{\mathbf{u}}}=\bigl(\Phi^+\bigr)_{\underline{B}\underline{A}}\otimes\bigl(\Phi_{\overline{\mathbf{t}}\overline{\mathbf{u}}}\bigr)_{\overline{A}E},
\end{equation}
where $\underline{B}$ denotes the system held by Bob, $\bigl(\Phi^+\bigr)_{\underline{B}\underline{A}}=\bigotimes_{k\in\Omega_{\text{qub}}}\bigl(\lvert\Phi^+\rangle\langle\Phi^+\rvert\bigr)_{B_kA_k}$, and where the state $\bigl(\Phi_{\overline{\mathbf{t}}\overline{\mathbf{u}}}\bigr)_{\overline{A}E}
$is defined by (\ref{neweq27}). We define the positive semi definite (and Hermitian) operators
\begin{equation}
\label{Nl1.7x}
D_{x\mathbf{a}\mathbf{b}}
=\sum_{\underline{\mathbf{u}}}P_{\underline{\mathbf{u}}}\sum_{\underline{\mathbf{t}}\in \xi^x_{\mathbf{a}\mathbf{b}\underline{\mathbf{u}}}} P_{\underline{\mathbf{t}}}2^{\lvert\Omega_{\text{qub}}\rvert}\bigl(\phi_{\underline{\mathbf{t}}\underline{\mathbf{u}}}\bigr)_{\underline{B}},
\end{equation}
and 
\begin{equation}
\label{Nl1.7}
\tilde{P}=\sum_{x,\mathbf{a},\mathbf{b}} \bigl(D_{x\mathbf{a}\mathbf{b}}
\bigr)_{\underline{B}}\otimes \bigl(\Pi_{x\mathbf{a}\mathbf{b}}\bigr)_{AE},
\end{equation}
where $\bigl(\phi_{\underline{\mathbf{t}}\underline{\mathbf{u}}}\bigr)_{\underline{B}}$ is given by  (\ref{neweq27}), replacing $\underline{A}$ by $\underline{B}$, i.e. $\bigl(\phi_{\underline{\mathbf{t}}\underline{\mathbf{u}}}\bigr)_{\underline{B}}=\bigotimes _{k\in\Omega_{\text{qub}}} \bigl(\lvert \phi_{t_ku_k}^k\rangle\langle\phi_{t_ku_k}^k\rvert\bigr)_{B_k}$; and where $\underline{\mathbf{u}}$ runs over $\{0,1\}^{\Omega_{\text{qub}}}$, $x$ runs over its set of possible values, and $\mathbf{a}$ and $\mathbf{b}$ run over $\{0,1\}^{n}$.
It follows straightforwardly from (\ref{N1009}) -- (\ref{Nl1.7}), and from $P_{\underline{\mathbf{t}}\underline{\mathbf{u}}}=P_{\underline{\mathbf{t}}}P_{\underline{\mathbf{u}}}$, which follows from (\ref{N1034}), that
\begin{eqnarray}
\label{N1010}
P_{\mathcal{S}}^{\Omega_{\text{qub}}\overline{\mathbf{t}}\overline{\mathbf{u}}}&\leq& \text{Tr} \bigl( \tilde{P}\rho_{\overline{\mathbf{t}}\overline{\mathbf{u}}}\bigr)\nonumber\\
&\leq&\lVert \tilde{P}
\rVert\nonumber\\
&=&\max_{x,\mathbf{a},\mathbf{b}} \bigl\lVert D_{x\mathbf{a}\mathbf{b}}
\bigr\rVert,
\end{eqnarray}
where in the second line we used Proposition \ref{proposition1}; and where in the third line we used (\ref{Nl1.7}) and Proposition \ref{proposition2}, since $\{\Pi_{x\mathbf{a}\mathbf{b}}\}_{x,\mathbf{a},\mathbf{b}}$ is a projective measurement acting on a finite dimensional Hilbert space, and $\{D
_{x\mathbf{a}\mathbf{b}}\}_{x,\mathbf{a},\mathbf{b}}$ is a finite set of positive semi definite operators acting on a finite dimensional Hilbert space. We show below that
\begin{equation}
\label{N1030}
 \max_{x,\mathbf{a},\mathbf{b}} \lVert D_{x\mathbf{a}\mathbf{b}}
 \rVert\leq e^{-Nf(\gamma_\text{err},\beta_\text{PS},\beta_\text{PB},\theta,\nu_\text{unf},\gamma_\text{det})},
\end{equation}
with $f(\gamma_\text{err},\beta_\text{PS},\beta_\text{PB},\theta,\nu_\text{unf},\gamma_\text{det})$ given by (\ref{Al3repeated}). Thus, the claimed bound (\ref{Nq3}) follows from (\ref{N1010}) and (\ref{N1030}).

\subsection{Using Lemma \ref{lemma0} to prove (\ref{N1030})}

We compute an upper bound on $\max_{x,\mathbf{a},\mathbf{b}} \bigl\lVert D_{x\mathbf{a}\mathbf{b}}\bigr\rVert$. First, we define the set
\begin{eqnarray}
\label{N1013}
\underline{\xi}^x_{\mathbf{a}\mathbf{b}\underline{\mathbf{u}}_0\underline{\mathbf{u}}_1}&=&\bigl\{(\underline{\mathbf{t}}_0,\underline{\mathbf{t}}_1)\in \{0,1\}^{\underline{\Lambda}_0}\times \{0,1\}^{\underline{\Lambda}_1}\vert d(\underline{\mathbf{a}}_0,\underline{\mathbf{r}}_0)\nonumber\\
&&\qquad +d(\underline{\mathbf{b}}_1,\underline{\mathbf{r}}_1)\leq  \bigl(\lvert \underline{\Delta}_0\rvert +\lvert \underline{\Delta}_1\rvert\bigr)\delta\bigr\}.
\end{eqnarray}
As explicitly stated by the notation, we note that the dependence of $\underline{\xi}^x_{\mathbf{a}\mathbf{b}\underline{\mathbf{u}}_0\underline{\mathbf{u}}_1}$ on $\mathbf{u}$ is only via the sub-strings $\underline{\mathbf{u}}_0$ and $\underline{\mathbf{u}}_1$. This follow in particular because the constant $\delta$ does not depend on $\mathbf{u}$, as follows from (\ref{Al4repeated}). Thus, from (\ref{N1037}), (\ref{N1008.3}), (\ref{Nl1.7x}) and (\ref{N1013}), we have
\begin{equation}
\label{N1014}
D_{x\mathbf{a}\mathbf{b}}=\tilde{D}_{x\mathbf{a}\mathbf{b}}\otimes \tilde{\phi}_x,
\end{equation}
where
\begin{eqnarray}
\label{N1015}
\tilde{D}_{x\mathbf{a}\mathbf{b}}&=&\!\!\sum_{\underline{\mathbf{u}}_0,\underline{\mathbf{u}}_1}\!\!\!
P_{\underline{\mathbf{u}}_0\underline{\mathbf{u}}_1}\!\!\!\!\sum_{(\underline{\mathbf{t}}_0,\underline{\mathbf{t}}_1)\in\underline{\xi}^x_{\mathbf{a}\mathbf{b}\underline{\mathbf{u}}_0\underline{\mathbf{u}}_1}} \!\!\!\!\!\!\!\!\!\!\!\!
\!P_{\underline{\mathbf{t}}_0\underline{\mathbf{t}}_1}\!\!\bigl(\!\phi_{\underline{\mathbf{t}}_0\underline{\mathbf{t}}_1\underline{\mathbf{u}}_0\underline{\mathbf{u}}_1}\!\bigr)_{\underline{B}_0\underline{B}_1}\!,\\
\label{N1016}
\tilde{\phi}_x&=&2^{\lvert \Omega_{\text{qub}}\rvert}\sum_{\underline{\mathbf{t}}',\underline{\mathbf{u}}'}P_{\underline{\mathbf{u}}'\underline{\mathbf{t}}'}\bigl(\phi_{\underline{\mathbf{t}}'\underline{\mathbf{u}}'}\bigr)_{\underline{B}'},
\end{eqnarray}
and where $\underline{B}=\underline{B}_0\underline{B}_1\underline{B}'$. It follows from (\ref{N1014}) that
\begin{equation}
\label{N1019}
\lVert D_{x\mathbf{a}\mathbf{b}}
\rVert =
\lVert \tilde{D}_{x\mathbf{a}\mathbf{b}}\rVert~ \lVert \tilde{\phi}_x\rVert.
\end{equation}

We deduce an upper bound on $\max_{\mathbf{a},\mathbf{b}}\lVert \tilde{D}_{x\mathbf{a}\mathbf{b}}\rVert$. For given values of $x,\mathbf{a}$ and $\mathbf{b}$, we define the operator
\begin{equation}
\label{N1023}
\tilde{\underline{D}}_{x\mathbf{a}\mathbf{b}}=\sum_{\underline{\mathbf{u}}_0,\underline{\mathbf{u}}_1}P_{\underline{\mathbf{u}}_0\underline{\mathbf{u}}_1}\!\!\!\sum_{(\underline{\mathbf{t}}_0,\underline{\mathbf{t}}_1)\in\underline{\xi}^x_{\mathbf{a}\mathbf{b}\underline{\mathbf{u}}_0\underline{\mathbf{u}}_1}} \!\!\!\bigl(\phi_{\underline{\mathbf{t}}_0\underline{\mathbf{t}}_1\underline{\mathbf{u}}_0\underline{\mathbf{u}}_1}\bigr)_{\underline{B}_0\underline{B}_1}.
\end{equation}
For a given $x$, it follows from (\ref{Nj0}), (\ref{N1013}) and (\ref{N1023}), and from Lemma \ref{lemma0} that
\begin{equation}
\label{N1022}
\max_{\mathbf{a},\mathbf{b}} \lVert \tilde{\underline{D}}_{x\mathbf{a}\mathbf{b}}\rVert \leq e^{-\bigl(\!\lvert \underline{\Delta}_0\rvert +\lvert \underline{\Delta}_1\rvert \!\bigr)\frac{\lambda(\theta,\beta_{\text{PB}})}{2}\!\Bigl(\!1-\frac{\delta}{\lambda(\theta,\beta_{\text{PB}}\!)}\Bigr)^2},
\end{equation}
where
\begin{equation}
\label{N1049}
\lambda(\theta,\beta_{\text{PB}})=\frac{1}{2}\Bigl(1\!-\!\sqrt{1\!-\![1\!-\!(O(\theta))^2](1\!-\!4\beta_\text{PB}^2)}\Bigr).
\end{equation}
To see this more clearly, recall the following facts. First, $\underline{\mathbf{e}}_i$ is a bit string with bit entries $e_j$ with labels $j\in\underline{\Delta}_i$, for $i\in\{0,1\}$ and $\mathbf{e}\in\{\mathbf{a},\mathbf{b},\mathbf{r},\mathbf{s}\}$. Second,  $\underline{\mathbf{e}}_i$ is a bit string with bit entries $e_k$ with labels $k\in\underline{\Lambda}_i$, for $i\in\{0,1\}$ and $\mathbf{e}\in\{\mathbf{u},\mathbf{t}\}$. Third, there is a one-to-one correspondence between the sets $\underline{\Lambda}_i$ and $\underline{\Delta}_i$ via the one-to-one function $g$, i.e $k\in \underline{\Lambda}_i$ iff $g(k)=j\in\underline{\Delta}_i$, for $i\in\{0,1\}$. Fourth, there is a one-to-one correspondence between $\underline{\mathbf{u}}_i$ and $\underline{\mathbf{s}}_i$, and between $\underline{\mathbf{t}}_i$ and $\underline{\mathbf{r}}_i$, i.e. $u_k=s_j$ and $t_k=r_j$ for $j=g(k)$ and $k\in\underline{\Lambda}_i$, and for $i\in\{0,1\}$. Fifth, the sets $\underline{\Lambda}_0$ and $\underline{\Lambda}_1$ do not intersect, hence the sets $\underline{\Delta}_0$ and $\underline{\Delta}_1$ do not intersect either, which implies that $\lvert \underline{\Lambda}_0\cup \underline{\Lambda}_1\rvert=\lvert \underline{\Delta}_0 \cup \underline{\Delta}_1\rvert=\lvert \underline{\Delta}_0 \rvert +\lvert \underline{\Delta}_1\rvert$. Sixth, the bit entries $\tilde{d}_{0,j}=d_{j}\oplus  c$ and $ \tilde{d}_{1,j}=d_{j}\oplus 1\oplus c$ of the respective strings $\tilde{\mathbf{d}}_0$ and $\tilde{\mathbf{d}}_1$ defined in the token scheme $\mathcal{QT}_1$ are different: $\tilde{d}_{0,j}=\tilde{d}_{1,j}\oplus 1$; which are also different in the token scheme $\mathcal{QT}_2$ by setting $d_j=0$ for $j\in[n]$. Seventh, $\underline{\Delta}_i\subseteq \Delta_i$ for $i\in\{0,1\}$, where $ \Delta_i=\{j\in[n]\vert \tilde{d}_{i,j}=s_j\}$, for $i\in\{0,1\}$. From the previous observations, we can associate the parameter $N$ and the $N-$bit string $\mathbf{h}$ in Lemma \ref{lemma0} with $\lvert \underline{\Delta}_0 \rvert +\lvert \underline{\Delta}_1\rvert$ and with a string of bit entries $h_j=\tilde{d}_{0,j}$ for $j\in\underline{\Delta}_0\cup\underline{\Delta}_1$, respectively.
From (\ref{Nj0}) and (\ref{N1049}), we associate the parameters $\gamma_\text{err}$, $\lambda$ and $O$ in Lemma \ref{lemma0} with the parameters $\delta$, $\lambda(\theta,\beta_{\text{PB}})$ and $O(\theta)$ here, respectively. Thus, since $\tilde{d}_{0,j}=\tilde{d}_{1,j}\oplus 1$, for $j\in\underline{\Delta}_0\cup\underline{\Delta}_1$, the set $S_i^{\mathbf{h}}$ in Lemma \ref{lemma0} corresponds to the set  $\underline{\Delta}_i$ here, for $i\in\{0,1\}$. Thus, from (\ref{N1013}) and (\ref{N1023}), we can associate the operator $D_{\mathbf{a},\mathbf{b}}$ in Lemma \ref{lemma0} with the operator $\tilde{\underline{D}}_{x\mathbf{a}\mathbf{b}}$ here. Therefore, the bound (\ref{N1022}) follows from Lemma \ref{lemma0}.

Since as follows from (\ref{N1034}), $P_{\underline{\mathbf{t}}}=\prod_{k\in\Omega_{\text{qub}}}P^k_{\text{PS}}(t_k)$ with $ \frac{1}{2}-\beta_\text{PS}\leq P^k_{\text{PS}}(t_k)\leq \frac{1}{2}+\beta_\text{PS}$ for  $t_k\in\{0,1\}$ and for $k\in\Omega_{\text{qub}}$, we have from (\ref{N1015}) and (\ref{N1023}) that
\begin{equation}
\label{N1024}
\tilde{D}_{x\mathbf{a}\mathbf{b}} \leq \Bigl(\frac{1}{2}+\beta_\text{PS}\Bigr)^{(\lvert \underline{\Delta}_0\rvert +\lvert \underline{\Delta}_1\rvert)}
\tilde{\underline{D}}_{x\mathbf{a}\mathbf{b}},
\end{equation}
where we used that $\lvert \underline{\Lambda}_i\rvert=\lvert \underline{\Delta}_i\rvert$, for $i\in\{0,1\}$. Thus, from (\ref{N1022}) and (\ref{N1024}), we have
\begin{eqnarray}
\label{N1025}
&&\max_{\mathbf{a},\mathbf{b}} \lVert \tilde{D}_{x\mathbf{a}\mathbf{b}}\rVert \nonumber\\
&&\quad\leq e^{-\bigl(\lvert \underline{\Delta}_0\rvert +\lvert \underline{\Delta}_1\rvert \bigr)\Bigl[\frac{\lambda(\theta,\beta_{\text{PB}})}{2}\bigl(1-\frac{\delta}{\lambda(\theta,\beta_{\text{PB}})}\bigr)^2-\text{ln}(1+2\beta_\text{PS})\Bigr]}\times\nonumber\\
&&\qquad\qquad \times 2^{-\bigl(\lvert \underline{\Delta}_0\rvert +\lvert \underline{\Delta}_1\rvert \bigr)}.
\end{eqnarray}

We derive an upper bound on $\lVert \tilde{\phi}_x\rVert$. From (\ref{N1037}) and (\ref{N1016}), we have
\begin{equation}
\label{N1017}
\tilde{\phi}_x=2^{\lvert \Omega_{\text{qub}}\rvert}\bigotimes_{k\in\Omega_{\text{qub}}\setminus\underline{\Lambda}}\bigl(\rho^k\bigr)_{B_k},
\end{equation}
where $\rho^k$ is a qubit density matrix given by
\begin{equation}
\label{N1018}
\rho^k=\sum_{u=0}^1 \sum_{t=0}^1P^k_{\text{PB}}(u)P^k_{\text{PS}}(t)\lvert \phi_{tu}^k\rangle\langle\phi_{tu}^k\rvert,
\end{equation}
for $k\in\Omega_{\text{qub}}\setminus\underline{\Lambda}$. Let $\mu_\pm^k$ be the eigenvalues of $\rho^k$, satisfying $\mu_-^k\leq\mu_+^k$, for $k\in\Omega_{\text{qub}}\setminus\underline{\Lambda}$. It follows from Lemma \ref{eigenvalue}  that
\begin{equation}
\label{N1020}
\mu_+^k\leq \frac{1}{2}\Bigl(1+h(\beta_\text{PS},\beta_\text{PB},\theta)\Bigr),
\end{equation}
where $h(\beta_\text{PS},\beta_\text{PB},\theta)$ is given by (\ref{Al4repeated}), for $k\in\Omega_{\text{qub}}\setminus\underline{\Lambda}$. It follows that
\begin{eqnarray}
\label{N1026}
\lVert \tilde{\phi}_x\rVert&=& 2^{\lvert \Omega_{\text{qub}}\rvert}\prod_{k\in\Omega_{\text{qub}}\setminus\underline{\Lambda}}\mu_+^k\nonumber\\
&\leq& 2^{(\lvert \underline{\Delta}_0\rvert+\lvert \underline{\Delta}_1\rvert)}\bigl(1+h(\beta_\text{PS},\beta_\text{PB},\theta)\bigr)^{\lvert \Omega_{\text{qub}}\rvert -\lvert\underline{\Delta}_0\rvert-\lvert \underline{\Delta}_1\rvert},\nonumber\\
\end{eqnarray}
where in the first line we used (\ref{N1017}) and (\ref{N1018}); and in the second line we used (\ref{N1020}), the fact that $\Omega_{\text{qub}}=\underline{\Lambda}_0\cup\underline{\Lambda}_1\cup\{\Omega_{\text{qub}}\setminus\underline{\Lambda}\}$, that the sets $\underline{\Lambda}_0$ and $\underline{\Lambda}_1$ do not intersect, that $\underline{\Lambda}=\underline{\Lambda}_0\cup\underline{\Lambda}_1$, and that $\lvert \underline{\Delta}_i\rvert =\lvert \underline{\Lambda}_i\rvert$ for $i\in\{0,1\}$.

It follows from (\ref{N1019}), (\ref{N1025}) and (\ref{N1026}) that
\begin{eqnarray}
\label{N1027}
&&\max_{\mathbf{a},\mathbf{b}}\lVert D_{x\mathbf{a}\mathbf{b}}\rVert \nonumber\\
&&\quad\leq e^{-(\lvert \underline{\Delta}_0\rvert+\lvert \underline{\Delta}_1\rvert)\Bigl[\frac{\lambda(\theta,\beta_{\text{PB}})}{2}\bigl(1-\frac{\delta}{\lambda(\theta,\beta_{\text{PB}})}\bigr)^2-\text{ln}(1+2\beta_\text{PS})\Bigr]}\times\nonumber\\
&&\qquad\qquad\times e^{\text{ln}\bigl(1+h(\beta_\text{PS},\beta_\text{PB},\theta)\bigr)\bigl(\lvert \Omega_{\text{qub}}\rvert-\lvert \underline{\Delta}_0\rvert-\lvert \underline{\Delta}_1\rvert\bigr)}\nonumber\\
&&\quad\leq e^{-Nf(\gamma_\text{err},\beta_\text{PS},\beta_\text{PB},\theta,\nu_\text{unf},\gamma_\text{det})},
\end{eqnarray}
where in the second line we used (\ref{Al3repeated}) and
\begin{eqnarray}
\label{Nq12}
&&Nf(\!\gamma_\text{err},\!\beta_\text{PS},\!\beta_\text{PB},\!\theta,\!\nu_\text{unf},\!\gamma_\text{det})\nonumber\\
&&\!\quad\!\leq\!(\!\lvert \underline{\Delta}_0\rvert\!+\!\lvert \underline{\Delta}_1\rvert\!)\!\biggl[\!\frac{\lambda(\theta,\beta_{\text{PB}})}{2}\biggl(\!1\!-\!\frac{\delta}{\lambda(\theta,\beta_{\text{PB}})}\!\biggr)^2\!\!-\!\text{ln}(1\!+\!2\beta_\text{PS})\!\biggr]\nonumber\\
&&\qquad-\text{ln}\bigl(\!1\!+\!h(\beta_\text{PS},\!\beta_\text{PB},\!\theta)\bigr)\!\bigl(\lvert \Omega_{\text{qub}}\rvert-\lvert \underline{\Delta}_0\rvert-\lvert \underline{\Delta}_1\rvert\bigr).
\end{eqnarray}
As we show below, (\ref{Nq12}) holds for all possible values of $x$. Thus, (\ref{N1027}) holds for all possible values of $x$. It follows that 
\begin{equation}
\label{N1029}
\max_{x,\mathbf{a},\mathbf{b}}\lVert D_{x\mathbf{a}\mathbf{b}}\rVert \leq e^{-Nf(\gamma_\text{err},\beta_\text{PS},\beta_\text{PB},\theta,\nu_\text{unf},\gamma_\text{det})},
\end{equation}
with $f(\gamma_\text{err},\beta_\text{PS},\beta_\text{PB},\theta,\nu_\text{unf},\gamma_\text{det})$ given by (\ref{Al3repeated}), which is the claimed bound (\ref{N1030}).

We show that (\ref{Nq12}) holds for all possible values of $x$. First, we note from (\ref{Al1repeated}) that 
\begin{equation}
\label{neweq29}
\frac{\lambda(\theta,\beta_\text{PB})}{2}\Bigl(1-\frac{\delta}{\lambda(\theta,\beta_\text{PB})}\Bigr)^2-\text{ln}(1+2\beta_\text{PS})>0.
\end{equation}
Second, from $\beta_\text{PS}>0$ and from the definition (\ref{Al4repeated}) we have $h(\beta_\text{PS},\beta_\text{PB},\theta)>0$. Thus, we have 
\begin{equation}
\label{neweq30}
\text{ln}\bigl(1+h(\beta_\text{PS},\beta_\text{PB},\theta)\bigr)>0.
\end{equation}
Third, from (\ref{neweq10}) and (\ref{Nq6}), and from the condition $n\geq \gamma_\text{det} N$ for Bob not aborting, we have 
\begin{equation}
\label{neweq31}
 \lvert \underline{\Delta}_0\rvert +\lvert \underline{\Delta}_1\rvert\geq N(\gamma_\text{det} - \nu_\text{unf}).
\end{equation}
Fourth, it follows from (\ref{Al1repeated}) that 
\begin{equation}
\label{neweq32}
\gamma_\text{det}-\nu_\text{unf}>0.
\end{equation}
Thus, since $\lvert \Omega_{\text{qub}}\rvert \leq N$, (\ref{Nq12}) follows from (\ref{neweq29}) -- (\ref{neweq32}) and from the definition of $f(\gamma_\text{err},\beta_\text{PS},\beta_\text{PB},\theta,\nu_\text{unf},\gamma_\text{det})$ given by (\ref{Al3repeated}).


\section{The case of $2^M$ presentation points}
\label{manyapplast}

The proof of Theorem \ref{manypoints} follows straightforwardly from Lemmas \ref{lastrobustM} -- \ref{lastBobM} and from Theorem \ref{lastAliceM} at the end of this section.

The quantum token schemes $\mathcal{QT}_1^M$ and $\mathcal{QT}_2^M$ presented below extend the quantum token schemes $\mathcal{QT}_1$ and $\mathcal{QT}_2$ of Tables \ref{real1} and \ref{real2} to the case of $2^M$ presentation points, for arbitrary integer $M\geq1$. Broadly speaking, the schemes $\mathcal{QT}_1^M$ and $\mathcal{QT}_2^M$ generate the classical inputs and outputs of the schemes $\mathcal{QT}_1$ and $\mathcal{QT}_2$ as subroutines, $M$ times in parallel, with a few differences arising due to the fact that $\mathcal{QT}_1^M$ and $\mathcal{QT}_2^M$ have $2^M$ presentation points instead of two. Similarly to $\mathcal{QT}_1$ and $\mathcal{QT}_2$, $\mathcal{QT}_1^M$ and $\mathcal{QT}_2^M$ can be implemented in practice with the photonic setups of \damian{Fig. \ref{setup}}, respectively.



In the schemes $\mathcal{QT}_1$ and $\mathcal{QT}_2$, there are two presentation points $Q_0$ and $Q_1$, Alice has agents $\mathcal{A},\mathcal{A}_0, \mathcal{A}_1$ and Bob has agents $\mathcal{B},\mathcal{B}_0, \mathcal{B}_1$. From Tables \ref{real1} and \ref{real2}, we see that in $\mathcal{QT}_1$ and $\mathcal{QT}_2$, $\mathcal{B}$ obtains $\bold{t},\bold{u}\in\{0,1\}^N$, $\Omega_\text{noqub}\subseteq[N]$, and $\bold{r},\bold{s}\in\{0,1\}^n$ in the intersection of the causal pasts of the presentation points; $\mathcal{A}$ obtains $\Lambda\subseteq[N]$, $n=\lvert\Lambda\rvert$, $W\in\{0,1\}^\Lambda$, $g:\Lambda\rightarrow [n]$, $\bold{y},\bold{x},\bold{d}\in\{0,1\}^n$, and $b,c,z\in\{0,1\}$ in the intersection of the causal pasts of the presentation points; $\mathcal{B}_i$ obtains $\bold{d}_i,\tilde{\bold{d}}_i\in\{0,1\}^n$ in the causal past of $Q_i$, for $i\in\{0,1\}$; $\mathcal{A}_b$ presents $\bold{x}$ to $\mathcal{B}_b$ in $Q_b$; and $\mathcal{B}_b$ obtains $\bold{x}_b,\bold{r}_b\in\{0,1\}^n$ and $\Delta_b$ in $Q_b$.

On the other hand, in the schemes $\mathcal{QT}_1^M$ and $\mathcal{QT}_2^M$, there are $2^M$ presentation points $Q_i$, Alice has agents $\mathcal{A},\mathcal{A}_i$, and Bob has agents $\mathcal{B},\mathcal{B}_i$, where $i=(i^1,\ldots,i^M)\in\{0,1\}^M$. In these schemes, Alice's and Bob's agents obtain the inputs and outputs of the schemes $\mathcal{QT}_1$ and $\mathcal{QT}_2$ in the corresponding spacetime regions, as mentioned above, in $M$ independent rounds. For the $l$th round, we label the inputs and outputs mentioned above by a superscript $l$. In the schemes $\mathcal{QT}_1^M$ and $\mathcal{QT}_2^M$, the experimental imperfections of Table \ref{tableimp} and the assumptions of Table \ref{tableassu} apply independently to each of the $M$ rounds.

The schemes $\mathcal{QT}_1^M$ and $\mathcal{QT}_2^M$ have a new step  (step 12 of $\mathcal{QT}_1^M$) compared to $\mathcal{QT}_1$ and $\mathcal{QT}_2$, in which $\mathcal{B}_i$ sends a signal to agents $\mathcal{B}_{i'}$ with $Q_{i'}$ in the causal future of $Q_i$ indicating whether a token was presented at $Q_i$ by $\mathcal{A}_i$. This extra step allows us to reduce the proof of unforgeability to the case of spacelike separated presentation points. We note that instant validation is still satisfied, as no extra delays  for token validation due to cross-checking are required.

The schemes $\mathcal{QT}_1^M$ and $\mathcal{QT}_2^M$ are presented precisely below. 

\subsection{Quantum token scheme $\mathcal{QT}_1^M$ for $2^M$ presentation points}

Steps 1 to 10 below are repeated in $M$ independent rounds, labelled by $l\in[M]$. Steps 1 to 9 take place within the intersection
of the causal pasts of the presentation points.

\subsubsection{Preparation stage}

\begin{itemize}
\item[$0.$] Alice and Bob agree on a reference frame, on presentation points $Q_i$ in the agreed frame, for $i\in\{0,1\}^M$, and on parameters $N\in\mathbb{N}$,  $\beta_{\text{PB}}\in\bigl(0,\frac{1}{2}\bigr)$, $\gamma_{\text{det}}\in(0,1)$ and $\gamma_{\text{err}}\in(0,1)$.
\end{itemize}

\subsubsection{Stage I}


\begin{itemize}
\item[$1.$] For $k\in[N]$, $\mathcal{B}$ prepares bits $t_k^l$ and $u_k^l$ with respective probability distributions $P_{\text{PS}}^{k,l}(t_k^l)$ and $P_{\text{PB}}^{k,l}(u_k^l)$, satisfying $\frac{1}{2}-\beta_{\text{X}}\leq P_{\text{X}}^{k,l}(t) \leq \frac{1}{2}+\beta_{\text{X}}$, where $\beta_{\text{X}}\in\bigl(0,\frac{1}{2}\bigr)$ is a small parameter, for $\text{X}\in\{\text{PS},\text{PB}\}$, $t\in\{0,1\}$ and $k\in[N]$. We define $\mathbf{t}^l=(t_1^l,\ldots,t_N^l)$ and $\mathbf{u}^l=(u_1^l,\ldots,u_N^l)$. For $k\in[N]$, $\mathcal{B}$ prepares a quantum system $A_k^l$ in a quantum state $\lvert \psi_k^l\rangle$ and sends it to $\mathcal{A}$ with its label $(k,l)$. $\mathcal{B}$ chooses $(k,l)\in\Omega_{\text{noqub}}^l$ with probability $P_\text{noqub}>0$ or $(k,l)\in\Omega_{\text{qub}}^l$ with probability $1-P_\text{noqub}$. For $(k,l)\in\Omega_{\text{qub}}$, $\lvert\psi_k^l\rangle=\lvert \phi_{t_k^lu_k^l}^{k,l}\rangle$ is a qubit state, where $\langle \phi_{0u}^{k,l}\vert \phi_{1u}^{k,l}\rangle=0$ for $u\in\{0,1\}$, where the qubit orthonormal basis $\mathcal{D}_{u}^{k,l}=\{\lvert\phi_{tu}^{k,l}\rangle\}_{t=0}^1$ is the computational (Hadamard) basis up to an uncertainty angle $\theta$ on the Bloch sphere if $u=0$ ($u=1$).  For $(k,l)\in\Omega_{\text{noqub}}^l$, $\lvert\psi_k^l\rangle=\lvert \Phi_{t_k^lu_k^l}^{l,k}\rangle$ is a quantum state of arbitrary finite Hilbert space dimension greater than two. In photonic implementations, a vacuum or one-photon pulse has label $(k,l)\in\Omega_{\text{qub}}^l$, with a one-photon pulse encoding a qubit state, while a multi-photon pulse has label $(k,l)\in\Omega_{\text{noqub}}^l$ and encodes a quantum state of finite Hilbert space dimension greater than two.

\item[$2.$]  For $k\in[N]$, $\mathcal{A}$ measures $A_k^l$ in the qubit orthonormal basis $\mathcal{D}_{w_k^l}$, for $w_k^l\in\{0,1\}$. Due to losses, $\mathcal{A}$ only successfully measures quantum states $\vert \psi_k^l\rangle$ with labels $(k,l)$ from a proper subset $\Lambda^l$ of $[N]$.
Let $W^l$ be the string of bit entries $w_k^l$ for $(k,l)\in\Lambda^l$ and let $n^l=\lvert\Lambda^l\rvert$. Conditioned on $(k,l)\in\Lambda^l$, the probability 
that $\mathcal{A}$ measures $A_k^l$ in the basis $\mathcal{D}_{w_k^l}$ satisfies $P_{\text{MB}}(w_k^l)=\frac{1}{2}$, for $w_k^l\in\{0,1\}$ and $k\in[N]$. $\mathcal{A}$ reports to $\mathcal{B}$ the set $\Lambda^l$ with its label $l$. $\mathcal{B}$ does not abort if and only if $n^l\geq \gamma_\text{det} N$.

\item[$3.$] $\mathcal{A}$ chooses a one-to-one function $g^l: \Lambda^l\rightarrow [n]$, for example the numerical ordering, and sends it to $\mathcal{B}$ with its label $l$. Let $y_j^l\in\{0,1\}$ indicate the basis $\mathcal{D}_{y_j^l}$ on which the quantum state $\vert \psi_k^l\rangle$ is measured by $\mathcal{A}$ and let $x_j^l\in\{0,1\}$ be the measurement outcome, where $j=g^l(k)$, for $k\in\Lambda^l$ and $j\in[n]$. Let $\mathbf{y}^l=(y_1^l,\ldots,y_n^l)\in\{0,1\}^{n}$ and $\mathbf{x}^l=(x_1^l,\ldots,x_n^l)\in\{0,1\}^{n}$ denote the strings of Alice's measurement bases and outcomes, respectively.

\item[$4.$] $\mathcal{A}$ sends $\mathbf{x}^l$ to $\mathcal{A}_i$ with its label $l$, for $i\in\{0,1\}^M$.

\item[$5.$] $\mathcal{A}$ chooses a bit $z^l\in\{0,1\}$ with probability $P^l_{\text{E}}(z^l)$ that satisfies $\frac{1}{2}-\beta_{\text{E}}\leq P^l_{\text{E}}(z^l) \leq \frac{1}{2}+\beta_{\text{E}}$, for $z^l\in\{0,1\}$, and for a small parameter $\beta_{\text{E}}\in\bigl(0,\frac{1}{2}\bigr)$. $\mathcal{A}$ computes the string $\mathbf{d}^l\in\{0,1\}^{n}$ with bit entries $d_j^l=y_j^l\oplus z^l$, for $j\in[n]$. $\mathcal{A}$ sends $\mathbf{d}^l$ to $\mathcal{B}$ with its label $l$.

\item[$6.$] For $i=(i^1,\ldots,i^M)\in\{0,1\}^M$, $\mathcal{B}$ sends $\mathbf{d}^l$ to $\mathcal{B}_i$ with its label $l$, and $\mathcal{B}_i$ computes the string $\mathbf{d}_{i^l}^l\in\{0,1\}^{n}$ with bit entries $d_{i^l,j}^l=d_j^l\oplus i^l$, for $j\in[n]$.

\item[$7.$] $\mathcal{B}$ uses $\mathbf{t}^l,\mathbf{u}^l,\Lambda^l$ and $g^l$ to compute the strings $\mathbf{s}^l,\mathbf{r}^l\in\{0,1\}^{n}$, as follows. We define $r_j^l=t_k^l$, and $s_j^l=u_k^l$, where $j=g^l(k)$, for $j\in[n]$ and $k\in\Lambda^l$. We define $\mathbf{r}^l$ and $\mathbf{s}^l$ as the strings with bit entries $r_j^l$ and $s_j^l$, for $j\in[n]$, respectively. For $\mathcal{B}$ sends $\mathbf{s}^l$ and $\mathbf{r}^l$ to $\mathcal{B}_i$ with its label $l$, for $i\in\{0,1\}^M$.
\end{itemize}

\subsubsection{Stage II}


\begin{itemize}
\item[$8.$] $\mathcal{A}$ chooses the $l$th entry $b^l\in\{0,1\}$ for the bit string $b = (b^1,\ldots,b^M)\in\{0,1\}^M$ that labels the presentation point $Q_b$ where to present the token.
$\mathcal{A}$ computes the bit $c^l=b^l\oplus z^l$ and sends it to $\mathcal{B}$ with its label $l$.

\item[$9.$] $\mathcal{B}$ sends $c^l$ with its label $l$ to $\mathcal{B}_i$, for $i\in\{0,1\}^M$.

\item[$10.$] For $i\in\{0,1\}^M$, in the causal past of $Q_i$, $\mathcal{B}_i$ computes the string $\tilde{\mathbf{d}}_{i^l}^l\in\{0,1\}^{n}$ with bit entries $\tilde{d}_{i^l,j}^l=d_{i^l,j}^l\oplus c^l$, for $j\in[n]$.

\item[$11.$] $\mathcal{A}$ sends a signal to $\mathcal{A}_b$ indicating to present the token at $Q_b$, and $\mathcal{A}_b$ presents the token $\mathbf{x}=(\mathbf{x}^1,\ldots,\mathbf{x}^M)$ to $\mathcal{B}_b$ in $Q_b$.

\item[$12.$] For all $i\in\{0,1\}^M$, if $\mathcal{B}_i$ receives a token from $\mathcal{A}_i$ at $Q_i$, $\mathcal{B}_i$ sends a signal to $\mathcal{B}_{i'}$ indicating so, for all $i'\in\{0,1\}^M$ such that $Q_{i'}$ is in the causal future of $Q_i$.

\item[$13.$] $\mathcal{B}_b$ validates the token $\mathbf{x}$ received in $Q_b$ if two conditions hold: 1) $\mathcal{B}_b$ does not receive signals from Bob's agent $\mathcal{B}_i$ indicating that a token has been presented by Alice at $Q_i$, for any $i\in\{0,1\}^M$ such that $Q_i$ is in the causal past of $Q_b$; and 2) for all $l\in[M]$, the Hamming distance between the strings $\mathbf{x}_{b^l}^l$ and $\mathbf{r}_{b^l}^l$ satisfies $d(\mathbf{x}_{b^l}^l,\mathbf{r}_{b^l}^l)\leq \lvert \Delta_{b^l}^l\rvert \gamma_\text{err}$, where $\Delta_v^l=\{j\in [n]\vert \tilde{d}_{v,j}^l=s_j^l\}$, and where $\mathbf{a}_{v}^l$ is the restriction of a string $\mathbf{a}^l\in\{\mathbf{x}^l,\mathbf{r}^l\}$ to entries $a_j^l$ with $j\in\Delta_v^l$, for $v\in\{0,1\}$.

\end{itemize}

\subsection{Quantum token scheme $\mathcal{QT}_2^M$ for $2^M$ presentation points}

Steps 1 to 7 below are repeated in $M$ independent rounds, labelled by $l\in[M]$. Steps 1 to 6 take place within the intersection of the causal pasts of the presentation points.

\subsubsection{Preparation stage}

\begin{itemize}
\item[$0.$] As step 0 of $\mathcal{QT}_1^M$.
\end{itemize}

\subsubsection{Stage I}
\begin{itemize}
\item[$1.$] As step 1 of $\mathcal{QT}_1^{M}$.

\item[$2.$] The step 2 of $\mathcal{QT}_1^M$ is replaced by the following. $\mathcal{A}$ chooses a bit $z^l$ with probability $P_{\text{E}}^l(z^l)$ satisfying $\frac{1}{2}-\beta_{\text{E}}\leq P_{\text{E}}^l(z^l)\leq \frac{1}{2}+\beta_{\text{E}}$, for $z^l\in\{0,1\}$, and for a small parameter $\beta_{\text{E}}\in\bigl(0,\frac{1}{2}\bigr)$. $\mathcal{A}$ measures $A_k^l$ in the qubit orthonormal basis $\mathcal{D}_{z^l}$, for $k\in[N]$. Due to losses, $\mathcal{A}$ only successfully measures quantum states $\vert \psi_k^l\rangle$ with labels $(l,k)$ from a proper subset $\Lambda^l$ of $[N]$. $\mathcal{A}$ reports to $\mathcal{B}$ the set $\Lambda^l$ with its label $l$. Let $n^l=\lvert\Lambda^l\rvert$. $\mathcal{B}$ does not abort if and only if $n^l\geq \gamma_\text{det} N$.


\item[$3.$] As step 3 of $\mathcal{QT}_1^M$. 
The string $\mathbf{y}^l\in\{0,1\}^{n}$ of Alice's measurement bases has bit entries $y^l_j=z^l$, for $j\in[n]$. 

\item[$4.$] As step 4 of $\mathcal{QT}_1^M$. The steps 5 and 6 of $\mathcal{QT}_1^M$ are discarded. 

\item[$5.$] As step 7 of $\mathcal{QT}_1^M$.

\subsubsection{Stage II}

\item[$6.$] As steps 8 and 9 of $\mathcal{QT}_1^M$.

\item[$7.$] The step 10 of $\mathcal{QT}_1^M$ is replaced by the following. 
For $i=(i^1,\ldots,i^M)\in\{0,1\}^M$, in the causal past of $Q_i$, $\mathcal{B}_i$ computes the string $\tilde{\mathbf{d}}_{i^l}^l\in\{0,1\}^{n}$ with bit entries $\tilde{d}_{{i^l},j}^l= i^l\oplus c^l$, for $j\in[n]$.

\item[$8.$] As steps 11, 12 and 13 of $\mathcal{QT}_1^M$.

\end{itemize}

\subsection{Comments}

We note that steps 1 to 11 and 1 to 7 of the token schemes $\mathcal{QT}_1^M$ and $\mathcal{QT}_2^M$ are straightforward extensions of the corresponding steps in the token schemes $\mathcal{QT}_1$ and $\mathcal{QT}_2$, respectively. As we discussed above, this basically comprises applying the corresponding steps of $\mathcal{QT}_1$ and $\mathcal{QT}_2$ in $M$ parallel and independent rounds. However, step 12 of $\mathcal{QT}_1^M$ is a new step, and steps 13 and 8 of $\mathcal{QT}_1^M$ and $\mathcal{QT}_2^M$ modify steps 12 and 8 of $\mathcal{QT}_1$ and $\mathcal{QT}_2$, respectively, to account for the new step.

We note from steps 12 and 13 of $\mathcal{QT}_1^M$, and from step 8 of $\mathcal{QT}_2^M$, that a token received by Bob's agent $\mathcal{B}_b$ from Alice's agent $\mathcal{A}_b$ at a presentation point $Q_b$ can be validated by $\mathcal{B}_b$ nearly instantly at $Q_b$. In particular, $\mathcal{B}_b$ does not need to wait for any signals coming from agents $\mathcal{B}_i$ who have possibly received tokens from Alice's agents at presentation points $Q_i$ that are not in the causal past of $Q_b$.

For $M>1$, steps 12 and 13 of $\mathcal{QT}_1^M$, and step 8 of $\mathcal{QT}_2^M$, allow us to guarantee unforgeability, as discussed 
below. In the case $M=1$, steps 12 and 13 of $\mathcal{QT}_1^M$, and step 8 of $\mathcal{QT}_2^M$, can simply be replaced by step 12 of $\mathcal{QT}_1$, and by step 8 of $\mathcal{QT}_2$, respectively.

Every observation made previously about the token schemes $\mathcal{QT}_1$ and $\mathcal{QT}_2$ also applies to the token schemes $\mathcal{QT}_1^M$ and $\mathcal{QT}_2^M$. In particular, the schemes $\mathcal{QT}_1^M$ and $\mathcal{QT}_2^M$ also allow for the experimental imperfections of Table \ref{tableimp} and make the assumptions of Table \ref{tableassu}, for the $l$th round and for $l\in[M]$.
Stage I includes the quantum communication, which can take place between adjacent laboratories, and must be implemented within the intersection of the
causal pasts of all the presentation points. This stage can take an arbitrarily long time and can be completed arbitrarily in the past of the presentation
points, which is very helpful for practical implementations. Stage II comprises only classical processing and communication, and must usually be completed within a very short time. 
We note that Alice chooses her presentation point in stage
II, meaning in particular that it can take place after her
quantum measurements have been completed, which gives Alice
great flexibility in spacetime to choose her presentation
point. The token schemes $\mathcal{QT}_1^M$ and $\mathcal{QT}_2^M$  can be modified in various ways, as discussed previously for the $\mathcal{QT}_1$ and $\mathcal{QT}_2$  schemes.

\subsection{Robustness, correctness, privacy and unforgeability}

As discussed for the token schemes $\mathcal{QT}_1$ and $\mathcal{QT}_2$, in the token schemes $\mathcal{QT}_1^M$ and $\mathcal{QT}_2^M$ we define $P_{\text{det}}$ as the probability that a quantum state $\lvert \psi_k^l\rangle$ transmitted by Bob is reported by Alice as being successfully measured, with label $(l,k)\in\Lambda^l$, for $k\in[N]$ and $l\in[M]$. 
We define $E$ as the probability that Alice obtains a wrong measurement outcome when she measures a quantum state $\lvert\psi_k^l\rangle$ in the basis of preparation by Bob; if the error rates $E_{tu}$ are different for different prepared states, labelled by $t$, and for different measurement bases, labelled by $u$, we simply take $E=\max_{t,u}\{E_{tu}\}$.

The robustness, correctness, privacy and unforgeability
of $\mathcal{QT}_1^M$ and $\mathcal{QT}_2^M$ are stated by the following lemmas and theorem.
These lemmas and theorem consider parameters $\gamma_\text{det},\gamma_\text{err}\in(0,1)$, allow for the experimental imperfections of Table \ref{tableimp} and make the assumptions of Table \ref{tableassu}, for each of the $M$ rounds labelled by $l\in[M]$, as discussed above.

\begin{lemma}
\label{lastrobustM}
If
\begin{equation}
				\label{lastrob1M}
		0<\gamma_\text{det} <P_{\text{det}},
	\end{equation}
then $\mathcal{QT}_1^M$ and $\mathcal{QT}_2^M$ are $\epsilon_{\text{rob}}^{M}-$robust with
\begin{equation}
			\label{lastrob2M}
						\epsilon_{\text{rob}}^{M}=1-(1-\epsilon_{\text{rob}})^M\leq M\epsilon_{\text{rob}},
\end{equation}
where
\begin{equation}
\label{lastlastrob}
\epsilon_{\text{rob}}=e^{-\frac{P_{\text{det}}N}{2}\bigl(1-\frac{\gamma_\text{det} }{P_{\text{det}}}\bigr)^2}.
\end{equation}
\end{lemma}

\begin{proof}
Suppose that (\ref{lastrob1M}) holds and that Alice and Bob follow the scheme $\mathcal{QT}_a^M$ honestly, for $a\in\{0,1\}$. From Lemma \ref{robust1}, the probability $P_\text{abort}^1$ that Bob aborts in the round label by $l=1$ satisfies $P_\text{abort}^1\leq \epsilon_\text{rob}$, with $\epsilon_\text{rob}$ given by (\ref{lastlastrob}). Since steps 1 to 10 of $\mathcal{QT}_1^M$ and steps 1 to 7 of $\mathcal{QT}_2^M$ are implemented in $M$ independent rounds, we also have from Lemma \ref{robust1} that  the probability $P_\text{abort}^l$ that Bob aborts in the $l$th round, given that he does not abort in the rounds $1,2,\ldots,l-1$, satisfies $P_\text{abort}^l\leq \epsilon_\text{rob}$, with $\epsilon_\text{rob}$ given by (\ref{lastlastrob}), for $l\in\{2,\ldots,M\}$. Thus, the probability $P_\text{abort}$ that Bob aborts in the scheme satisfies
\begin{equation}
P_\text{abort}\leq 1- (1-\epsilon_\text{rob})^M.
\end{equation}
Thus, the schemes $\mathcal{QT}_1^M$ and $\mathcal{QT}_2^M$ are $\epsilon_{\text{rob}}^{M}-$robust with $\epsilon_{\text{rob}}^{M}$ given by (\ref{lastrob2M}), as claimed. The inequality in (\ref{lastrob2M}) follows from Bernoulli's inequality.
\end{proof}

\begin{lemma}
\label{lastcorrectM}
If
\begin{eqnarray}
				\label{lastcor1M}
				0&<&\frac{\gamma_\text{err}}{2}<E<\gamma_\text{err},\nonumber\\
		0&<&\nu_\text{cor}<\frac{P_{\text{det}}(1-2\beta_\text{PB})}{2},
	\end{eqnarray}
then $\mathcal{QT}_1^M$ and $\mathcal{QT}_2^M$ are $\epsilon_{\text{cor}}^M-$correct with
\begin{equation}
			\label{lastcor2M}
			\epsilon_{\text{cor}}^{M}=1-(1-\epsilon_{\text{cor}})^M\leq M\epsilon_{\text{cor}},
\end{equation}
where
	\begin{equation}
			\label{lastlastcor}		
						\epsilon_{\text{cor}}=e^{-\frac{P_{\text{det}}(1-2\beta_\text{PB})N}{4}\bigl(1-\frac{2\nu_\text{cor}}{P_{\text{det}}(1-2\beta_\text{PB})}\bigr)^2}+e^{-\frac{E\nu_\text{cor} N}{3}\bigl(\frac{\gamma_\text{err}}{E}-1\bigr)^2}.
			\end{equation}
\end{lemma}

\begin{proof}
Suppose that (\ref{lastcor1M}) holds and that Alice and Bob follow the scheme $\mathcal{QT}_a^M$ honestly, for $a\in\{0,1\}$. We see from step 13 of $\mathcal{QT}_1^M$ and step 8 of $\mathcal{QT}_2^M$ that $\mathcal{B}_b$ validates Alice's token $\mathbf{x}=(\mathbf{x}^1,\ldots,\mathbf{x}^M)$ at the presentation point $Q_b$ if the condition $d(\mathbf{x}_{b^l}^l,\mathbf{r}_{b^l}^l)\leq \lvert \Delta_{b^l}^l\rvert \gamma_\text{err}$ is satisfied for all $l\in[M]$. Since steps 1 to 10 of $\mathcal{QT}_1^M$ and steps 1 to 7 of $\mathcal{QT}_2^M$ are implemented in $M$ independent rounds, we see that the probability that each of these conditions is satisfied is independent of whether the other conditions are satisfied. We see that steps 1 to 10 (1 to 7) of the $l$th round in $\mathcal{QT}_1^M$ ($\mathcal{QT}_2^M$) and the $l$th condition for token validation in step 13 (8) of $\mathcal{QT}_1^M$ ($\mathcal{QT}_2^M$) are equivalent to the corresponding steps and the condition for token validation in $\mathcal{QT}_1$ ($\mathcal{QT}_2$), for $l\in[M]$. Thus, we have from Lemma \ref{correct1} that the probability $P_\text{fail}^l$ that the $l$th condition for token validation in $\mathcal{QT}_1^M$ ($\mathcal{QT}_2^M$) is not passed satisfies $P_\text{fail}^l\leq \epsilon_\text{cor}$, with $\epsilon_\text{cor}$ given by (\ref{lastlastcor}), for $l\in[M]$. 
Thus, the probability $P_\text{fail}$ that the token $\textbf{x}$ is not validated by $\mathcal{B}_b$ in $Q_b$ in either scheme $\mathcal{QT}_1^M$ or $\mathcal{QT}_2^M$ satisfies
\begin{equation}
P_\text{fail}\leq 1- (1-\epsilon_\text{cor})^M.
\end{equation}
Thus, the schemes $\mathcal{QT}_1^M$ and $\mathcal{QT}_2^M$ are $\epsilon_{\text{cor}}^{M}-$correct with $\epsilon_{\text{cor}}^{M}$ given by (\ref{lastcor2M}), as claimed. The inequality in (\ref{lastcor2M}) follows from Bernoulli's inequality.
\end{proof}

\begin{lemma}
\label{lastBobM}
$\mathcal{QT}_1^M$ and $\mathcal{QT}_2^M$ are $\epsilon_{\text{priv}}^{M}-$private with
\begin{equation}
\label{lastBo1M}
\epsilon_{\text{priv}}^{M}=\frac{1}{2^M}\bigl[(1+2\epsilon_{\text{priv}})^M-1\bigr],
\end{equation}
with
\begin{equation}
\label{lastlastBo1M}
\epsilon_{\text{priv}}=\beta_\text{E}.
\end{equation}
\end{lemma}

\begin{proof}
Suppose that Alice follows the scheme $\mathcal{QT}_a^M$ honestly, for $a\in\{0,1\}$. From assumption C (see Table \ref{tableassu}),
 the set $\Lambda^l$ of labels transmitted to $\mathcal{B}$ in step 2 of $\mathcal{QT}_1^M$
and $\mathcal{QT}_2^M$ in the $l$th round gives $\mathcal{B}$ no information about the string $W^l$ and the bit $z^l$, for $l\in[M]$. Furthermore, from assumption E (see Table \ref{tableassu}),
 $\mathcal{B}$ cannot use degrees of freedom
not previously agreed for the transmission of the quantum states to affect, or obtain information about, the statistics of the quantum measurement devices of $\mathcal{A}$. Moreover, in our setting, we assume that Alice's laboratories are secure and that communication among Alice's agents is made through secure and authenticated classical channels. It follows from these assumptions that the only way in which Bob can obtain information about Alice's bit string $b=(b^1,\ldots,b^M)$ before she presents the token is via the bit messages $c^l=z^l\oplus b^l$, for $l\in[M]$. Since Alice prepares the bits $z^l$ in independent rounds, for $l\in[M]$, the probability that Bob guesses Alice's bit string $b$ is given by
\begin{equation}
\label{eqeq1}
P_{\text{Bob}}=\prod_{l=1}^MP_{\text{Bob}}^{l},
\end{equation}
where $P_{\text{Bob}}^{l}$ is the probability that Bob guesses Alice's bit $b^l$, for $l\in[M]$.

We see that the steps 1 to 10 (1 to 7) of the $l$th round in $\mathcal{QT}_1^M$ ($\mathcal{QT}_2^M$) are equivalent to the corresponding steps in $\mathcal{QT}_1$ ($\mathcal{QT}_2$), for $l\in[M]$. Thus, from Lemma \ref{Bob1}, we have
\begin{equation}
\label{eqeq2}
P_{\text{Bob}}^{l}\leq \frac{1}{2}+\epsilon_\text{priv},
\end{equation}
for $l\in[M]$, with $\epsilon_\text{priv}$ given by (\ref{lastlastBo1M}). From (\ref{eqeq1}) and (\ref{eqeq2}), we have that
\begin{eqnarray}
\label{eqeq3}
P_{\text{Bob}}&\leq&\Bigl(\frac{1}{2}+\epsilon_\text{priv}\Bigr)^M\nonumber\\
&=&\frac{1}{2^M}+\epsilon_\text{priv}^M,
\end{eqnarray}
with $\epsilon_\text{priv}^M$ given by (\ref{lastBo1M}). Thus, the schemes $\mathcal{QT}_1^M$ and $\mathcal{QT}_2^M$ are $\epsilon_{\text{priv}}^{M}-$private with $\epsilon_\text{priv}^M$ given by (\ref{lastBo1M}), as claimed.
\end{proof}

\begin{theorem}
\label{lastAliceM}

Consider the constraints
\begin{eqnarray}
\label{lastAl1M}
0&<&\gamma_\text{err}<\lambda(\theta,\beta_{\text{PB}}),\nonumber\\
0&\!<\!&\!P_\text{noqub}\!<\!\nu_\text{unf}\!<\!\min\biggl\{\!\!2P_{\text{noqub}},\!\gamma_\text{det}\biggl(\!1\!-\!\frac{\gamma_\text{err}}{\lambda(\theta,\beta_{\text{PB}})}\!\biggr)\!\!\biggr\},\nonumber\\
0&<&\beta_\text{PS}<\frac{1}{2}\biggl[e^{\frac{\lambda(\theta,\beta_{\text{PB}})}{2}\bigl(1-\frac{\delta}{\lambda(\theta,\beta_{\text{PB}})}\bigr)^2}-1\biggr].
	\end{eqnarray}
	We define the function
\begin{eqnarray}
\label{lastAl3M}	
	&&f(\!\gamma_\text{err},\!\beta_\text{PS},\!\beta_\text{PB},\!\theta,\!\nu_\text{unf},\!\gamma_\text{det})\nonumber\\
	&&\!\!\quad=(\gamma_\text{det}\!-\!\nu_\text{unf})\biggl[\!\frac{\lambda(\theta,\!\beta_{\text{PB}}) }{2}\!\biggl(\!1\!-\!\frac{\delta}{\lambda(\theta,\beta_{\text{PB}})}\!\biggr)^2\!\!-\!\ln(\!1\!+\!2\beta_\text{PS})\!\biggr]\nonumber\\
	&&\qquad-\bigl(1\!-\!(\gamma_\text{det}\!-\!\nu_\text{unf})\bigr)\ln\bigl[1+h(\beta_\text{PS},\beta_\text{PB},\theta)\bigr],
	\end{eqnarray}	
	 where 
	 \begin{eqnarray}
\label{lastAl4M}
h(\beta_\text{PS},\beta_\text{PB},\theta)&=&2\beta_\text{PS}\sqrt{\frac{1}{2}\!+\!2\beta_\text{PB}^2\!+\!\Bigl(\frac{1}{2}\!-\!2\beta_\text{PB}^2\Bigr)\sin(2\theta)},\nonumber\\
\delta&=&\frac{\gamma_\text{det}\gamma_\text{err}}{\gamma_\text{det}-\nu_\text{unf}}.
\end{eqnarray}
Let $L\leq 2^M$ be the number of pair-wise spacelike separated presentation points among the $2^M$ presentation points and let $C=\frac{L(L-1)}{2}$  be the number of pairs of spacelike separated presentation points, if $L\geq 2$, and $C=0$ if $L=0$.
For any $M\geq 1$, there exist parameters satisfying the constraints (\ref{lastAl1M}), for which $f(\!\gamma_\text{err},\!\beta_\text{PS},\!\beta_\text{PB},\!\theta,\!\nu_\text{unf},\!\gamma_\text{det})>0$. For these parameters, $\mathcal{QT}_1^M$ and $\mathcal{QT}_2^M$ are $\epsilon_{\text{unf}}^M-$unforgeable with
\begin{equation}
\label{lastAl5M}
\epsilon_{\text{unf}}^M=C\epsilon_{\text{unf}},
\end{equation}
with
\begin{equation}
\label{lastlastAl5M}
\epsilon_{\text{unf}}\!=\!e^{-\frac{P_\text{noqub}N}{3}\bigl(\frac{\nu_\text{unf}}{P_\text{noqub}}-1\bigr)^2}\! +e^{-Nf(\gamma_\text{err},\beta_\text{PS},\beta_\text{PB},\theta,\nu_\text{unf},\gamma_\text{det})}.	
\end{equation}
	 \end{theorem}

\begin{proof}
From Theorem \ref{Alice1}, there exist parameters satisfying the constraints (\ref{lastAl1M}), for which $f(\!\gamma_\text{err},\!\beta_\text{PS},\!\beta_\text{PB},\!\theta,\!\nu_\text{unf},\!\gamma_\text{det})>0$. This holds for arbtrary $M\geq1$ because the constraints (\ref{lastAl1M}) and the function $f(\!\gamma_\text{err},\!\beta_\text{PS},\!\beta_\text{PB},\!\theta,\!\nu_\text{unf},\!\gamma_\text{det})$ are independent of $M$.

Suppose that the constraints (\ref{lastAl1M}) hold and that $f(\!\gamma_\text{err},\!\beta_\text{PS},\!\beta_\text{PB},\!\theta,\!\nu_\text{unf},\!\gamma_\text{det})>0$. Suppose that Bob follows the scheme $\mathcal{QT}_a^M$ honestly and Alice follows an arbitrary cheating strategy $\mathcal{S}$, for $a\in\{0,1\}$. From step 12 (8) of $\mathcal{QT}_1^M$ ($\mathcal{QT}_2^M$), Alice cannot succeed in making Bob validate tokens at timelike separated presentation points. For this reason, we consider without loss of generality that Alice tries to make Bob validate tokens at spacelike separated presentation points.

Let $L\leq 2^M$ be the number of spacelike separated presentation points.
 Alice's general cheating strategy $\mathcal{S}$ comprises using the received classical information and quantum sates from Bob to output classical data to give Bob as required by the scheme $\mathcal{QT}_1^M$ ($\mathcal{QT}_2^M$) in steps 1 to 8 (1 to 6), and to obtain a token to give Bob at each of the $L$ spacelike separated presentation points. We note that in our schemes there are no penalties for Alice if Bob catches her cheating. Thus, it does not affect Alice to present tokens at all spacelike separated presentation points, even if this can in principle increase the probability that Bob catches her cheating. Moreover, by presenting tokens at all spacelike separated presentation points, Alice has a greater probability to make Bob validate tokens at any two or more different presentation points. For example, if by giving tokens at $K<L$ spacelike separated presentation points Alice can make Bob validate tokens at any two or more different presentation points with probability $P$, Alice can additionally give a random token at another spacelike separated presentation point and in this way increase the probability to some value $P'>P$.

Let $\mathcal{P}_\text{spacelike}$ be the set of labels $i=(i^1,\ldots,i^M)\in\{0,1\}^M$ for the spacetime presentation points $Q_i$ that are spacelike separated. Let $v,w\in\mathcal{P}_\text{spacelike}$ with $v\neq w$. Let $\mathbf{a}=(\mathbf{a}^1,\ldots,\mathbf{a}^M)$ and $\mathbf{b}=(\mathbf{b}^1,\ldots,\mathbf{b}^M)$ be the tokens that Alice gives Bob at $Q_v$ and $Q_w$, respectively. Let $P^{\mathcal{S}}_{vw}$ be the probability that Bob validates the token $\mathbf{a}$ at $Q_v$ and the token $\mathbf{b}$ at $Q_w$. Let $P^\mathcal{S}$ be the probability that Bob validates tokens at any two or more different presentation points. We have
\begin{equation}
\label{oo1}
P^\mathcal{S}\leq\sum_{\substack{v,w\in\mathcal{P}_\text{spacelike}\\ v\neq w}}P^{\mathcal{S}}_{vw}.
\end{equation}
We show below that
\begin{equation}
 \label{oo2}
P^{\mathcal{S}}_{vw}\leq \epsilon_\text{unf},
\end{equation}
for any $v,w\in\mathcal{P}_\text{spacelike}$ with $v\neq w$, and for any cheating strategy $\mathcal{S}$ by Alice, where $\epsilon_\text{unf}$ is given by (\ref{lastlastAl5M}). By noticing that by definition, $L=\lvert \mathcal{P}_\text{spacelike}\rvert$ is the number of spacelike separated presentation points and $C$ is the number of pairs of spacelike separated presentation points, it follows from (\ref{oo1}) and (\ref{oo2}) that $\mathcal{QT}_1^M$ and $\mathcal{QT}_2^M$ are $\epsilon_\text{unf}^M$ unforgeable, with $\epsilon_\text{unf}^M$ given by (\ref{lastAl5M}), as claimed.

We show (\ref{oo2}). Let $v,w\in\mathcal{P}_\text{spacelike}$ with $v\neq w$. Let $v = (v^1,\ldots,v^M)$ and $w=(w^1,\ldots,w^M)$, where $v^l,w^l\in\{0,1\}$, for $l\in[M]$. Since $v\neq w$, there exists $l'\in[M]$ such that 
 \begin{equation}
 \label{eqeqeq0}
 v^{l'}=w^{l'}\oplus 1.
 \end{equation}
 Thus, without loss of generality, let
  \begin{equation}
 \label{eqeqeq3}
 v^{l'}=0\qquad \text{and}\qquad w^{l'}= 1.
 \end{equation}
 Bob validating the token $\mathbf{a}$ at $Q_v$ and the token $\mathbf{b}$ at $Q_w$ requires satisfaction of the conditions $d(\mathbf{a}_{v^l}^l,\mathbf{r}_{v^l}^l)\leq \lvert \Delta^l_{v^l}\rvert\gamma_\text{err}$ at $Q_v$ and $d(\mathbf{b}_{w^l}^l,\mathbf{r}_{w^l}^l)\leq \lvert \Delta^l_{w^l}\rvert\gamma_\text{err}$ at $Q_w$, for all $l\in[M]$. Thus, it requires in particular satisfaction of the conditions 
 \begin{eqnarray}
 \label{eqeqeq1}
 d(\mathbf{a}_{0}^{l'},\mathbf{r}_{0}^{l'})&\leq& \lvert \Delta^{l'}_{0}\rvert\gamma_\text{err},\nonumber\\
 d(\mathbf{b}_{1}^{l'},\mathbf{r}_{1}^{l'})&\leq& \lvert \Delta^{l'}_{1}\rvert\gamma_\text{err},
 \end{eqnarray}
at $Q_v$ and $Q_w$, respectively, where we used (\ref{eqeqeq3}).

Since Bob follows the scheme honestly, he follows the steps 1 to 10 (1 to 7) of the $l'$th round in $\mathcal{QT}_1^M$ ($\mathcal{QT}_2^M$) independently of rounds with label $l\neq l'$. we see that Bob's steps 1 to 10 (1 to 7) of the $l'$th round in $\mathcal{QT}_1^M$ ($\mathcal{QT}_2^M$) and his $l'$th conditions for token validation at $Q_v$ and $Q_w$ given by (\ref{eqeqeq1}) are equivalent to Bob's corresponding steps and conditions for token validation at the two presentation points in $\mathcal{QT}_1$ ($\mathcal{QT}_2$). Thus, from Theorem \ref{Alice1}, the probability that both conditions (\ref{eqeqeq1}) are satisfied is upper bounded by $\epsilon_\text{unf}$, with $\epsilon_\text{unf}$ given by (\ref{lastlastAl5M}). It follows that
 \begin{equation}
 \label{eqeqeq2}
P^{\mathcal{S}}_{vw}\leq \epsilon_\text{unf},
\end{equation}
for any $v,w\in\mathcal{P}_\text{spacelike}$ with $v\neq w$, and for any cheating strategy $\mathcal{S}$ by Alice.




\end{proof}

We note that Lemmas \ref{lastrobustM}, \ref{lastcorrectM} and \ref{lastBobM} reduce to Lemmas \ref{robust1}, \ref{correct1} and \ref{Bob1} in the case $M=1$, respectively. Similarly, Theorem \ref{lastAliceM} reduces to Theorem \ref{Alice1} in the case $M=1$ if the presentation points are spacelike separated. This follows straightforwardly from the fact that $\mathcal{QT}_a^M$ reduces to $\mathcal{QT}_a$ for the case $M=1$, for $a\in\{1,2\}$, except for steps 12 and 13 of $\mathcal{QT}_1^M$ and step 8 of $\mathcal{QT}_2^M$, which as mentioned above can simply be replaced by step 12 of $\mathcal{QT}_1$ and step 8 of $\mathcal{QT}_2$ in this case, respectively. In the case $M=1$ with timelike separated presentation points, differently to Theorem \ref{Alice1}, Theorem \ref{lastAliceM} sates that the probability that Bob validates tokens at both presentation points is zero. This is due to the extra step 12 (8) in $\mathcal{QT}_1^M$ ($\mathcal{QT}_2^M$). In any case, Theorem \ref{lastAliceM} is consistent with Theorem \ref{Alice1} in the case $M=1$, if the presentation points are timelike or spacelike separated.

\end{document}